\documentclass{amsart}

\usepackage{macros, amsaddr}

\allowdisplaybreaks

\def\scalar#1{\textcolor{green!65!black}{#1}}
\def\fermi#1{\textcolor{red!65!black}{#1}}

\def\del{\partial}
\def\delbar{\bar\partial}
\newcommand{\defterm}[1]{\textbf{\emph{#1}}}
\def\hCS{{\rm hCS}}

\def\psBV{presymplectic BV}

\newcommand*\circleds[1]{\tikz[baseline=(char.base), transform canvas={scale=0.6}]{
            \node[shape=circle,draw,solid,inner sep=2pt] (char) {#1};}}
\newcommand*\circled[1]{\tikz[baseline=(char.base)]{
            \node[shape=circle,draw,solid,inner sep=2pt] (char) {#1};}}            

\begin{document}

\title[Constraints in the BV formalism]{Constraints in the BV formalism: \\ six-dimensional supersymmetry and its twists}

\author{Ingmar Saberi}
\address{Mathematisches Institut der Universit\"at Heidelberg \\ Im Neuenheimer Feld 205 \\ 69120 Heidelberg \\ Deutschland}
\email{saberi@mathi.uni-heidelberg.de}

\author{Brian R. Williams}
\address{School of Mathematics\\
University of Edinburgh \\ 
Edinburgh \\ 
UK}
\email{brian.williams@ed.ac.uk}

\date{September 8, 2020}

\begin{abstract}
We formulate the abelian six-dimensional $\cN=(2,0)$ theory perturbatively, in a generalization of the Batalin--Vilkovisky formalism. 
Using this description, we compute the holomorphic and non-minimal twists at the perturbative level. 
This calculation hinges on the existence of an $L_\infty$ action of the supersymmetry algebra on the abelian tensor multiplet, which we describe in detail.  Our formulation appears naturally in the pure spinor superfield formalism, but understanding it requires developing a presymplectic generalization of the BV formalism, inspired by Dirac's theory of constraints. The holomorphic twist consists of symplectic-valued holomorphic bosons from the $\N=(1,0)$ hypermultiplet, together with a degenerate holomorphic theory of holomorphic coclosed one-forms from the $\N=(1,0)$ tensor multiplet, which can be interpreted as representing the intermediate Jacobian. 
  We check that our formulation and our results match with known ones under various dimensional reductions, as well as comparing the holomorphic twist to Kodaira--Spencer theory. Matching our formalism to five-dimensional Yang--Mills theory after reduction leads to some issues related to electric--magnetic duality; we offer some speculation on a nonperturbative resolution.
\end{abstract}

\maketitle

\newpage

\spacing{1.5}

\tableofcontents


\section{Introduction}

There is a supersymmetric theory in six dimensions whose fields include a two-form with self-dual field strength. 
Concrete and direct formulations of this theory, whose field content is referred to as the tensor multiplet, have remained elusive, despite an enormous amount of work and numerous applications, predictions, and consistency checks. One main difficulty is that the theory is believed not to admit a Lagrangian description, meaning that its equations of motion---even for the free theory---do not arise from a standard covariant action functional via the usual methods of variational calculus. 

Part of the desire to better understand theories of tensor multiplets is due to their ubiquity in the context of string theory and $M$-theory. 
The tensor multiplet with $\cN = (2,0)$ supersymmetry, valued in the Lie algebra $\lie{u}(N)$, famously appears as the worldvolume theory of $N$ coincident $M5$-branes; this theory has been the topic of, and inspiration for, an enormous amount of research. The literature is too large to survey here, but we give a few selected references below. 


Our main objective in this paper is to compute the twists of the abelian tensor multiplet. (We restrict to a perturbative analysis of the free theory with abelian gauge group; as such, we do not touch on issues relating to the interacting superconformal theories expected in the nonabelian case $N>1$, although we believe that some of our structural insights should be of use in that setting as well.) On general grounds, the $\N=(2,0)$ supersymmetry algebra in  six dimensions admits two twists: a holomorphic or minimal twist, together with a non-minimal twist that is defined on the product of a Riemann surface and a smooth four-manifold. 
Just at the level of the physical fields, a first rough statement of our results is as follows.

\begin{thm*}
The abelian $\cN=(2,0)$ admits two inequivalent classes of twists described as follows.
\begin{itemize}
\item[(1)]
The holomorphic twist exists on any complex three-fold $X$ equipped with a square-root of the canonical bundle $K^{\frac12}_X$.
It is equivalent to a theory whose physical fields are a $(1,1)$-form $\chi^{1,1}$, a $(0,2)$-form $\chi^{0,2}$, and a pair of symplectic fermions $\psi^i$, $i=1,2$ which transform as $(0,1)$ forms with values in $K^{\frac12}_X$. 
These fields obey the equations  
\begin{align*}
\dbar \chi^{1,1} + \partial \chi^{0,2} & = 0 ,\\
\dbar \chi^{0,2} & = 0 , \\
\dbar \psi^i & = 0 .
\end{align*}
The gauge symmetries of this theory are parameterized by form fields $\alpha^{0,1}$ and~$\alpha^{1,0}$, together with a pair of symplectic fermion gauge fields $\xi^i$, $i=1,2$ which are sections of $K^{\frac12}_X$.
They act by the formulas
\begin{align*}
  \chi^{1,1} &\mapsto \chi^{1,1} + \partial \alpha^{0,1} + \dbar \alpha^{1,0}, \\
  \chi^{0,2} &\mapsto \chi^{0,2} + \dbar \alpha^{0,1}, \\
  \psi^i &\mapsto \psi^i + \dbar \xi^i.
 \end{align*}

\item[(2)] 
The non-minimal twist exists on any manifold of the form $M \times \Sigma$ where $M$ is a smooth four-manifold and $\Sigma$ is a Riemann surface. 
It is equivalent to a theory whose physical fields are a pair 
\[
  (\chi^{2 ; 0,0} , \chi^{1 ; 0,1}) \in  \Omega^2(M) \otimes \Omega^{0}(\Sigma)  \oplus  \Omega^1(M) \otimes \Omega^{0,1}(\Sigma) . 
\]
This pair obeys the equations of motion 
\begin{align*}
  \dbar \chi^{2 ; 0,0} + \d \chi^{1;0,1} & = 0 \\
\d \chi^{2;0,0} & = 0 .
\end{align*}
Here $\dbar$ is the $\dbar$-operator on $\Sigma$ and $\d$ is the de Rham operator on $M$. 
The theory has gauge symmetries by fields $\alpha^{1 ; 0,0}$ and~$\alpha^{0 ; 0,1}$, which act via $\chi^{2 ; 0,0} \mapsto \chi^{2 ; 0,0} + \d \alpha^{1 ; 0,0}$ and $\chi^{1 ; 0,1} \mapsto \chi^{1 ; 0,1} + \d\alpha^{0 ; 0,1} + \dbar \alpha^{1 ; 0,0}$. 
\end{itemize}
\end{thm*}

The full statements of these results appear below in Theorems~\ref{thm:twist}, \ref{thm:nm}, and~\ref{thm:nmSO4}. 
Making sense of these twists and proving the theorems rigorously requires a great deal of groundwork, which leads us to develop some general theoretical tools that we expect to be of use outside the context of six-dimensional supersymmetry.

The main subtlety of the $\cN=(2,0)$ theory, and the twists above, is that they do {\em not} arise as the variational equations of motion of a local action functional. 
Thus, our first goal is to give a precise mathematical formulation of the perturbative theory of the free $\N=(2,0)$ tensor multiplet. (Of course, a corresponding formulation of the $\N=(1,0)$ tensor multiplet follows immediately from this.)  
Throughout the paper, we make use of the Batalin--Vilkovisky (BV) formalism;
see~\cite{CG2, CosRenorm} for a modern treatment of this setup, and~\cite{Schwarz-BV, BV} for a more traditional outlook.
Roughly, the data of a classical theory in the BV formalism is a graded space of fields $\cE_{\rm BV}$ (given as the space of sections of some graded vector bundle on spacetime), together with a symplectic form $\omega_{\rm BV}$ of cohomological degree $(-1)$ on $\cE_{\rm BV}$ and an action functional. The (degree-one) Hamiltonian vector field associated to the action functional defines a differential on~$\cE_{\rm BV}$. Under appropriate conditions, this differential provides a free resolution to  the sheaf of solutions to the equations of motion of the theory, modulo  gauge equivalence. 

It is clear that this formalism does not extend to the tensor multiplet in a straightforward way. The issue arises from the presence of the self-duality constraint on the field strength of the two-form, and is independent of supersymmetry and of other details about this particular theory. 
In Lorentzian signature, self-dual constraints on real $2k$-form fields can be imposed in spacetime dimension $4 k + 2$, where $k = 0,1, \ldots$.\footnote{In the literature, such constraints are sometimes called ``chiral.'' To avoid confusion, we will reserve this term for a different constraint that can be defined on complex geometries, and that will play a large role in what follows.} We will always work in Euclidean signature in this paper, and therefore also with complexified coefficients; for us, the self-dual constraint in six dimensions therefore takes the form 
\deq{
  \alpha \in \Omega^{2}(M^6)  , \quad  \star \d \alpha = \sqrt{-1} \, \d \alpha . 
}
The Yang--Mills style action of a higher form gauge theory would be given by the $L^2$-norm $\|\d \alpha \|_{L^2} = \int \d \alpha \wedge \star \d \alpha$. 
It is clear that the self-duality condition implies the norm vanishes identically, so an action functional of Yang--Mills type is not feasible~\cite{WittenConformal}.
Writing a covariant Lagrangian of any standard form for the tensor multiplet has been the subject of much effort, and is generally thought to be impossible (although various formulations have been proposed in the abelian case; see, for example, \cite{BLNPST,APPS,HS1,HS2}). A standard BV formulation of the theory, along the lines of more familiar examples, is thus out of reach for this reason alone.

The formulation we use was motivated by the desire to understand the pure spinor superfield formalism for $\N=(2,0)$ supersymmetry; the relevant cohomology was first computed in~\cite{CederwallM5}, and was rediscovered and reinterpreted in~\cite{NV}. Roughly speaking, this formalism takes as input an equivariant sheaf over the space of Maurer--Cartan elements, or nilpotence variety, of the supertranslation algebra, and produces a chain complex of locally free sheaves over the spacetime, together with a homotopy action of the corresponding supersymmetry algebra. The resulting multiplet can be interpreted as the BRST or BV formulation of the corresponding free multiplet, according to whether the action of the supersymmetry algebra closes on shell or not; the differential, which is also an output of the formalism, corresponds in the latter case to the Hamiltonian vector field mentioned above. 

In the case of $\N=(2,0)$ supersymmetry, the action of the algebra is, as always, guaranteed on general grounds. The fields exhibit an obvious match to the content of the $(2,0)$ tensor multiplet, and the differential includes the correct linearized equations of motion. 
(In fact, the resulting multiplet contains no auxiliary fields at all.) 
One thus expects to have obtained an on-shell formalism, but the interpretation of the resulting resolution as a BV theory is subtle for a new reason: there is no obvious or natural shifted symplectic pairing.  
In fact, developing a framework for studying the multiplets produced by pure-spinor techniques 
requires a generalization of the standard formalism, which necessarily allows for \emph{degenerate} pairings.

In classical symplectic geometry, symplectic pairings that are not required to be nondegenerate are called \emph{presymplectic}. In fact, presymplectic structures have played a role in physics before, in Dirac's theory of constrained systems. In this context, the origin is clear: while symplectic structures do not pull back, presymplectic structures (which are just closed two-forms) do. Any submanifold of a symplectic phase space, such as a constraint surface, thus naturally inherits a presymplectic structure. 


The simplest situation where the issues of self-duality constraints arise occurs for $k=0$, in the context of two-dimensional conformal field theory.
Here, the constraint is precisely the condition of holomorphy, and the theory of a self-dual zero-form is just the well-known chiral boson. We take a brief intermezzo to remark on this theory briefly, to offer the reader some familiar context for our more general considerations. 

\subsubsection*{The chiral boson}
In Lorentzian signature, 
the theory of the (periodic) chiral boson describes left-moving circle-valued maps.
Working perturbatively, as we do throughout this paper, the periodicity plays no role, so that the field is simply a left-moving real function; after switching to Euclidean signature (and correspondingly complexifying), we have a theory of maps $\varphi$ that are simply holomorphic functions on a Riemann surface. 

As discussed above, one role of the BV formalism is to provide a resolution of the sheaf of solutions to the equations of motion by smooth vector bundles. For the sheaf of holomorphic functions, such a resolution is straightforward to write down: 
it is just given by the Dolbeault complex $\Omega^{0,\bu}(\Sigma)$.

The chiral boson is not a theory in the usual sense of the word, perturbatively or otherwise, as it is not described by an action functional: the equations of motion, namely that $\varphi$ be holomorphic, do not arise as the variational problem of a classical action functional. Relatedly, the free resolution $\Omega^{0,\bu}(\Sigma)$ is not a BV theory, as it does not admit a nondegenerate pairing of an appropriate kind. 
Nevertheless, there is a way to formulate the chiral boson in a slightly modified version of the BV formalism, by interpreting holomorphy (which, in this setting, is the same as self-duality) as a \emph{constraint}.

To do this, we first consider a closely related theory, the (non-chiral) free boson, which does have a description in the BV formalism.
The free boson is a two-dimensional conformal field theory whose perturbative fields, in Euclidean signature, are just a smooth complex-valued function~$\varphi$ on~$\Sigma$;
the equations of motion impose that $\varphi$ is harmonic. 
In the BV formalism, one can model this free theory by the following two-term cochain complex
\deq[eq:freebos]{
\begin{array}{ccccc}
& & \ul{0} & &  \ul{1} \\
\cE_{\rm BV} & = & \Omega^0(\Sigma) & \xto{\partial \dbar} & \Omega^2(\Sigma) .
\end{array}
}
We can equip $\cE$ with a degree $(-1)$ antisymmetric non-degenerate pairing, which in this case is just given by multiplication and integration. 
That is
\[
\omega_{\rm BV} (\varphi, \varphi^+) = \int \varphi \varphi^+
\]
where $\varphi \in \Omega^0(\Sigma)$ and $\varphi^+ \in \Omega^2(\Sigma)$. 
This is the $(-1)$-shifted symplectic form, and the differential in~\eqref{eq:freebos} is the Hamiltonian vector field associated to the free action functional, as described in general above.

Now, there is a natural map of cochain complexes 
\[
i : \Omega^{0,\bu} (\Sigma) \to \cE_{\rm BV}
\]
which in degree zero is the identity map on smooth functions, and in degree one is defined by the holomorphic de Rham operator $\partial : \Omega^{0,1}(\Sigma) \to \Omega^2 (\Sigma)$. 
We can pull back the degree $(-1)$ symplectic form $\omega$ on $\cE$ to a two-form $i^* \omega$ on $\Omega^{0,\bu}(\Sigma)$, which is closed because $i$ is a cochain map. 
Explicitly, this two-form on the space $\Omega^{0,\bu}(\Sigma)$ is $(i^* \omega)(\alpha, \alpha') = \int \alpha \partial \alpha'$. 

Since $i$ is not a quasi-isomorphism, $i^* \omega$ is degenerate, and hence does not endow $\Omega^{0,\bu}(\Sigma)$ with a BV structure. However, it is useful to think of $i^*\omega$ as a shifted \emph{presymplectic} structure on the chiral boson, encoding ``what remains'' of the standard BV structure after the constraint of holomorphy has been imposed.

In analogy with ordinary symplectic geometry, we will refer to the data of a pair $(\cE, \omega)$ where $\cE$ is a graded space of fields, and $\omega$ is a closed two-form on $\cE$, as a {\em presymplectic} BV theory.
We make this precise in Definition \ref{dfn:preBV}, at least for the case of free theories.
In the example of the chiral boson this pair is $(\Omega^{0,\bu}(\Sigma), i^* \omega_{\rm BV})$.
\\

The theory of the self-dual two-form in six-dimensions (more generally a self-dual $2k$-form in $4k+2$ dimensions) arises in an analogous fashion. 
There is an honest BV theory of a nondegenerate two-form on a Riemannian six-manifold, which endows the theory of the self-dual two-form with the structure of a presymplectic BV theory. 
Among other examples, we give a precise formulation of the self-dual two-form in \S \ref{sec:pre}.

As for standard BV theories, one would hope to have a theory of observables, techniques for quantization, and so on in the context of \psBV{} theories. To develop a theory of classical observables, we make use of the theory of \emph{factorization algebras}.
Costello and Gwilliam have developed a mathematical approach to the study of observables in perturbative field theory, of which local operators are a special case. 
The general philosophy is that the observables of a perturbative (quantum) field theory have the structure of a factorization algebra on spacetime~\cite{CG1,CG2}. 
Roughly, this factorization algebra of observables assigns to an open set $U$ of spacetime a cochain complex $\Obs(U)$ of ``observables with support contained in $U$.''
When two open sets $U$ and~$V$ are disjoint and contained in some bigger open set $W$, the factorization algebra structure defines a rule of how to ``multiply" observables $\Obs(U) \otimes \Obs(V) \to \Obs(W)$. 
For local operators, one should think of this as organizing the operator product expansion in a sufficiently coherent way.

In the ordinary BV formalism, the factorization algebra of observables has a very important structure, namely a Poisson bracket of cohomological degree $+1$ induced from the shifted symplectic form $\omega_{\rm BV}$. 
This is reminiscent of the Poisson structure on functions on an ordinary symplectic manifold, and is a key ingredient in quantization.

In the case of a presymplectic manifold, the full algebra of functions does {\em not} carry such a bracket. But there is a subalgebra of functions, called the {\em Hamiltonian} functions, that does.
This issue persists in the \psBV{} formalism, and some care must be taken to define a notion of observables that carries such a shifted Poisson structure. 
We tentatively solve this problem, and for special classes of free \psBV{} theories we provide an appropriate notion of ``Hamiltonian observables."
The corresponding factorization algebra carries a shifted Poisson structure, which is a direct generalization of the work of Costello--Gwilliam that works to include presymplectic BV theories.\footnote{The development of the theory of observables for more general \psBV{} theories is part of ongoing work with Eugene Rabinovich.}
While the Hamiltonian observables provide a way of understanding a large class of observables in \psBV{} theories, we emphasize that a full theory should be expected to contain additional, nonperturbative observables: the Hamiltonian observables of the chiral boson, for example, agree with the $\U(1)$ current algebra, and therefore do \emph{not} see observables (such as vertex operators) that have to do with the bosonic zero mode. 

Using this formalism, we formulate the abelian tensor multiplet as a  \psBV{} theory, and go on to work out the full $L_\infty$ module structure encoding  the on-shell action of supersymmetry. Our formalism is distinguished from other formulations of the abelian tensor multiplet in that it extends supersymmetry off-shell \emph{without} using any auxiliary fields, in the homotopy-algebraic spirit of the BV formalism. Using this $L_\infty$ module structure, we rigorously compute both twists; like the full theory, these are \psBV{} theories. In eliminating acyclic pairs to obtain more natural descriptions of the twisted theories, we are forced to carefully consider what it means for a quasi-isomorphism to induce an equivalence of shifted presymplectic structures; understanding these equivalences is crucial for correctly describing  both the presymplectic structure on the holomorphic twist and the action of the residual supersymmetry there. 

Our results allow us to compare concretely to Kodaira--Spencer theory on Calabi--Yau threefolds, which is expected to play a role in the proposed description of holomorphically twisted supergravity theories due to Costello and Li~\cite{CostelloLi}. It would be  interesting to try and incorporate our results into the framework of the nonminimal twist of 11d supergravity, which we expect to agree with the proposals for ``topological M-theory'' considered in the literature~\cite{TopoM,GerasimovShatashvili,NekrasovZ}; branes in topological M-theory were considered in~\cite{CederwallM5top}. However, we reserve more substantial comparisons for future work.

We are also able to perform a number of consistency checks with known results on holomorphic twists of theories arising by dimensional reduction of the tensor multiplet.
At the level of the holomorphic twist, we show that the reduction to four-dimensions yields the expected supersymmetric Yang--Mills theories. Furthermore, when we compactify along along a four-manifold we recover the ordinary chiral boson on Riemann surfaces. 

Finally, we discuss dimensional reduction to five-dimensional Yang--Mills theory at the level of the \emph{untwisted} theory. Issues related to electric--magnetic duality appear naturally and play a key role here; furthermore, obtaining the correct result on the nose requires correctly accounting for nonperturbative phenomena that are missed by our perturbative approach. Although we do not rigorously develop the presymplectic BV formalism at a nonperturbative level in this work, we speculate about a nonperturbative formulation for gauge group $\U(1)$, and argue that our proposal gives the correct  dimensional reduction on the nose at the level of chain complexes of sheaves. Doing this requires a conjectural description of the theory of abelian $p$-form fields in terms of a direct sum of two Deligne cohomology groups, which can be interpreted as a complete (nonperturbative) presymplectic BV theory in novel fashion. 

\subsection*{Previous work}
 
There has been an enormous amount of previous work in the physics literature on topics related to M5 branes and $\N=(2,0)$ superconformal theories in six dimensions, and any attempt to provide exhaustive references is doomed to fail. In  light of this, our bibliography makes no pretense to be complete or even representative. 
The best we can offer is an extremely brief and cursory overview of some selected past literature, which may serve to orient the reader; for more complete background, the reader is referred to the references in the cited literature, and in particular to the reviews~\cite{SezginSundell,BermanReview,DijkgraafFivebrane}. 

Tensor multiplets in six dimensions were constructed in~\cite{HoweSierraTownsend}. The earliest approaches to the M5-brane involved the study of relevant ``black brane'' solutions in eleven-dimensional supergravity theory~\cite{Guven}; perhaps the first intimation that corresponding six-dimensional theories should exist was made by considering type IIB superstring theory on K3 singularities in~\cite{WittenComments}. The abelian M5-brane theory was worked out, including various proposals for Lagrangian formulations, in~\cite{HS1,HS2}, in~\cite{APPS}, in~\cite{CederwallM5-action}, and in~\cite{BLNPST}, following the general framework for chiral fields in~\cite{PST}. These formulations were later shown to be equivalent in~\cite{PST-equiv}. The connection of the tensor multiplet to supergravity solutions on AdS$_4 \times S^7$ was discussed in~\cite{CKvP} with an emphasis on $\N=(2,0)$ superconformal symmetry.

As to twisting the theory, the non-minimal twist was studied in~\cite{TwistingTrouble,TwistOffShell}, and a close relative earlier in~\cite{BaulieuWest}. (The approach of the latter paper effectively made use of the twisting homomorphism appropriate to the unique topological twist in five-dimensional $\N=2$ supersymmetry; this is the dimensional reduction of the six-dimensional non-minimal twist.)  While these studies compute the nonminimal twist at a nonperturbative level, \cite{TwistingTrouble,TwistOffShell} do so only after compactification to four dimensions along the Riemann surface in the spacetime $\Sigma \times M^4$, and thus do not see the holomorphic dependence on~$\Sigma$ explicitly. Our results are thus in some sense orthogonal. The relevance of the full nonminimal twist for the AGT correspondence was emphasized in~\cite{YagiAGT}; it would be interesting to connect our results to the AGT~\cite{AGT} and 3d-3d~\cite{3d3d} correspondences. 

The holomorphic twist has, as far as we know, not been considered explicitly before, although the supersymmetric index of the abelian theory was computed in~\cite{BakGustavsson}. We expect agreement between the character of local operators in the holomorphic theory \cite{SWchar} and the index studied there, after correctly accounting for nonperturbative operators, but do not consider that question in the present work and hope to study it in the future. We note, however, that the $\PP_0$ factorization algebra arising as the Hamiltonian observables of the holomorphically twisted (1,0) theory was studied in~\cite{BrianOwenEugene} as a boundary system for seven-dimensional abelian Chern--Simons theory. (The relation between the six-dimensional self-dual theory and seven-dimensional Chern--Simons theory is the subject of earlier work by~\cite{BelovMoore}, among others.) We see both these results and our results here as progress towards an understanding of the holomorphically twisted version of the AdS$_7$/CFT$_6$ correspondence.

Recently, there has been new progress on the question of finding a formulation of the nonabelian theory; much of this progress makes use of higher algebraic or homotopy-algebraic structure. See, for example, \cite{FSS-M5}, \cite{SaemannSchmidt}, and~\cite{LP,LambertLagrangian}. It would be interesting either to study twisting some of these proposals, or to attempt to make further progress on these questions by searching for nonabelian or interacting generalizations of the twisted theories studied here. These might be easier to find than their nontwisted counterparts and offer new insight into the nature of the interacting $(2,0)$ theory. We look forward to working on such questions in the future, and hope that others are inspired to pursue similar lines of attack. 

For the physicist reader, we emphasize that we deal here with a formulation that is lacking, even at a purely classical level, in at least three respects. Firstly, we make no effort to formulate the theory non-perturbatively, even for gauge group $U(1)$; in a sense, our discussion deals only with the gauge group $\R$. (Some more speculative remarks about this, though, are given in~\S\ref{sec:dimred}.) In keeping with this, our analysis here does not yet deal carefully with issues of charge quantization; as such, the subtle issues considered in~\cite{WittenM5,HopkinsSinger} and generalized in~\cite{FreedGeneralized} make no appearance, although we expect them to play a role in correctly extending our theory to the ``nonperturbative'' setting of gauge group $\U(1)$.  
Lastly, we start with a formulation which does not involve any coupling to eleven-dimensional supergravity, and makes no attempt to connect to the M5 brane, in the sense that we ignore the formulation of the theory in terms of a theory of maps. Associated issues (such as WZW terms and kappa symmetry) therefore make no appearance, although the connection to Kodaira--Spencer theory is indicated to show how we see our results as fitting into a larger story about twisted supergravity theories in the sense of Costello and Li~\cite{CostelloLi}, as mentioned above. 

\subsection*{An outline of the paper}

We begin in \S \ref{sec:pre} by setting up a presymplectic version of the BV formalism for free theories. 
After stating some general results and reviewing a list of examples, the section culminates with a definition of the factorization algebra of Hamiltonian  observables for a class  of \psBV{} theories.
In \S \ref{sec:susy} we recall the necessary tools of six-dimensional supersymmetry and provide a definition of the $\cN=(1,0)$ and $\cN=(2,0)$ versions of the tensor multiplet in the \psBV{} formalism. 
We review the classification of possible twists, and then give an explicit description of the \psBV{} theory as an $L_\infty$ module for the supersymmetry algebra. 
We perform the calculation of the minimal twist of the tensor multiplet in \S \ref{sec:holtwist}, and of the non-minimal twist in \S \ref{sec:nonmin}. 
We touch back with string theory in \S \ref{sec:bcov}, where we relate our twisted theories to the conjectural twist of Type IIB supergravity due to Costello--Li. 
Finally, in \S \ref{sec:dimred}, we explore some consequences of our description of the twisted theories upon dimensional reduction. We perform some sanity checks with theories that are conjecturally obtained as the reduction of the theory on the M5 brane, culminating in a computation of the dimensional reduction of the \emph{untwisted} theory along a circle. Some interesting issues related to electric-magnetic duality appear naturally; we discuss these, and end with some speculative remarks on nonperturbative generalizations of our results.

\subsection*{Conventions and notations}
\label{sec: conventions}

\begin{itemize}
\item If $E \to M$ is a graded vector bundle on a smooth manifold $M$, then we define the new vector bundle $E^! = E^* \otimes {\rm Dens}_M$, where $E^*$ is the linear dual and ${\rm Dens}_M$ is the bundle of densities on $M$. 
We denote by $\cE$ the space of smooth sections of $E$, and $\cE^!$ the space of sections of $E^!$. 
The notation $\cE_c$ refers to the space of compactly supported sections of $E$. 
The notation ($\Bar{\cE}_c$) $\Bar{\cE}$ refers to the space of (compactly supported) distributional sections of $E$. 
\item The sheaf of (smooth) $p$-forms on a smooth manifold $M$ will be denoted $\Omega^p(M)$ and $\Omega^\bu (M) = \bigoplus \Omega^p(M)[-p]$ is the $\ZZ$-graded sheaf of de Rham forms, with $\Omega^p(M)$ in degree $p$.
Often times when $M$ is understood we will denote the space of $p$-forms by $\Omega^p$. 
More generally, our grading conventions are cohomological, and are chosen such that the cohomological degree of a chain complex of differential forms is determined by the (total) form degree, but always taken to start with the lowest term of the complex in degree zero. Thus $\Omega^p$ is a degree-zero object, $\Omega^{\leq p}$ is a chain complex with support in degrees zero to $p$, and $\Omega^{\geq p}(\R^d)$ begins with $p$-forms in degree zero and runs up to $d$-forms in degree $d-p$. 
\item On a complex manifold $X$, we have the sheaves $\Omega^{i,\hol}(X)$ of holomorphic forms of type $(i,0)$. 
The operator $\partial : \Omega^{i,\hol} (X) \to \Omega^{i+1, \hol}(X)$ is the holomorphic de Rham operator. 
The standard Dolbeault resolution of holomorphic $i$-forms is $(\Omega^{i,\bu}(X), \dbar)$ where $\Omega^{i,\bu}(X) = \oplus_{k} \Omega^{i,k}(X)[-k]$ is the complex of Dolbeault forms of type $(i,\bu)$ with $(i,k)$ in cohomological degree $+k$. 
Again, when $X$ is understood we will denote forms of type $(i,j)$ by $\Omega^{i,j}$. 
\end{itemize}

\subsection*{Acknowledgements}

We thank D.~Butson, K.~Costello, C.~Elliott, D.~Freed, O.~Gwilliam, Si Li, N.~Paquette, E.~Rabinovich, S.~Raghavendran, P.~Safronov, P.~Teichner, J.~Walcher, P.~Yoo for conversation and inspiration of all kinds. I.S.\ thanks the Fields Institute at the University of Toronto for hospitality, as well as the Mathematical Sciences Research Institute in Berkeley, California and the Perimeter Institute for Theoretical Physics for generous offers of hospitality that did not take place due to Covid-19. This work is funded by the Deutsche Forschungsgemeinschaft (DFG, German Research Foundation) under Germany's Excellence Strategy EXC 2181/1 --- 390900948 (the Heidelberg STRUCTURES Excellence Cluster). The work of B.R.W. is supported by the National Science Foundation Award DMS-1645877 and by the National Science Foundation under Grant No. 1440140, while the author was in residence at the Mathematical Sciences Research Institute in Berkeley, California, during the semester of Spring 2020.

\section{A presymplectic Batalin--Vilkovisky formalism}
\label{sec:pre}

In the standard Batalin--Vilkovisky (BV) formalism \cite{Schwarz-BV}, one is interested in studying the (derived) critical locus of an action functional. 
On general grounds, derived critical loci are equipped with canonical $(-1)$-shifted symplectic structures~\cite{PTVV}.
In perturbation theory, where we work around a fixed classical solution, we can assume that the space of BV fields $\cE$ are given as the space of sections of some graded vector bundle $E \to M$, where $M$ is the spacetime. 
In this context, the $(-1)$-symplectic structure boils down to an equivalence of graded vector bundles $\omega : E \cong E^! [-1]$. 

We remind the reader that in the standard examples of ``cotangent'' perturbative BV theories, $E$ is of the form 
\deq{
E = T^*[-1]F \overset{\rm def}{=} F \oplus F^![-1],
}
where $F$ is some graded vector bundle, which carries a natural $(-1)$-symplectic structure. 
The differential $\QBV$ is constructed such that
\deq{
  H^0(\cE, \QBV) \cong \Crit(S),
}
i.e.\ so that the sheaf of chain complexes $(\cE,\QBV)$ is a model of the derived critical locus.

In general, we can think of the $(-1)$-symplectic structure $\omega$ as a two-form (with constant coefficients) on the infinite-dimensional linear space $\cE$. Moreover, this two-form is of a very special nature: it arises locally on spacetime. 
For a more detailed introduction to the BV formalism its description of perturbative classical field theory, see~\cite{CosRenorm, CG2}. 

We will be interested in a generalization of the BV formalism, motivated by the classical theory of presymplectic geometry and its appearance in Dirac's theory of constrained systems in quantum mechanics. 
In ordinary geometry, a presymplectic manifold is a smooth manifold $M$ equipped with a closed two-form $\omega \in \Omega^{2}(M)$, $\d \omega = 0$.
Equivalently, $\omega$ can be viewed as a skew map of bundles $TM \to T^*M$.
This is our starting point for the presymplectic version of the BV formalism in the derived and infinite dimensional setting of field theory.


\subsection{Presymplectic BV formalism}

We begin by introducing the presymplectic version of the BV formalism in terms of a two-form on the space of classical fields. 
This generalization shares many features with the usual BV setup: the two-form of degree $(-1)$ on arises ``locally" on spacetime, in the sense that it is defined by a differential operator acting on the fields.
In this paper we are only concerned with free theories, so we immediately restrict our attention to this case.

It is important for us that our complexes are bigraded by the abelian group $\ZZ \times \ZZ/2$.
We will refer to the integer grading as the \emph{cohomological} or \emph{ghost degree}, and the supplemental $\Z/2$ grading as \emph{parity} or \emph{fermion number}.

Before stating the definition of a free \psBV{} theory, we set up the following notion about the skewness of a differential operator. 
Let $E$ be a vector bundle on $M$ and suppose $D : \cE \to \cE^! [n]$ is a differential operator of degree $n$. 
The continuous linear dual of $\cE$ is $\cE^\vee = \Bar{\cE}_c^!$ (see \S \ref{sec: conventions}).
So, $D$ defines the following composition
\[
\Bar{D} : \cE_c \hookrightarrow \cE \xto{D} \cE^! [n] \hookrightarrow \Bar{\cE}^! [n] .
\]
The continuous linear dual of $\Bar{D}$ is a linear map of the same form $\Bar{D}^\vee : \cE_c \to \Bar{\cE}^![n]$. 
We say the original operator $D$ is {\em graded skew symmetric} if $\Bar{D} = (-1)^{n+1} \Bar{D}^\vee$. 

\begin{dfn}
\label{dfn:preBV}
A (perturbative) \defterm{free presymplectic BV theory} 
on a manifold $M$ is a tuple $(E, Q_{\rm BV}, \omega)$ where:
\begin{itemize}
\item $E$ is a finite-rank, $\Z\times\Z/2$-graded vector bundle on $M$, equipped with a differential operator
\[
Q_{\rm BV} \in {\rm Diff}(\cE , \cE) [1] 
\]
of bidegree $(1,0)$;

\item a differential operator
\[
\omega \in {\rm Diff}\left(\cE ,\cE^! \right) [-1]
\]
of bidegree $(-1,0)$;
\end{itemize}
which satisfy:
\begin{enumerate}
  \item[$(1)$] the operator $Q_{\rm BV}$ satisfies $(Q_{\rm BV})^2 = 0$, and the resulting complex $(\cE, Q_{\rm BV})$ is elliptic;
\item[$(2)$] the operator $\omega$ is graded skew symmetric with regard to the totalized $\Z/2$ grading;
\item[$(3)$] the operators $\omega$ and $Q_{\rm BV}$ are compatible: $[Q_{\rm BV} , \omega] = 0$.
\end{enumerate}
\end{dfn}

We refer to the fields $\phi \in \cE$ of cohomological degree zero as the ``physical fields". 
For free theories, the linearized equations of motion can be read off as $Q_{\rm BV} \phi = 0$. 
As is usual in the BRST/BV formalism, gauge symmetries are imposed by the fields of cohomological degree $-1$. 

The differential operator $\omega$ determines a bilinear pairing of the form
\[
\int_M \omega : \cE_c \times \cE_c \to {\rm Dens}_M [-1] \xto{\int_M} \CC [-1]
\]
which endows the compactly supported sections $\cE_c$ with the structure of a $(-1)$-shifted presymplectic vector space. 
Often, we will refer to a shifted presymplectic structure by prescribing the data of such a bilinear form on compactly supported sections. 

Of course, it should be clear that a (perturbative) free BV theory \cite[Definition 7.2.1.1]{CG2} is a free \psBV{} theory such that $\omega$ is induced from a bilinear map of vector bundles which is fiberwise non-degenerate. 
The notion of a free \psBV{} theory is thus a weakening of the more familiar definition. 
Indeed, when $\omega$ is an order zero differential operator such that $\omega : \cE \xto{\cong} \cE^![-1]$ is an isomorphism, the tuple $(E, Q_{\rm BV}, \omega)$ defines a free BV theory in the usual sense.

\begin{rmk}
There are two natural ways to generalize Definition \ref{dfn:preBV} that we do not pursue here:
\begin{description}[font=\normalfont\textit]
  \item[$-$ Non-constant coefficient presymplectic forms]
More generally, one can ask that $\omega$ be given as a polydifferential operator of the form
\[
\omega \in \prod_{n \geq 0} {\rm PolyDiff} (\cE^{\otimes n} \otimes \cE , \cE^!) [-1] .
\]
The right-hand side is what one should think of as the space of ``local" two-forms on $\cE$. 
\item[$-$ ``Interacting" \psBV{} formalism]
Here, we require that $\cL = \cE[-1]$ be equipped with the structure of a local $L_\infty$ algebra.
Thus, the space of fields $\cE$ should be thought of as the formal moduli space given by the classifying space $B \cL$. 
In the situation above, the free theory corresponds to an abelian local $L_\infty$ algebra, in which only the unary operation (differential) is nontrivial. 
\end{description}
There is a natural compatibility between these two more general structures that is required.
Using the description of the fields as the formal moduli space ${\rm B} \cL$, for some $L_\infty$ algebra $\cL$, one can view $\omega$ as a two-form $\omega \in \Omega^2({\rm B} \cL) = \clie^\bu(\cL , \wedge^2 \cL[1]^*)$. 
There is an internal differential on the space of two-forms given by the Chevalley--Eilenberg differential $\d_{\rm CE}$ corresponding to the $L_\infty$ structure. 
There is also an external, de Rham type, differential of the form $\d_{\rm dR} : \Omega^2({\rm B} \cL)  \to \Omega^3 ({\rm B} \cL)$. 
In this setup we require $\d_{\rm CE} \omega = 0$ and $\d_{\rm dR} \omega = 0$. 
We could weaken this condition further by replacing strictly closed two-forms on ${\rm B} \cL$ by $\Omega^{\geq 2}({\rm B} \cL)$ and asking that $\omega$ be a cocycle here. 
\end{rmk}

Since we only consider {\em free} presymplectic BV theories in this paper, we will simply refer to them as presymplectic BV theories.

\subsection{Examples of \psBV{} theories}
\label{ssec:examples}

We proceed to give some examples of \psBV{} theories, beginning with simple examples of degenerate pairings and proceeding to more ones more relevant to six-dimensional theories. 
The secondary goal of this section is to set up notation and terminology that will be used in the rest of the paper.

\begin{eg}
Suppose $(V, w)$ is a finite dimensional presymplectic vector space.
That is, $V$ is a finite dimensional vector space and $w: V \to V^*$ is a (degree zero) linear map which satisfies $w^* = - w$. 
Then, for any $1$-manifold $L$, the elliptic complex
\[
(\cE, Q_{\rm BV}) = \left(\Omega^\bu \otimes V , \; \d_{\rm dR} \right)
\]
is a \psBV{} theory on $L$ with
\[
\omega = {\id}_{\Omega^\bu} \otimes w : \Omega^\bu \otimes V \to \Omega_L^\bu \otimes V^* = \cE^! [-1] .
\]
Similarly, if $\Sigma$ is a Riemann surface equipped with a spin structure $K^{\frac12}$, then the elliptic complex
\[
(\cE, Q_{\rm BV}) = \left(\Omega^{0,\bu} \otimes K^{\frac12} \otimes V , \; \dbar \right)
\]
is a \psBV{} theory on $\Sigma$ with
\[
  \omega = {\id_{\Omega^{0,\bu} \otimes K^{\frac12}}} \otimes w  : \Omega^{0,\bu} \otimes K^{\frac12} \otimes V \to \Omega^{0,\bu} \otimes K^{\frac12} \otimes V^*  .
\]
\end{eg}

Each theory in this example arose from an ordinary presymplectic vector space, which was also the source of the degeneracy of~$\omega$.
The first example that is really intrinsic to field theory, and also relevant for the further discussion in this paper, is the following. 

\begin{eg}
  \label{eg:chiralboson}
Let $\Sigma$ be a Riemann surface and suppose $(W, h)$ is a finite dimensional vector space equipped with a symmetric bilinear form thought of as a linear map $h : W \to W^*$.
Then 
\[
(\cE, Q_{\rm BV}) = \left(\Omega^{0,\bu} \otimes W , \; \dbar \right) 
\]
is a \psBV{} theory with 
\[
\omega = \partial \otimes h : \Omega^{0,\bu} \otimes W \to \Omega^{1,\bu} \otimes W^* = \cE^! [-1] .
\]
We refer to this free \psBV{} theory as the \defterm{chiral boson} with values in $W$, and will denote it by $\chi(0, W)$ (see the next example). 
In the case that $W = \CC$, we will simply denote this by $\chi(0)$. 
\end{eg}

\begin{rmk}
  \label{rmk:h}
While we did not require $(W,h)$ to be nondegenerate in the above example, the theory is a genuinely presymplectic BV theory even if $h$ is nondegenerate. 
This corresponds to the standard notion of the chiral boson in the physics literature, and we will have no cause to consider degenerate pairings $h$ in what follows.
\end{rmk}

\begin{eg} \label{eg:chiral2kform}
Suppose $X$ is a $(2k+1)$-dimensional complex manifold.
Let $\Omega^{\bu, \hol} = \left(\Omega^{\bu, \hol}, \partial \right)$ be the holomorphic de Rham complex and let $\Omega^{\geq k+1, \hol}$ be the complex of forms of degree $\geq k + 1$. 
By the holomorphic Poincar\'{e} lemma, $\Omega^{\geq k+1,\hol}$ is a resolution of the sheaf of holomorphic closed $(k+1)$-forms.
Further, $\Omega^{\geq k+1, \hol}[-k-1]$ is a subcomplex of $\Omega^{\bu, \hol}$ and there is a short exact sequence of sheaves of cochain complexes
\[
\Omega^{\geq k+1, \hol} [-k-1] \to \Omega^{\bu,\hol} \to \Omega^{\leq k, \hol} 
\]
which has a locally free resolution of the form
\deq{
\Omega^{\geq k+1 , \bu}[-k-1] \to \Omega^{\bu,\bu} \to \Omega^{\leq k, \bu} .
}
In this sequence, all forms are smooth and the total differential is $\partial + \dbar$ in each term. 
We use this quotient complex $\Omega^{\leq k, \bu}$ to define another class of \psBV{} theories.

Let $(W, h)$ be as in the previous example. (Following Remark~\ref{rmk:h}, it may as well be nondegenerate.) 
The elliptic complex
\[
(\cE, Q_{\rm BV}) = \left(\Omega^{\leq k, \bu}_X \otimes W [2k] , \d = \partial + \dbar \right).
\]
is a \psBV{} theory with
\[
\omega = \partial \otimes h : \Omega^{\leq k, \bu}_X \otimes W [2k] \to \Omega^{\geq k+1, \bu}_X \otimes W^* [k] .
\]
We denote this \psBV{} theory by $\thy(2k,W)$, which we will refer to as the \defterm{chiral $2k$-form} with values in $W$. 
In the case $W = \CC$ we will simply denote this by $\thy(2k)$. 
\end{eg}

\begin{eg} 
  \label{eg:sd}
  Let $M$ be a Riemannian $(4k+2)$-manifold, and $(W,h)$ as above.
The Hodge star operator $\star$ defines a decomposition 
\deq{
  \Omega^{2k+1}(M) = \Omega^{2k+1}_+(M) \oplus \Omega^{2k+1}_-(M)
}
on the middle de Rham forms, 
where $\star$ acts by $\pm \sqrt{-1}$ on $\Omega^{2k+1}_\pm (M)$. 

Consider the following exact sequence of sheaves of cochain complexes:
  \deq{
    0 \to \Omega^{\geq 2k+1}_- [-2k-1] \to \Omega^\bu \to \Omega^{\leq 2k+1}_+ \to 0
  }
where
  \deq{
    \Omega^{\leq 2k+1}_+ = \left( \Omega^0 \xrightarrow{\d} \Omega^1 [-1] \xrightarrow{\d} \cdots \xrightarrow{\d} \Omega^{2k}[-2k] \xrightarrow{\d_+} \Omega^{2k+1}_+ [-2k-1]\right) , \qquad \d_+ = \frac{1}{2} (1 - \sqrt{-1}\star) \d,
  }
and
  \deq{
    \Omega^{\geq 2k+1}_- = \left(\Omega^{2k+1}_- \xrightarrow{\d} \Omega^{2k+2} [-1] \xrightarrow{\d} \cdots \xrightarrow{\d} \Omega^{4k+2}[-2k-1]\right)   }

Let
  \deq{
    (\cE,\QBV) = (\Omega^{\leq 2k+1}_+ \otimes W [2k], \; \d \;)
  }
and
\[
\omega = \d \otimes h : \Omega^{\leq 2k+1}_+ \otimes W [2k] \to \Omega^{\geq 2k+1}_- \otimes W^* .
\]
This data defines a \psBV{} theory $\thy_+(2k,W)$ on any Riemannian $(4k+2)$-manifold, which we will refer to as the \defterm{self-dual $2k$-form} with values in $W$. 
Again, in the case $W = \CC$ we will simply denote this by $\thy_+(2k)$. 
\end{eg}

\begin{rmk}
In general, the theories $\thy(2k)$ and $\thy_+(2k)$ are defined on different classes of manifolds; they can, however, be simultaneously defined when $X$ is a complex manifold equipped with a K\"ahler metric. 
Even in this case, they are distinct theories (although their dimensional reductions along $\C P^2$ both agree with the usual chiral boson; see \S \ref{sec:dimred}).
In \S \ref{sec:holtwist} we will show explicitly that the $\N=(1,0)$ tensor multiplet (which consists of $\thy_+(2)$ together with fermions and one scalar) becomes precisely $\thy(2)$ under a holomorphic twist.

There is, however, one case where the two theories $\thy(2k)$ and~$\thy_+(2k)$ coincide. A choice of metric on a Riemann surface determines a conformal class, which then corresponds precisely to a complex structure. As such, both of the theories $\thy(0)$ and $\thy_+(0)$ are always well-defined, and in fact agree; both are the theory of the chiral boson defined in Example~\ref{eg:chiralboson}.
\end{rmk}

We now recall a couple of examples of nondegenerate theories, for later convenience and to fix notation, that fit the definition of a standard free BV theory \cite[Definition 7.2.1.1]{CG2}.

\begin{eg}\label{eg:scalar}
Let $M$ be a Riemannian manifold of dimension $d$. 
Let $(W,h)$ be a complex vector space equipped with a non-degenerate symmetric bilinear pairing $h : W \cong W^*$. The theory $\Scalar(0,W)$ of the \defterm{free boson with values in~$W$} is the data
  \deq{
    (\cE,\QBV) = \left( \Omega^0(M) \otimes W \xrightarrow{\d \star \d \otimes \id_W} \Omega^{d}(M) \otimes W [-1] \right) , 
  }
and $\omega = {\id}_{\Omega^0} \otimes h + {\id}_{\Omega^d} \otimes h$.
Notice this is a BV theory, the $(-1)$ presymplectic structure is non-degenerate. 
\end{eg}


\begin{eg}\label{eg:freepform}
Let $(W, h)$ be as in the previous example, $p \geq 0$ an integer, and suppose $M$ is a Riemannian manifold of dimension $d \geq p$.
The theory $\Scalar(p,W)$ of \defterm{free $p$-form fields valued in~$W$} is defined \cite{ElliottAbelian} by the data
  \deq{
    (\cE,\QBV) = \left(\Omega^{\leq p} \otimes W [p] \xrightarrow{\d \star \d \otimes \id_W} \Omega^{\geq d-p} \otimes W [p-1] \right),
  }
with $(-1)$-symplectic structure $\omega = {\id}_{\Omega^{\leq p}} \otimes h + {\id}_{\Omega^{\geq d-p}} \otimes h$.
Notice again this is an honest BV theory, the presymplectic structure is non-degenerate. 
If $\alpha \in \cE$ denotes a field, the classical action functional reads $\frac{1}{2}\int h(\alpha, \d \star \d \alpha)$. 
\end{eg}

This example clearly generalizes the free scalar field theory, and also does not depend in any way on our special choice of dimension. 
We will simply write $\Scalar(p)$ for the case $W=\C$ when the spacetime $M$ is understood.

\begin{eg}
Let $M$ be as in the last example, and suppose in addition it carries a spin structure compatible with the Riemannian metric. 
Let $(R,w)$ be a complex vector space equipped with an antisymmetric non-degenerate bilinear pairing. 
The theory $\Fermi_-(R)$ of \defterm{chiral fermions valued in~$R$} is the data  \deq{
    (\cE,\QBV) = \Gamma(\Pi S_- \otimes R) \xrightarrow{\dslash \otimes \id_R} \Gamma(\Pi S_+ \otimes R)[-1],
    \qquad
  }
with $(-1)$-symplectic structure $\omega = {\rm id}_{S_+} \otimes w + {\rm id}_{S_-} \otimes w$. 
\end{eg}

We depart from the world of Riemannian manifolds to exhibit theories natural to the world of complex geometry that will play an essential role later on in the paper.

\begin{eg}
Suppose $X$ is a complex manifold of complex dimension $3$ which is equipped with a square-root of its canonical bundle $K_X^{\frac12}$.  
Let $(S, w)$ be a $\ZZ/2$-graded vector space equipped with a graded symmetric non-degenerate pairing. 
\defterm{Abelian holomorphic Chern--Simons theory} valued in $S$ is the free BV theory $\hCS (S)$ whose complex of fields is
\[
\Omega^{0,\bu}(X, K_{X}^{\frac12} \otimes S) [1]
\]
with $(-1)$-symplectic structure $\omega = {\rm id}_{\Omega^{0,\bu}} \otimes w$. 
This theory is naturally $\ZZ \times \ZZ/2$-graded and has action functional $\frac{1}{2}\int w(\alpha \wedge \dbar \alpha)$.
Notice that the fields in cohomological degree zero consist of $\alpha \in \Omega^{0,1}(X, K_{X}^{\frac12} \otimes S)$, and the equation of motion is $\dbar \alpha = 0$. 
This theory thus describes deformations of complex structure of the $\ZZ/2$-graded bundle $K_X^{\frac12} \otimes S$. 

We will be most interested in the case $S = \Pi R$ where $R$ is an ordinary (even) symplectic vector space, see Theorem \ref{thm:twist}. 
\end{eg}

\subsection{Presymplectic BV theories and constraints}

Perturbative \psBV{} theories stand in the same relationship to perturbative BV theories as presymplectic manifolds do to symplectic manifolds. 
Presymplectic structures obviously pull back along embeddings, whereas symplectic structures do not. There is thus always a preferred presymplectic structure on submanifolds of any (pre)symplectic manifold. 
In fact, this is the starting point for Dirac's theory of constrained mechanical systems \cite{DiracHamiltonian,Gotay}.

Each of the examples of \psBV{} theories we have given so far can be similarly understood as constrained systems relative to some (symplectic) BV theory. 

\begin{eg}[The chiral boson and the free scalar] \label{eg:chiralconstraint}
The chiral boson $\thy(0,W)$ on a Riemann surface $\Sigma$, from Example \ref{eg:chiralboson}, can be understood as a constrained system relative to the free scalar $\Phi(0,W)$, see Example \ref{eg:scalar}. 
At the level of the equations of motion this is obvious: the constrained system picks out the harmonic functions that are holomorphic. 

In the BV formalism, this constraint is realized by the following diagram of sheaves on $\Sigma$:
  \begin{equation}
    \begin{tikzcd}
      \Omega^{0,0} \arrow{r}{\del\delbar} & \Omega^{1,1} \\
      \Omega^{0,0} \arrow{u}{\id} \arrow{r}{\delbar} & \Omega^{0,1} \arrow{u}{\del}
    \end{tikzcd}
  \end{equation}
  It is evident that the diagram commutes, and that the vertical arrows define a cochain map upon tensoring with $W$:
  \deq{
    \thy(0, W) \to \Scalar(0, W).
  }
Furthermore, a moment's thought reveals that the $(-1)$-shifted presymplectic pairing on~$\thy(0,W)$ arises by pulling back the $(-1)$-shifted symplectic pairing on~$\Scalar(0,W)$.
\end{eg}

\begin{eg}[The self-dual $2k$-form and the free $2k$-form]
It is easy to form generalizations of the previous example. Consider the following diagram of sheaves on a Riemannian $(4k+2)$-manifold:
  \begin{equation}
    \begin{tikzcd}
      \Omega^0 \ar[r] & \cdots \ar[r] & \Omega^{2k} \arrow{r}{\d * \d} & \Omega^{2k+2} \ar[r] & \cdots \ar[r] & \Omega^{4k+2} \\
      \Omega^0 \arrow{u}{\id} \ar[r] & \cdots \ar[r] & \Omega^{2k} \arrow{u}{\id} \arrow{r}{\d_+} & \Omega^{2k+1}_+ \arrow{u}{\d}
    \end{tikzcd}
  \end{equation}
  Just as above, the vertical arrows of this commuting diagram define a cochain map
  \deq{
    \thy_+(2k, W) \to \Scalar(2k, W),
  }
under which the natural $(-1)$-shifted presymplectic structure of Example~\ref{eg:sd} arises by pulling back the $(-1)$-shifted symplectic form on $\Scalar(2k,W)$. 
\end{eg}

If $X$ is a complex manifold of complex dimension $2k+1$, the \psBV{} theory of the chiral $2k$-form $\chi(2k)$ is defined, see Example \ref{eg:chiral2kform}. 
As a higher dimensional generalization of Example \ref{eg:chiralconstraint}, $\chi(2k)$ can also be understood as a constrained system relative to theory of the free $2k$-form $\Phi(2k,W)$, see Example \ref{eg:freepform}. 
It is an instructive exercise to construct the similar diagram that witnesses the presymplectic structure on the chiral $2k$-form~$\thy(2k,W)$ by pullback from the ordinary (nondegenerate) BV structure on~$\Scalar(2k,W)$. 

\subsection{The observables of a \psBV{} theory} \label{sec:obs}
\label{sec:fact}

The classical BV formalism, as formulated in \cite{CG2}, constructs a factorization algebra from a classical BV theory, which plays the role of functions on a symplectic manifold in the ordinary finite dimensional situation. 

In symplectic geometry, functions carry a Poisson bracket.
In the classical BV formalism there is a shifted version of Poisson algebras that play a similar role. 
By definition, a $\PP_0$-algebra is a commutative dg algebra together with a graded skew-symmetric bracket of cohomological degree $+1$ which acts as a graded derivation with respect to the commutative product.
Classically, the BV formalism outputs a $\PP_0$-factorization algebra of classical observables \cite[\S 5.2]{CG2}. 

In this section, we will see that there is a $\PP_0$-factorization algebra associated to a {\em presymplectic} BV theory, which agrees with the construction of \cite{CG2} in the case that the \psBV{} theory is nondegenerate. 
Unlike the usual situation, this algebra is not simply the functions on the space of fields, but consists of certain class of functions.
We begin by recalling the situation in presymplectic mechanics. 

To any presymplectic manifold $(M,\omega)$ one can associate a Poisson algebra.
This construction generalizes the usual Poisson algebra of functions in the symplectic case, and goes as follows.
Let $\Vect(M)$ be the Lie algebra of vector fields on $M$, and define the space of \emph{Hamiltonian pairs} 
\deq{
\Ham(M, \omega) \subset \Vect(M) \oplus \cO(M)
}
to be the linear subspace of pairs $(X, f)$ satisfying $i_X \omega = \d f$.
Correspondingly, we can define the space of \emph{Hamiltonian functions} or \emph{Hamiltonian vector fields} to be the image of~$\Ham(M,\omega)$ under the obvious (forgetful) maps to~$\O(M)$ or~$\Vect(M)$ respectively. We will denote these spaces by~$\O^\omega(M)$ and~$\Vect^\omega(M)$. 
Notice that $\cO^\omega(M)$ is the quotient of $\Ham(M, \omega)$ by the Lie ideal $\ker(\omega) \subset \Ham(M,\omega)$. 

There is a bracket on $\Ham(M, \omega)$, defined by
\[
[(X, f), (Y, g)] = ([X,Y], i_X i_Y (\omega)) .
\]
On the right-hand side the bracket $[-,-]$ is the usual Lie bracket of vector fields.
Furthermore, there is a commutative product on $\Ham(M, \omega)$ defined by
\[
(X, f) \cdot (Y, g) = (g X + f Y , fg) .
\]
Together, they endow $\Ham(M, \omega)$ with the structure of a Poisson algebra. 
This Poisson bracket on Hamiltonian pairs induces a Poisson algebra structure on the algebra of Hamiltonian functions $\cO^\omega(M)$. 
%
%

In some situations, one can realize the Poisson algebra of Hamiltonian functions $\cO^\omega(M)$ as functions on a particular symplectic manifold. 
Associated to the presymplectic form $\omega$ is the subbundle
\deq{
  \ker(\omega) \subseteq TM
}
of the tangent bundle. 
The closure condition on~$\omega$ ensures that $\ker(\omega)$ is always involutive.
If one further assumes that the leaf space $M / \ker(\omega)$ is a smooth manifold, then $\omega$ automatically descends to a symplectic structure along the quotient map $q : M \to M / \ker(\omega)$. 
Pulling back along this map determines an isomorphism of Poisson algebras
\[
q^* : \cO(M / \ker(\omega)) \xto{\cong} \cO^\omega(M) .
\]
In particular, one can view the Poisson algebra of Hamiltonian functions as the $\ker(\omega)$-invariants of the algebra of functions $\cO^\omega (M) = \cO(M)^{\ker(\omega)}$.
Notice that this formula makes sense without any conditions on the niceness of the quotient $M / \ker(\omega)$. 

In our setting, the presymplectic data is given by a \psBV{} theory.
A natural problem is to define and characterize a version of Hamiltonian functions in this 
setting. 

\subsubsection{The factorization algebra of observables} \label{sec:preobs}

%
%
%

As we've already mentioned, given a (nondegenerate) BV theory the work of \cite{CG2} produces a factorization algebra of classical observables. 
If $(\cE, \omega, Q_{\rm BV})$ is the space of fields of a free BV theory on a manifold $M$ then this factorization algebra $\Obs_\cE$ assigns to the open set $U \subset M$ the cochain complex $\Obs_\cE (U) = \left(\cO^{sm}(\cE(U)) , Q_{\rm BV} \right)$. 
Here $\cO^{sm}(\cE(U))$ refers to the ``smooth" functionals on $\cE(U)$, which by definition are\footnote{Notice $\cE^!_c(U) \hookrightarrow \cE(U)^\vee$, so $\cO^{sm}$ is a subspace of the space of all functionals on $\cE(U)$.} 
\[
  \cO^{sm} (\cE(U)) = \Sym \left(\cE^!_c(U) \right).
\]
Furthermore, since $\omega$ is an isomorphism, it induces a bilinear pairing 
\[
\omega^{-1} : \cE^!_c \times \cE^!_c \to \CC[1] .
\]
By the graded Leibniz rule, this then determines a bracket 
\[
\{-,-\} : \cO^{sm}(\cE(U)) \times \cO^{sm}(\cE(U)) \to \cO^{sm}(\cE(U)) [1]
\]
endowing $\Obs_{\cE}$ with the structure of a $\PP_0$-factorization algebra, see \cite[Lemma 5.3.0.1]{CG2}. 

In this section, we turn our attention to defining the observables of a \psBV{} theory, modeled on the notion of the algebra of Hamiltonian functions in the finite dimensional presymplectic setting. 
Suppose that $(\cE, \omega, Q_{\rm BV})$ is a free \psBV{} theory. 
The shifted presymplectic structure is defined by a differential operator 
\[
\omega : \cE \to \cE^![-1] .
\]
In order to implement the structures we recounted in the ordinary presymplectic setting, the first object we must come to terms with is the solution sheaf of this differential operator $\ker(\omega) \subset \cE$. 

In general $\ker(\omega)$ is not given as the smooth sections of a finite rank vector bundle, so it is outside of our usual context of perturbative field theory. 
However, suppose we could find a semi-free resolution $(\cK_\omega^\bu, D)$ by finite rank bundles
\[
\ker(\omega) \xto{\simeq} \left(\cK_\omega^\bu, D \right) 
\]
which fits in a commuting diagram
\[
\begin{tikzcd}
\ker(\omega) \ar[dr] \ar[rr,"\simeq"] & & \cK_\omega^\bu \ar[dl, "\pi"] \\
& \cE & 
\end{tikzcd}
\]
where the bottom left arrow is the natural inclusion, and $\pi$ is a linear differential operator. 
In the more general case, where $\omega$ is nonlinear, we would require that $\cK_\omega^\bu$ have the structure of a dg Lie algebra resolving $\ker(\omega) \subset \Vect(\cE)$. 

Given this data, the natural ansatz for the classical observables is the (derived) invariants of $\cO(\cE)$ by $\cK_\omega^\bu$. 
A model for this is the Lie algebra cohomology:
\[
{\rm C}^\bu (\cK_\omega^\bu , \cO(\cE)) = {\rm C}^\bu (\cK_\omega^\bu \oplus \cE[-1]) .
\]
In this free case that we are in, this cochain complex is isomorphic to functions on the dg vector space $\cK_\omega^\bu [1] \oplus \cE$ where the differential is $D + Q_{\rm BV} + \pi$.

As in the case of the ordinary BV formalism, in the free case we can use the smoothed version of functions on fields. 

\begin{dfn}
Let $(\cE, \omega, Q_{\rm BV})$ be a free \psBV{} theory on $M$, and suppose $(\cK^\bu, D)$ is a semi-free resolution of $\ker(\omega) \subset \cE$ as above.
The cochain complex of \defterm{classical observables supported on the open set} $U \subset M$ is 
\begin{align*}
\Obs_{\cE}^\omega (U) & = \cO^{sm} \left(\cK_\omega^\bu (U) \oplus \cE(U) [-1] , D + Q_{\rm BV} + \pi \right) \\
& = \bigg( \Sym \left((\cK^\bu_\omega)^!_c (U) \oplus \cE^!_c (U) [1]\right) , D + Q_{\rm BV} + \pi \bigg) .
\end{align*}
By \cite[Theorem 6.0.1]{CG1} the assignment $U \mapsto \Obs_{\cE}^\omega (U)$ defines a factorization algebra on $M$, which we will denote by $\Obs^\omega_\cE$. 
\end{dfn}

\begin{eg}\label{eg:chiralobs}
Consider the chiral boson \psBV{} theory $\chi(0)$, see Example \ref{eg:chiralboson}, on a Riemann surface $\Sigma$. 
The kernel of $\omega = \partial$ is the sheaf of constant functions 
\[
\ker(\omega) = \ul{\CC}_\Sigma \subset \Omega^{0,\bu}(\Sigma) .
\]
By Poincar\'{e}'s Lemma, the de Rham complex $\left(\Omega^\bu_\Sigma, \d_{\rm dR} = \partial + \dbar\right)$ is a semi-free resolution of $\ul{\CC}_\Sigma$. 
Thus, the classical observables are given as the Lie algebra cohomology of the abelian dg Lie algebra
\[
\left(\Omega^\bu_\Sigma \oplus \Omega^{0,\bu}_\Sigma [-1] , \d_{\rm dR} + \dbar + \pi \right)
\]
where $\pi : \Omega^\bu_\Sigma \to \Omega^{0,\bu}_\Sigma$ is the projection. 
This dg Lie algebra is quasi-isomorphic to the abelian dg Lie algebra $\Omega^{1,\bu}_\Sigma [-1]$, so the factorization algebra of classical observables is
\[
\Obs_{\chi(0)}^\omega \simeq \cO^{sm} (\Omega^{1,\bu}_\Sigma) = \Sym \left(\Omega^{0,\bu}_{\Sigma, c} [1] \right) .
\]
\end{eg}

There are two special cases to point out. 

\begin{itemize}
\item[(1)] 
Suppose the shifted presymplectic form $\omega$ is an order zero differential operator. 
Then, $\ker(\omega)$ is a subbundle of $\cE$, so there is no need to seek a resolution. 
Furthermore, in this case $\cE / \ker(\omega)$ is also given as the sheaf of sections of a graded vector bundle $E / \ker(\omega)$, and $\omega$ descends to a bundle isomorphism $\omega : E / \ker(\omega) \xto{\cong} \left(E / \ker (\omega) \right)^![-1]$. 

In other words, $(\cE / \ker(\omega), \omega , Q_{\rm BV})$ defines a (nondegenerate) free BV theory.
The factorization algebra of the classical observables of the pre BV theory $\Obs^\omega_{\cE}$ agrees with the factorization algebra of the BV theory $\cE / \ker(\omega)$
\[
\Obs_{\cE / \ker(\omega)} = \left( \cO^{sm}(\cE / \ker(\omega) , Q_{\rm BV} \right) .
\]
In this case, the observables inherit a $\PP_0$-structure by \cite[Lemma 5.3.0.1]{CG2}. 

\item[(2)] 
This next case may seem obtuse, but fits in with many of the examples we consider. 
Suppose that the two-term complex 
\[
  \begin{tikzcd}[row sep = 0ex]
    &[0ex] \ul{0} & \ul{1} \\
    \Cone(\omega)[-1] :
    &[0ex] \cE \ar[r, "\omega"] & \cE^![-1],
\end{tikzcd}
\]
defined by the presymplectic form $\omega$,
is itself a semi-free resolution of $\ker(\omega)$.
(Though it is not quite precise, one can imagine this condition as requiring that $\omega$ have trivial cokernel.) In this case, it is immediate to verify that the factorization algebra of observables is
\[
\Obs_\cE^\omega = \left(\cO^{sm}(\cE^![-1]), Q_{\rm BV} \right) .
\]
We mention that in this case $\Obs_{\cE}^\omega$ is also endowed with a $\PP_0$-structure defined directly by $\omega$.
\end{itemize}

We can summarize the discussion in the two points above as follows. 

\begin{prop}
If the \psBV{} theory $(\cE, \omega, Q_{\rm BV})$ satisfies $(1)$ or $(2)$ above then the classical observables $\Obs^\omega_{\cE}$ form a $\PP_0$-factorization algebra. 
\end{prop} 

\begin{rmk} 
Generally speaking, the resolution of the solution sheaf $\ker(\omega)$ is given by the Spencer resolution.
We expect a definition of a $\PP_0$-factorization algebra of observables associated to any (non-linear) \psBV{} theory, though we do not pursue that here. 
\end{rmk}

For any $k$, the self-dual $2k$-form $\chi(2k, W)$ and the chiral $2k$-form satisfy condition (2) and so give rise to a $\PP_0$-factorization algebra of Hamiltonian observables. 
We will study this factorization algebra in depth in \S \ref{sec:bcov}.

\section{The abelian tensor multiplet} \label{sec:susy}

We provide a definition of the (perturbative) abelian $\N=(2,0)$ tensor multiplet in the \psBV{} formalism, together with the $\N=(1,0)$ tensor and hypermultiplets. 
As discussed in the previous section, in the BV formalism one must specify a $(-1)$-shifted symplectic (infinite dimensional) manifold, the fields, together with the data of a homological vector field which is compatible with the shifted symplectic form.
The tensor multiplets in six dimensions are peculiar, because they only carry a \psBV{} (shifted presymplectic) structure, as opposed to a symplectic one. 

Roughly speaking, the fundamental fields of the tensor multiplet consist of a two-form field whose field strength is constrained to be self-dual, a scalar field valued in some $R$-symmetry representation, and fermions transforming in the positive spin representation of $\Spin(6)$. 
The degeneracy of the shifted symplectic structure arises from the presence of the self-duality constraint on the two-form in the multiplet, just as in the examples in~\S\ref{ssec:examples}.

We begin by defining the field content of each multiplet precisely and giving the \psBV{} structure.
A source for the definition of the fields of the tensor multiplet in the BV formalism can be traced to the description in terms of the six-dimensional nilpotence variety given in \cite{ESW}. 
See Remark \ref{rmk:ESW}. 

The next step is to formulate the action of supersymmetry on the $(1,0)$ and $(2,0)$ tensor multiplets at the level of the BV formalism.
Here, one makes use of the well-known linear transformations on physical fields that are given in the physics literature. See, for example, \cite{BSvP} for the full superconformal transformations of the $\N=(2,0)$ multiplet; we will review the linearized super-Poincar\'e transformations below. 

However, these transformations do not define an action of~$\sp{2}$ on the space of fields. In the physics terminology, they close only on-shell (and after accounting for gauge equivalence). In the BV formalism, this is rectified by extending the action to an $L_\infty$ action on the BV fields. (See, just for example, \cite{Baulieu-susy} for an application of this technique.) 
For the hypermultiplet, this was performed explicitly in~\cite{ESW}; the hypermultiplet, however, is a symplectic BV theory in the standard sense. 
For the tensor multiplet, supersymmetry also only exists on-shell; no strict Lie module structure can be given. 
We work out the required $L_\infty$ correction terms, which play a nontrivial role  in our later calculation of the non-minimal twist.

We will first recall the definitions of the relevant supersymmetry algebras; afterwards, we will construct the multiplets as free perturbative \psBV{} theories, and go on to give the $L_\infty$ module structure on the $\N=(2,0)$ tensor multiplet. Of course the $\N=(1,0)$ transformations follow trivially from this by restriction. 

\subsection{Supersymmetry algebras in six dimensions} \label{sec:susy6}

Let $S_\pm \cong \CC^4$ denote the complex $4$-dimensional spin representations of $\Spin(6)$ and let $V \cong \CC^6$ be the vector representation. 
There exist natural $\Spin(6)$-invariant isomorphisms
\[
\wedge^2(S_\pm) \xto{\cong} V
\]
and a non-degenerate $\Spin(6)$-invariant pairing
\[
  (-,-) : S_+ \otimes S_- \to \CC .
\]
The latter identifies $S_+ \cong (S_-)^*$ as $\Spin(6)$-representations. Under the exceptional isomorphism $\Spin(6) \cong SU(4)$, $S_\pm$ are identified with the fundamental and antifundamental representation respectively.

The odd part of the complexified six-dimensional $\cN=(n,0)$ supersymmetry algebra is of the form
\[
  \Sigma_n = S_+ \otimes R_n ,
\]
where $R_n$ is a $(2n)$-dimensional complex symplectic vector space whose symplectic form we denote by $\omega_R$. There is thus a natural action of $\Sp(n)$ on~$R_n$ by the defining representation.
Note that we can identify the dual $\Sigma_n^* = S_- \otimes R_n$ as representations of $\Spin(6) \times \Sp(n)$.

The full $\N=(n,0)$ supertranslation algebra in six dimensions is the super Lie algebra
\[
  \st{n} = V \oplus \Pi \Sigma_n
\]
with bracket
\deq{
  [-, -]  = \wedge \otimes \omega_R: \wedge^2( \Pi  \Sigma_n) \to V.
}
This algebra admits an action of $\Spin(6) \times \Sp(n)$, where the first factor is the group of (Euclidean) Lorentz symmetries
and the second is called the $R$-sym\-me\-try group $G_R = \Sp(n)$. 
Extending the Lie algebra of $\Spin(6) \times \Sp(n)$ by this module produces the full $\N=(n,0)$ super-Poincar\'e algebra, denoted $\sp{n}$.

\begin{rmk}
  We can view $\sp{n}$ as a graded Lie  algebra by assigning degree zero to $\so(6) \oplus \lie{sp}(n)$, degree one to $\Sigma_n$, and degree two to~$V$. In physics, this consistent $\Z$-grading plays the role of the conformal weight. Both this grading and the $R$-symmetry action become inner in the \defterm{superconformal algebra}, which is the simple super Lie algebra
\deq{
  \sc{n} = \osp(8|n).
}
The abelian $\N=(2,0)$ multiplet in fact carries a module structure for~$\osp(8|2)$; computing the holomorphic twist of this action should lead to an appropriate algebra acting by supervector fields on the holomorphic theory we compute below, which should then extend to an action of all holomorphic vector fields on an appropriate superspace, following the pattern of~\cite{SCA}. However, we leave this computation to future work. 
\end{rmk}

For theories of physical interest, one considers $n = 1$ or~$2$. 
In the latter case, an accidental isomorphism identifies
$\Sp(2)$ with~$\Spin(5)$, which further identifies $R_2$ with the unique complex spin representation of $\Spin(5)$.

\subsubsection{Elements of square zero}
\label{ssec:nilps}

With an eye towards twisting, we recall the classification of square-zero elements in~\sp{n} for $n=1$ and~$2$, following~\cite{ChrisPavel,NV}. 
As above, we are interested in odd supercharges 
\deq[eq:oddQs]{
  Q \in \Pi \Sigma_n = \Pi S_+ \otimes R_n ,
}
which satisfy the condition $[Q,Q]=0$. 
Such supercharges define twists of a supersymmetric theory.

We will find it useful to refer to supercharges by their \emph{rank} with respect to the tensor product decomposition~\eqref{eq:oddQs} (meaning the rank of the corresponding linear map $R_n \to (S_+)^*$). It is immediate from the form of the supertranslation algebra that elements of rank one square to zero for any $n$. 

When $n = 1$, it is also easy to see that any square-zero element must be of rank one, so that the space of such elements is isomorphic to the determinantal variety of rank-one matrices in~$M^{4\times 2}(\C)$. This can in turn be thought of as the image of the Segre embedding
\deq{
  \PP^3 \times \PP^1 \hookrightarrow \PP^7.
}

For $n=2$, there are two distinct classes of such supercharges: those of rank one, which we will also refer to as minimal or holomorphic, and a certain class of rank-two elements, also called non-minimal or partially topological.
A closer characterization of the two types of square-zero supercharges is the following:

\begin{description}
  \item[Minimal (or holomorphic)]
A supercharge of this type is automatically square-zero.
Moreover, such a supercharge has three invariant directions, and so the resulting twist is a holomorphic theory defined on complex three-folds. 
Similarly to the $n=1$ case, the space of such elements is isomorphic to the determinantal variety of rank-one matrices in~$M^{4\times4}(\C)$, which is the image of the Segre embedding
\deq{
\PP^3 \times \PP^3 \hookrightarrow \PP^{15} .
}
We remark that in the case $n=2$, the supercharge $Q$ of rank one defines a $\N=(1,0)$ subalgebra $\sp{1} \cong \sp{1}^Q \subset \sp{2}$.
\item[Non-minimal (or partially topological)]
  Suppose $Q \in \Pi \Sigma_2$ is a rank-two supercharge (there is no such supercharge when $n=1$). 
  It can be written in the form
    \deq{
      Q = \xi_1 \otimes r_1 + \xi_2\otimes r_2.
    }
Since $\wedge^2 S_+ \cong V$, such an element must satisfy a single quadratic condition
\deq{
      w(r_1,r_2) = 0
    }
in order to be of square zero. 
Such a supercharge has five invariant directions, and the resulting twist can be defined on the product of a smooth four-manifold with a Riemann surface.
The space of all such supercharges is a subvariety of the determinantal variety of rank-two matrices in~$M^{4\times 4}(\C)$, cut out by this single additional quadratic equation. 
Just as for the determinantal variety itself, its singular locus is precisely the space of rank-one (holomorphic) supercharges.
\end{description}

We will compute the holomorphic twist below in~\S\ref{sec:holtwist} and the rank-two twist in~\S\ref{sec:nonmin}. There, we will also recall some further details about nilpotent elements in~$\st{2}$, showing how the non-minimal twist can be obtained as a deformation of a fixed minimal twist.

\begin{rmk}\label{rmk:ESW}
  In fact, a study of the space of Maurer--Cartan elements in~$\sp{2}$ was also a major motivation for the formulation of the supersymmetry multiplets that we use throughout this paper. In physics, the pure spinor superfield formalism~\cite{Cederwall} has been used as a tool to construct multiplets for some time. The relevant cohomology, corresponding to the field content of~$\cT_{(2,0)}$, was first computed in~\cite{CederwallM5}. 

  In~\cite{NV}, the pure  spinor superfield formalism was reinterpreted as a construction that produces a supermultiplet (in  the form of a cochain complex of vector  bundles)  from the data of an  equivariant  sheaf over the nilpotence variety. It was further observed that, when  the nilpotence variety is Calabi--Yau, Serre duality gives rise to the structure of a shifted symplectic pairing on the resulting multiplet, so that the full data of a BV theory is produced. More  generally, when the canonical bundle is not  trivial, the multiplet resulting from the canonical bundle  admits a pairing with the multiplet  associated to the structure sheaf. 

  As mentioned before, applying this formalism to the structure sheaf of the nilpotence variety for $\sp{2}$---the geometry of which was reviewed above---produces a cochain complex with a homotopy action of~$\sp{2}$ that corresponds precisely to the formulation we use in  this paper and explore in detail in the following section. For this space, however, the canonical bundle  is \emph{not} trivial; the multiplet associated to the  canonical  bundle is, roughly speaking, $\cT^!_{(2,0)}$, which can be  identified with  the space of linear Hamiltonian observables of~$\cT_{(2,0)}$.  It would be  extremely interesting to give a geometric description of the origin  of the presymplectic pairing on~$\cT_{(2,0)}$, but we do not pursue this here; our  use of this  pairing, as described above, is motivated by interpreting self-duality as a constraint and pulling back the pairing from the standard structure on the nondegenerate two-form.  
\end{rmk}

\subsection{Supersymmetry multiplets}

The two theories we are most interested in are the abelian $(1,0)$ and $(2,0)$ tensor multiplets. We define these here at the level of (perturbative, free) \psBV{} theories, and then go on to discuss the $\N=(1,0)$ hypermultiplet, which will also play a role in what follows.

First, we define the $(1,0)$ theory.
Recall that $R_1$ denotes the defining representation of $\Sp(1)$. 
\begin{dfn} 
The six-dimensional \defterm{abelian $\N=(1,0)$ tensor multiplet} is the \psBV{} theory $\mplet{1}$ defined by the direct sum of \psBV{} theories:
  \deq{
    \mplet{1} = \thy_+(2) \oplus \Fermi_-(R_1) \oplus \Scalar(0,\C),
  }
defined on a Riemannian spin manifold $M$. 
This theory has a symmetry by the group $G_R = {\rm Sp}(1)$ which acts on $R_1$ by the defining representation and trivially on the summands $\thy_+(2)$, $\Scalar(0,\CC)$. 
\end{dfn}

This theory admits an action by the supertranslation algebra $\st{1}$, which will be constructed explicitly below in~\S\ref{ssec:module}.

Note that the fields of cohomological degree zero together with their linear equations of motion are:
\begin{itemize}
\item a two-form $\beta \in \Omega^2(M)$, satisfying the linear constraint $\d_+(\beta) = 0 \in \Omega^3_+(M)$;
\item a spinor $\psi \in \Omega^0(M, S_- \otimes R_1)$, satisfying the linear equation of motion $(\dslash \otimes {\id}_{R_1}) \psi = 0 \in \Omega^0(M , S_+ \otimes R_1)$;
\item a scalar $\varphi \in \Omega^0(M)$, satisfying the linear equation of motion $\d \star \d \varphi = 0 \in \Omega^6(M)$. 
\end{itemize}

Next, we define the $(2,0)$ theory.
Recall, $R_2$ denotes the defining representation of $\Sp(2)$.
Let $W$ be the vector representation of $\Spin(5) \cong \Sp(2)$.
\begin{dfn}
The six-dimensional \defterm{abelian $\N=(2,0)$ multiplet} is the \psBV{} theory $\mplet{2}$ defined by the direct sum of \psBV{} theories:
  \deq{
    \mplet{2} = \thy_+(2) \oplus \Fermi_-(R_2) \oplus \Scalar(0,W).
  }
defined on a Riemannian spin manifold.
This theory has a symmetry by the group $G_R = {\rm Sp}(2)$ which acts on $R_2$ by the defining representation and $W$ by the vector representation upon the identification $\Sp(2) \cong \Spin(5)$. 
Note, $G_R = \Sp(2)$ acts trivially on the summand $\thy_+(2)$. 
\end{dfn}

This theory admits an action by the supertranslation algebra $\st{2}$, which will be constructed explicitly below in~\S\ref{ssec:module}.

Note that the fields of cohomological degree zero consist of
\begin{itemize}
\item a two-form $\beta \in \Omega^2(M)$, satisfying the linear constraint $\d_+(\beta) = 0 \in \Omega^3_+(M)$;
\item a spinor $\psi \in \Omega^0(M, S_- \otimes R_2)$, satisfying the linear equation of motion $(\dslash \otimes {\id}_{R_2}) \psi = 0 \in \Omega^0(M , S_+ \otimes R_2)$;
\item a scalar $\varphi \in \Omega^0(M, W)$, satisfying the linear equation of motion $(\d \star \d \otimes {\id}_W) \varphi = 0 \in \Omega^6(M, W)$. 
\end{itemize}

Lastly, we discuss the six-dimensional $\N=(1,0)$ hypermultiplet.
\begin{dfn}
Let $R$ be a finite-dimensional symplectic vector space over~$\C$, as above. The \defterm{$\N=(1,0)$ hypermultiplet valued in~$R$} is the following free (nondegenerate) BV theory in six dimensions:
  \deq{
    \hyper{R} = \Scalar(0,R_1 \otimes R) \oplus \Fermi_-(R) 
  }
The theory admits an action of the flavor symmetry group $\Sp(R)$. (Note that $R_1 \otimes R$ obtains a symmetric pairing from the tensor product of the symplectic pairings on~$R$ and~$R_1$.)
\end{dfn}

Exhibiting each of these theories as an $L_\infty$-module for the relevant supersymmetry algebra is the subject of the next subsection.

\subsection{The module structure}
\label{ssec:module}

The main goal of this section is to define an action of the $(2,0)$ supersymmetry algebra $\sp{2}$ on the tensor multiplet $\mplet{2}$. The action of the $(1,0)$ supersymmetry algebra on the constituent multiplets $\mplet{1}$ and~\hyper{R_1'} will then be obtained trivially by restriction, which we will spell out at the end of this section.

This action is only defined up to homotopy, which means we will give a description of $\mplet{2}$ as an $L_{\infty}$-{\em module} over~\sp{2}. This amounts to giving a Lorentz- and $R$-symmetry invariant $L_\infty$ action of the supertranslation algebra~$\st{2}$.  

Associated to the cochain complex $\cT_{(2,0)}$ is the dg Lie algebra of endomorphisms ${\rm End}(\cT_{(2,0)})$. 
Sitting inside of this dg Lie algebra is a sub dg Lie algebra consisting of linear differential operators ${\rm Diff}(\cT_{(2,0)}, \cT_{(2,0)})$. 
The differential is given by the commutator with the classical BV differential $Q_{\rm BV}$. 
For us, an $L_\infty$-action will mean a homotopy coherent map, or $L_\infty$ map, of dg Lie algebras $\rho : \fp_{(2,0)} \rightsquigarrow {\rm Diff}(\cT_{(2,0)}, \cT_{(2,0)})$. 

Such an $L_\infty$ map is encoded by the data of a sequence of polydifferential operators $\{\rho^{(j)}\}_{j \geq 1}$ of the form
\deq{
  \sum_{j\geq 1} \rho^{(j)} : \bigoplus_j \Sym^j \left( \st{2} [1] \right) \otimes \mplet{2} \to \mplet{2} [1] ,
}
satisfying a list of compatibilities. 
For instance, the failure for $\rho^{(1)} : \st{2} \otimes \mplet{2} \to \mplet{2}$ to define a Lie algebra action is by the homotopy $\rho^{(2)}$:
\beqn\label{eqn:rho2}
\rho^{(1)} (x) \rho^{(1)} (y) - \rho^{(1)} (y) \rho^{(1)}(x) - \rho^{(1)} ([x,y]) = [Q_{\rm BV} , \rho^{(2)}(x, y)] .
\eeqn

In the case at hand, $\rho^{(1)}$ will be given by the known supersymmetry transformations from the physics literature, extended to the remaining complex by the requirement that it preserve the shifted presymplectic structure.
While $\rho^{(1)}$ does not define a representation of~$\st{2}$, we can find $\rho^{(j)}$, $j \geq 2$ so as to define an $L_\infty$ module structure. 
In fact, we will see that $\rho^{(j)} = 0$ for $j \geq 3$, so we will only need to work out the quadratic term $\rho^{(2)}$. 

\begin{thm}\label{thm:Linfinity}
There are linear maps $\{\rho^{(1)}, \rho^{(2)}\}$ that define an $L_\infty$-action of $\st{2}$ on $\mplet{2}$. 
Furthermore, both $\rho^{(1)}$ and $\rho^{(2)}$ strictly preserve the $(-1)$-shifted presymplectic structure.
\end{thm}

We split the proof of this result into two steps.
First, we will construct the linear component $\rho^{(1)}$ and verify that it preserves the BV differential and shifted presymplectic form. 
Then we will define the quadratic homotopy $\rho^{(2)}$ and show that together with the linear term defines an $L_\infty$-module structure on $\mplet{2}$. 

\subsubsection{The physical transformations}\label{sec:physical}

We define the linear component $\rho^{(1)}$ of the action of supersymmetry on $\mplet{2}$. 
The map $\rho^{(1)}$ consists standard supersymmetry transformations on the physical fields (in cohomological degree zero), together with certain transformations on the antifields which guarantee that $\rho^{(1)}$ preserve the shifted presymplectic structure on $\mplet{2}$. 

The linear term $\rho^{(1)}$ is a sum of four components:
\begin{equation}
  \begin{aligned}
\rho_V  : & ~ V \otimes \mplet{2}  \to  \mplet{2} \\
\rho_\Fermi  : & ~ \Sigma_2 \otimes \Fermi_-(R_2)  \to  \thy_+(2) \oplus \Scalar(0,W) \\
\rho_\Scalar  : & ~ \Sigma_2 \otimes \Scalar(0,W)  \to  \Fermi_-(R_2) \\
\rho_\thy  : & ~ \Sigma_2 \otimes \thy_+(2)  \to  \Fermi_-(R_2) 
\end{aligned}
\end{equation}
We will define each of these component maps in turn. 

The first transformation is simply the action by (complexified) translations on the fields.
An translation invariant vector field $X \in V \subset \Vect(\RR^6)$ acts via the Lie derivative $L_X \alpha$, where $\alpha$ is any BV field.
That is, $\rho_V(X \otimes \alpha) = L_X \alpha$. 

We now turn to describe the supersymmetry transformations. 
We will first describe the action on the physical fields, that is, the fields in cohomological degree zero. 
We will deduce the action on the antifields in the next subsection. 

The transformation of the physical fermion field (the component $(\Pi S_- \otimes R_2 )$ of the BV complex $\Fermi_-(R_2)$ in degree zero) is given by $\rho_\Fermi$, which is defined as follows. 
Consider the isomorphism
\deq{
(\Pi S_+ \otimes R_2) \otimes (\Pi S_- \otimes R_2) \cong \left(\CC \oplus \wedge^2 V\right) \otimes \left(\CC \oplus W \oplus \Sym^2(R_2)\right) }
of $\Spin(6)\times\Sp(2)$ representations. It is clear by inspection that there are equivariant projection maps onto the irreducible representations $\wedge^2 V \otimes \C$ and~$\C \otimes W$. These projections allow us to define $\rho_\Fermi$ as the composition of the following sequence of maps:
\begin{equation}
  \begin{tikzcd}
    & & \Omega^0 \otimes W \arrow{d}{\subset} \\
    \Pi \Sigma_2 \otimes \Gamma(\Pi S_- \otimes R_2) \arrow["="]{r}{} \ar[dotted, bend left = 10, urr, "\rho_{\Psi,0}"] \ar[dotted, bend right = 10, drr, "\rho_{\Psi,2}"']  \ & (S_+ \otimes R_2) \otimes (S_- \otimes R_2)  \arrow[swap, two heads]{rd}{ } \arrow[two heads]{ru}{ } \arrow[dashed]{r}{ } & \thy_+(2) \oplus \Scalar(0,W). \\
    & & \Omega^2 \arrow[swap]{u}{\subset}
  \end{tikzcd}
\end{equation}
Of course, this map is canonically decomposed as the sum of two maps (along the direct sum in the target), which we will later refer to as $\rho_{\Fermi,0}$ and $\rho_{\Fermi,2}$ respectively. 

The transformation of the physical scalar field (the component $C^\infty (\RR^6 ; W)$ of the BV complex $\Phi(0,W)$ in degree zero) is defined as follows.
We observe that there is a map of $\Spin(6) \times \Sp(2)$ representations of the form
\deq{
\label{noname}
  (\Pi S_+ \otimes R_2) \otimes (V \otimes W) \to S_- \otimes R_2,
}
which can be thought of (using the accidental isomorphism $B_2 \cong C_2$) as the tensor product of the six- and five-dimensional Clifford multiplication maps. 
$\rho_\Phi^{(1)}$ can then be defined as the composition of the maps in the diagram
\begin{equation}
  \begin{tikzcd}
\Pi \Sigma_2 \otimes (\Omega^0 \otimes W)
    \arrow{r}{\d} & (\Pi S_+ \otimes R_2) \otimes (\Omega^1 \otimes W) \arrow[two heads]{d}{ } \\
 & \Gamma(\Pi S_- \otimes R_2) \arrow{r}{\subset} & \Fermi_-(R_2).
  \end{tikzcd}
\end{equation}
The vertical map is induced by (\ref{noname}).

On the degree zero component $\Omega^2 (\RR^4)$ of the presymplectic BV complex $\thy_+(2)$, the map $\rho_\thy$ is defined as follows.
Recall that there is a projection map of $\Spin(6)$ representations
\[
\pi : S_+ \otimes \wedge^3_- (V) \to S_-
\]
obtained via the isomorphism $\wedge^3_- (V) \otimes S_+ \cong S_- \oplus [012]$.\footnote{The notation refers to the Dynkin labels of type $D_3$.} This isomorphism is most easily seen using the accidental isomorphism with $\SU(4)$, where it can be derived using the standard rules for Young tableaux and takes the form
\deq{
  {\ydiagram{1,1}} \otimes {\ydiagram{1}} \cong  \text{\ydiagram{1,1,1}} \oplus \text{\ydiagram{2,1}}.
}
The map $\rho_\thy$ is then defined on physical fields by the following sequence of maps: 
\begin{equation}
  \begin{tikzcd}
\Pi \Sigma_2 \otimes \Omega^2 \arrow{r}{\d_-} & (\Pi S_+ \otimes R_2) \otimes \Omega^3_- \arrow[two heads]{d}{ }\\
 & \Gamma(\Pi S_- \otimes R_2) \arrow{r}{\subset}  & \Fermi_-(R_2)  .
\end{tikzcd}
\end{equation}

\subsubsection{Supersymmetry transformations on the anti-fields}
In the standard BV approach, there is a prescribed way to extend the linear action of any Lie algebra on the physical fields to an action on the BV complex in a way that preserves the shifted symplectic structure. 
The idea is that the action of a physical symmetry algebra $\lie{g}$ is usually defined by a map
\deq[eq:BRSTtrans]{
  \rho: \lie{g} \to \Vect(F)
}
that implements the physical symmetry transformations on the physical (BRST) fields, just as in the previous section. Of course there are strong conditions on~$\rho$ coming from, for example, the requirement of locality. In the BV formalism, there is additionally the requirement that the action of $\lie{g}$ on the BV fields must preserve the shifted symplectic structure. 
There is an immediate way to extend the vector fields~\eqref{eq:BRSTtrans} to \emph{symplectic} vector fields on the space $E = T^*[-1]F$ of BV fields: one can take the transformation laws of the antifields to be determined by the condition of preserving the shifted symplectic form. (In fact, such vector fields are always Hamiltonian in the standard case.) The induced transformations of the antifields are sometimes known as the \emph{anti-maps} of the original transformations, and we will denote them with the superscript $\rho^+$.


For the anti-map component of $\rho_\Scalar$, no complexity appears: we can simply define it as the composition
\begin{equation}
  \begin{tikzcd}
 & \Omega^6 \otimes W [-1] \arrow{r}{\subset} & \Scalar(0,W).\\
    \Pi \Sigma_2 \otimes \Gamma(\Pi S_+[-1] \otimes R_2) \arrow{r}{\dslash} & 
    (S_+ \otimes R_2) \otimes \Gamma(S_-[-1] \otimes R_2) \arrow[two heads]{u}{ }
  \end{tikzcd}
\end{equation}
The anti-map component of $\rho_{\Fermi,0}$ is similarly straightforward, and can be expressed with the diagram
\begin{equation}
  \begin{tikzcd}
    \Pi \Sigma_2 \otimes \Scalar(0,W) \arrow[two heads]{r}{ } & (\Pi S_+ \otimes R_2) \otimes (\Omega^6 \otimes W) [-1]
    \arrow{r}{\cong} & \Gamma(\Pi S_+[-1]\otimes R_2) \arrow{r}{\subset} & \Fermi_-(R_2).
  \end{tikzcd}
\end{equation}

The other two maps is determined by the nature of the shifted presymplectic pairing $\omega_{\chi_+}$ on~$\thy_+(2)$. As such, the number of derivatives appearing is, at first glance, somewhat surprising. The anti-map to $\rho_{\Fermi,2}$ takes the form
\begin{equation}
  \begin{tikzcd}
 & \Gamma(\Pi S_+[-1] \otimes R_2) \arrow{r}{\subset} & \Fermi_-(R_2).\\
   \Pi \Sigma_2 \otimes \Omega^3_+[-1] \arrow{r}{\d } & (\Pi S_+ \otimes R_2) \otimes \Omega^4 [-1] \arrow[two heads]{u}{ }
  \end{tikzcd}
\end{equation}
Finally, the anti-map component of $\rho_\thy$ takes the form
\begin{equation}
  \begin{tikzcd}
    \Pi \Sigma_2 \otimes \Gamma(S_+ [-1] \otimes R_2) \arrow["="]{r}{ } & 
    (\Pi S_+ \otimes R_2) \otimes \Gamma(\Pi S_+[-1]\otimes R_2) \arrow[two heads]{r}{ } & 
    \Omega^3_+[-1] \arrow{r}{\subset} & 
    \thy_+(2).
  \end{tikzcd}
\end{equation}

We have thus constructed the linear component of supersymmetry.
It is straightforward to check that $\rho^{(1)}$ commutes with the classical BV differential and preserves the $(-1)$-shifted presymplectic structure. 

\subsubsection{The $L_\infty$ terms}

We turn to the proof of the remaining part of Theorem \ref{thm:Linfinity}. 
We will show that $\rho^{(1)}$ sits as the linear component of an $L_\infty$-action of $\st{2}$ on the $(2,0)$ theory. 
In fact, we will only need to introduce a quadratic action term
\[
\rho^{(2)} : \st{2} \otimes \st{2} \otimes \mplet{2} \to \mplet{2} [-1] 
\]
and will show the following.
This quadratic term splits up into the following three components:
\begin{equation}
  \begin{aligned}
\rho^{(2)}_{\chi} : & ~ \st{2} \otimes \st{2} \otimes \chi_+ (2) \to \chi_+(2) [-1] \\
\rho^{(2)}_\Psi : & ~ \st{2} \otimes \st{2} \otimes \Psi_- (R_2) \to \Psi_- (R_2) [-1] \\
\rho^{(2)}_\Phi : & ~ \st{2} \otimes \st{2} \otimes \Phi(0,W) \to \chi_+ (2) [-1] .
\end{aligned}
\end{equation}

First, $\rho^{(2)}_\chi = \sum \rho^{(2)}_{\chi, j}$ is defined by the sum over form type of the linear maps
\[
\rho_{\chi,j}^{(2)} : \left(\Sigma_2 \otimes \Sigma_2\right) \otimes \Omega^{j} \xto{[\cdot, \cdot] \otimes 1} V \otimes \Omega^{j} \xto{i_{(\cdot)}} \Omega^{j-1}
\]
where $[\cdot, \cdot]$ is the Lie bracket defining the $(2,0)$ algebra and $i_{X}$ denotes contraction with the vector field $X$.

The next map, $\rho_{\Psi}^{(2)}$, acts on a fermion anti-field and produces a fermion field.
To define it, we introduce the following notation. 
Recall that $\wedge^2(S_+) \cong V$ as $\Spin(6)$-representations and $\wedge^2 R_2 \cong \CC \oplus W$ as $\Sp(2)$-representations. 
Thus, there is the following composition of $\Spin(6) \times \Sp(2)$-representations
\[
\star : \Sigma_2 \otimes \Sigma_2 \to (\wedge^2 S_+) \otimes (\wedge^2 R_2) \to V \otimes W .
\]
So, given $Q_1, Q_2 \in \Sigma_2$ the image of $Q_1 \otimes Q_2$ along this map is an element in $V \otimes W$ that we will denote by $Q_1 \star Q_2$. 
Now, we define $\rho_{\Psi}^{(2)}$ as the sum $\rho_{\Psi,0}^{(2)} + \rho_{\Psi, 2}^{(2)}$ where $\rho_{\Psi,0}^{(2)}$ is the composition
\[
\rho_{\Psi,0}^{(2)} : \left(\Sigma_2 \otimes \Sigma_2 \right) \otimes \Gamma(S_+ \otimes R_2) \xto{\star \otimes 1} (V \otimes W) \otimes \Gamma(S_+ \otimes R_2) \to \Gamma(S_- \otimes R_2)
\]
where the second arrow is the map of $\Spin(6) \times \Sp(2)$-representations in (\ref{noname}). 
Next, $\rho_{\Psi,2}^{(2)}$ is defined by the composition
\[
\rho_{\Psi,2}^{(2)} : \left(\Sigma_2 \otimes \Sigma_2 \right) \otimes \Gamma(S_+ \otimes R_2) \xto{[\cdot, \cdot] \otimes 1} V \otimes \Gamma(S_+ \otimes R_2) \to \Gamma(S_- \otimes R_2)
\]
where the last map is Clifford multiplication.

Finally, the map $\rho^{(2)}_{\Phi}$ acts on a scalar field and produces a ghost one-form in $\chi_+(2)$. 
Using the map $\star$ above, $\rho^{(2)}_{\Phi}$ is described by the composition
\[
\left(\Sigma_2 \otimes \Sigma_2 \right) \otimes \left(\Omega^0 \otimes W \right) \xto{\star} \Omega^1 \otimes (W \otimes W) \to \Omega^1 
\]
where the last map utilizes the symmetric form on $W$. 

To finish the proof of Theorem \ref{thm:Linfinity} we must show that $\rho^{(1)}$ and $\rho^{(2)}$ satisfy~\eqref{eqn:rho2} for all $x,y \in \st{2}$.

   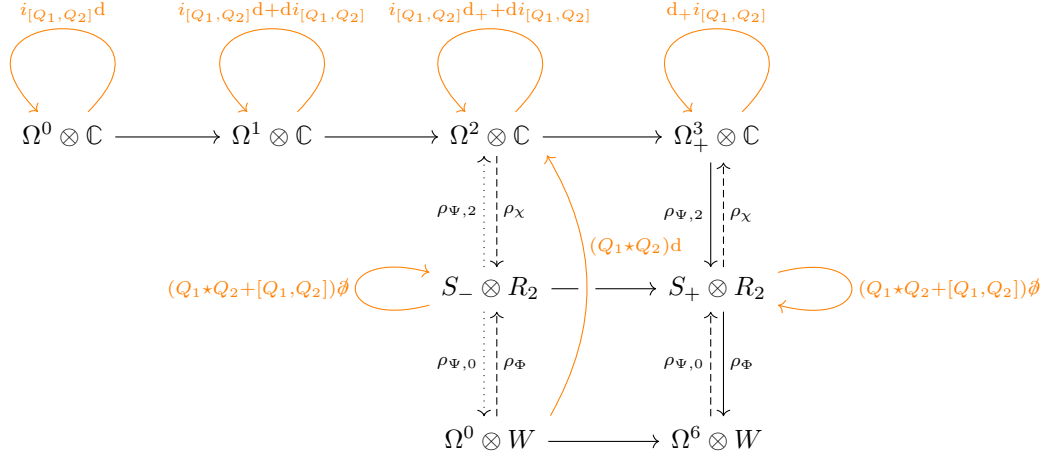
\begin{figure}
    \[
      \begin{tikzcd}[row sep = 4 em, column sep = 4 em]
      \Omega^0 \otimes \C \ar[r] \arrow[loop , orange]{u}[above]{i_{[Q_1,Q_2]} \d} & \Omega^1 \otimes\C \ar[r] \arrow[loop , orange]{u}[above]{i_{[Q_1,Q_2]} \d + \d i_{[Q_1,Q_2]}} & \Omega^2 \otimes\C \ar[r] \arrow[shift left,dashed]{d}{\rho_\thy} 
      \arrow[loop , orange]{u}[above]{i_{[Q_1,Q_2]} \d_+ + \d i_{[Q_1,Q_2]}}
      & \Omega^3_+ \otimes\C  \arrow[shift right]{d}[left]{\rho_{\Fermi,2}} 
      \arrow[loop , orange]{u}[above]{\d_+ i_{[Q_1,Q_2]} }
      \\
      & & \arrow[loop left, orange]{d}[left]{(Q_1\star Q_2 + [Q_1,Q_2]) \dslash} S_- \otimes R_2 \ar[r] \arrow[dotted, shift left]{u}{\rho_{\Fermi,2}} \arrow[dotted, shift right]{d}[left]{\rho_{\Fermi,0}}  & S_+ \otimes R_2 \arrow[loop right, orange]{u}[right]{(Q_1 \star Q_2 + [Q_1,Q_2])\dslash} \arrow[shift right, dashed]{u}[right]{\rho_\thy} \arrow[shift left]{d}[right]{\rho_\Scalar}\\
      & & \Omega^0 \otimes W \arrow[shift right, dashed]{u}[right]{\rho_\Scalar}\ar[r] 
      \arrow[bend right , shift right = 4 ex, orange, crossing over]{uu}[right, pos=0.65]{(Q_1 \star Q_2) \d} 
      & \Omega^6 \otimes W \arrow[shift left, dashed]{u}{\rho_{\Fermi,0}}
    \end{tikzcd}
    \]
    \caption{The failure of $\rho^{(1)}$ to be a Lie map.} \label{fig:mu}
  \end{figure}

It will be convenient to define the following linear map.
  \deq{
    \begin{aligned}
\mu:  \st{2} \otimes \st{2} \otimes \mplet{2} & \to \mplet{2}, \\
    x \otimes  y \otimes f &\mapsto \rho^{(1)} ([x,y], f)  - \rho^{(1)}(x, \rho^{(1)}(y, f))  \pm  \rho^{(1)}(y, \rho^{(1)}(x, f))
  \end{aligned}
    }
This map $\mu$ represents the failure of $\rho^{(1)}$ to define a strict Lie algebra action. 
In terms of $\mu$, \eqref{eqn:rho2} simply reads 
\beqn\label{eqn:master}
[Q_{\rm BV} , \rho^{(2)} (x,y)] = \mu(x,y) .
\eeqn
We have represented $\mu$ via the orange arrows in Figure \ref{fig:mu}.
In this figure, the dashed and dotted arrows denote the action of $Q_1$ and $Q_2$ through the linear term $\rho^{(1)}$.

It is sufficient to consider the case when $x = Q_1, y = Q_2 \in \Sigma_2$.
We observe that the first term in $\mu$ simply produces the Lie derivative of any field in the direction $[Q_1, Q_2]$. 
Since $\mu$ is an even degree-zero map, we can consider each degree and parity separately, beginning with the ghosts: here, it is easy to see that 
    \deq{
      \begin{aligned}
      \mu (Q_1,Q_2) |_{\Omega^0} &= \Lie_{[Q_1,Q_2]}:  \Omega^0[2] \to \Omega^0[2], \\
      \mu (Q_1,Q_2)|_{\Omega^1} &= \Lie_{[Q_1,Q_2]}: \Omega^1[1] \to \Omega^1[1],
    \end{aligned}
  }
since the supersymmetry variations make no contribution. We next work out the action of $\mu$ on the two-form field, which is given by 
  \deq{
    \mu (Q_1,Q_2)|_{\Omega^2} = \Lie_{[Q_1,Q_2]} - \rho_\Fermi(Q_1) \circ \rho_\thy(Q_2) - \rho_\Fermi(Q_2) \circ \rho_\thy(Q_1)
  }
which is a map of the form $\Omega^2 \to \Omega^2 \oplus (\Omega^0\otimes W) \subset \Phi (0,W)$. 
The map must be symmetric in the two factors of~$\Sigma_2$; since $\Omega^2$ is neutral under $\Sp(2)$ $R$-symmetry, the only possible contractions of $(R_2)^{\otimes 2}$ land in the trivial representation or in~$W$, and both are antisymmetric. So the pairing on $(\Pi S_+)^{\otimes 2}$ must also be antisymmetric, showing that 
  \deq{
    \mu (Q_1,Q_2)|_{\Omega^2}  = \Lie_{[Q_1,Q_2]} - i_{[Q_1,Q_2]} \d_- = \d i_{[Q_1,Q_2]} + i_{[Q_1,Q_2]} \d_+.
  }
  In degree one, there is also a unique equivariant map that can contribute: it is not difficult to show that 
  \deq{
    \mu (Q_1,Q_2)|_{\Omega^{3}_+} = \Lie_{[Q_1,Q_2]} - \pi_+ i_{[Q_1,Q_2]} \d.
  }
 Since $[Q_1,Q_2]$ is a constant vector field, the Lie derivative preserves the self-duality condition; from this, it follows via Cartan's formula that the anti-self-dual part of $i_{[Q_1,Q_2]} \d$ is equal to $\d_- i_{[Q_1,Q_2]}$, so that 
  \deq{
    \mu (Q_1,Q_2)|_{\Omega^3_+} = \d_+ i_{[Q_1,Q_2]}.
  }

Similar arguments apply for the component of $\mu$ acting on the scalar field. 
One can check that the restriction of $\mu$ to the scalar field is of the form
  \deq{ \label{muscalar}
\mu (Q_1,Q_2)|_{\Omega^0 \otimes W} : \Omega^0 \otimes W \to \Omega^2 \subset \chi_+(2) .
}
The diagonal term in $\mu$ restricted to $\Omega^0 \otimes W$ is seen to vanish upon applying Cartan's magic formula. 
The same argument shows that $\mu_{1,\Scalar}$ also vanishes.

The component (\ref{muscalar}) comes from a contraction of the supersymmetry generators with the de Rham differential acting on the scalar. 
There is precisely one such map, which takes the form 
 \deq{
 \mu (Q_1,Q_2) |_{\Omega^0 \otimes W} = \d \circ (Q_1 \star Q_2, \cdot)_W }
 where $(\cdot , \cdot)_W$ is the symmetric form on $W$.

Finally, the component of $\mu$ acting on $\Psi_-(R_2)$ maps a fermion to itself and a fermion anti-field to itself. 
For the fermion field, the restriction of $\mu$ is given as a sum of two terms
\[
\mu (Q_1,Q_2)|_{\Gamma(S_+ \otimes R_2)} = \mu_{\Psi,0}(Q_1,Q_2) + \mu_{\Psi,2} (Q_1,Q_2)
\]
where
$\mu_{\Psi,0}$ is given by the composition
\beqn\label{mufermi0}
\mu_{\Psi,0} : \Gamma(S_- \otimes R_2) \xto{Q_1 \star Q_2} \Gamma(S_+ \otimes R_2)  \xto{\dslash}  \Gamma(S_- \otimes R_2) 
\eeqn
and $\mu_{\Psi,2}$ is given by the composition
\beqn\label{mufermi2}
\mu_{\Psi,2} : \Gamma(S_- \otimes R_2) \xto{[Q_1,Q_2]} \Gamma(S_+ \otimes R_2) \xto{\dslash} \Gamma(S_- \otimes R_2) .
\eeqn
The action of $\mu(Q_1, Q_2)$ on the anti-fermion fields is completely analogous. 

We proceed to verify~\eqref{eqn:master}.
For the restriction of $\mu(Q_1,Q_2)$ to $\chi_+(2)$ the equation follows from repeated use of Cartan's formula. 

Next, the restriction of $\mu(Q_1,Q_2)$ to the scalar is given by (\ref{muscalar}). 
The restriction of the left-hand side of (\ref{eqn:master}) to the scalar is 
\[
Q_{\rm BV} \circ \rho^{(2)}_\Phi (Q_1, Q_2) = \d_{\Omega^1 \to \Omega^2} \circ (Q_1 \star Q_2, \cdot)_W
\]
as desired. 

Finally, the restriction of $\mu(Q_1,Q_2)$ to the fermion is given by the sum of (\ref{mufermi0}) and (\ref{mufermi2}). 
The left-hand side of (\ref{eqn:master}) also splits into two pieces.
Note that 
\[
[Q_{\rm BV} , \rho^{(2)}_{\Psi,0} (Q_1, Q_2)] = \dslash \circ \left(Q_1 \star Q_2 \cdot (\cdot) \right)
\]
which is precisely $\mu_{\Psi, 0}(Q_1,Q,2)$ acting on $\Psi_-(R_2)$.

The other term is
\[
[Q_{\rm BV} , \rho^{(2)}_{\Psi,0} (Q_1, Q_2)] = \dslash \circ \left(Q_1 \star Q_2 \cdot (\cdot) \right)
\]
which is precisely $\mu_{\Psi, 2}(Q_1,Q,2)$ acting on $\Psi_-(R_2)$.

We conclude by noting the following result:
\begin{prop} 
  \label{prop:decompose(2,0)}
  With respect to a fixed $\N=(1,0)$ subalgebra of~$\sp{2}$, the abelian tensor multiplet decomposes as 
  \deq{
    \cT_{(2,0)} \cong \cT_{(1,0)} \oplus \cT^\text{hyp}_{(1,0)}(R_1').
  }
\end{prop}
\begin{proof}
  At the level of field content, the statements reduce to simple representation-theoretic facts: under the subgroup $\Sp(1) \times \Sp(1') \subseteq \Sp(2)$, the vector and spinor representations decompose as 
  \deq{
    W \cong (R_1 \otimes R_1') \oplus \C, \quad R_2 \cong R_1 \oplus R_1'
  }
  respectively. (Here $\Sp(1)$ denotes the $R$-symmetry of~$\sp{1}$, and $\Sp(1)'$ its commutant inside of the $(2,0)$ $R$-symmetry.)

  The $L_\infty$ module structure for $\sp{2}$ obviously restricts to an $L_\infty$ module structure for $\sp{1}$, and it is trivial to see that the resulting module structure extends the physical $\N=(1,0)$ transformations. (At the level of physical transformations, the proposition is standard.) 
\end{proof}

\section{The minimal twists} \label{sec:holtwist}

In this section we will compute the holomorphic twist of the abelian $\N=(1,0)$ and $(2,0)$ tensor multiplets, using the formulation and supersymmetry action developed in the preceding sections. 
We will begin by placing the theory on a K\"ahler manifold and decomposing the fields with respect to the K\"ahler structure; at the level of representation theory, this corresponds to recalling the branching rules from~$\SO(6)$ to~$\U(3)$ (more precisely, at the level of the double covers $\MU(3) \hookrightarrow \Spin(6)$), followed by a regrading. 

We will then deform the differential by a compatible holomorphic supercharge. (As is usual in twist calculations, choices of holomorphic supercharge are in one-to-one correspondence with choices of complex structure on~$\R^6$.) Since the $L_\infty$ module structure worked out in the previous section preserves the presymplectic structure, we are guaranteed that the twisted theory $\cT^Q_{(1,0)}$ is a well-defined \psBV{} theory after deforming the differential. 

One subtlety appears when we attempt to simplify the resulting theory by discarding acyclic portions of the BV complex. 
There is a natural quasi-isomorphism of chain complexes of the form 
\deq{
  \Phi: \cT^Q_{(1,0)}  \to \chi(2),
}
whose kernel consists of an acyclic subcomplex of~$\cT^Q_{(1,0)}$. However, $\Phi$ does \emph{not} respect the presymplectic structure on $\cT^Q_{(1,0)}$ in a naive fashion!

Given a general quasi-isomorphism 
\deq{
\Phi: \cT \to \cT'
}
of cochain complexes underlying some presymplectic BV theories, the appropriate notion of compatibility is to ask that the two shifted presymplectic forms are \emph{equivalent} in the larger theory; in other words, that 
\deq{
  \Phi^* \omega' - \omega = [\QBV, h].
}
Here $h$ is a degree-$(-2)$ element in the space of symplectic structures, witnessing a homotopy between the two $(-1)$-shifted structures. In other words, we should not require that the difference of the symplectic structures vanish strictly, but only that it be $Q_{\rm BV}$-exact.

\begin{rmk} Indeed, suppose $\ker(\Phi)$ is a nondegenerate BV theory whose differential is acyclic. For instance,  suppose the complex of fields is of the form 
\deq{ T^*[-1]\left( V \oplus V[-1]\right) \cong (V \oplus V[-1]) \oplus (V^\vee[-1] \oplus V^\vee),
  }
  where $V$ is some chain complex of vector bundles, and the acyclic differential is the shift morphism between the two copies of $V$ and its anti-map between the two copies of~$V^\vee$. Now, the symplectic form pairs $V$ with~$V^\vee[-1]$ and $V[-1]$ with~$V^\vee$; there is an obvious nullhomotopy given by the degree-$(-2)$ pairing that pairs $V[-1]$ with~$V^\vee[-1]$, which witnesses the equivalence of this pairing with the zero pairing. 
\end{rmk}

In our case, the kernel of $\Phi$ pairs nontrivially with the rest of~$\cT^Q_{(1,0)}$, so that the homotopy equivalence plays an essential role in determining the appropriate \psBV{} structure on the holomorphic theory. 
With this in mind, we demonstrate an equivalence between $\cT^Q_{(1,0)}$ and~$\chi(2)$, not just as chain complexes, but as \psBV{} theories. Of course, an identical phenomenon occurs in the holomorphic twist of the $(2,0)$ multiplet, which can be thought of as one $(1,0)$ tensor multiplet and one $(1,0)$ hypermultiplet. 
Together, these rigorous twist computations are our main result in this section, which we state precisely as follows: 

\begin{thm}
\label{thm:twist}
Let $Q$ be a rank one supercharge in either of the supersymmetry algebras $\sp{1}$ or $\sp{2}$.
The respective twists of the abelian $(1,0)$ and~$(2,0)$ tensor multiplets on $\CC^3$ are as follows.
\begin{itemize}
\item[$\bu$ {\bf (1,0)}] The holomorphic twist $\mplet{1}^Q$ is equivalent to the $\ZZ$-graded \psBV{} theory of the chiral $2$-form:
\[
\mplet{1}^Q \; \simeq \; \chi(2) .
\]
\item[$\bu$ {\bf (2,0)}] The holomorphic twist $\mplet{2}^Q$ is equivalent to the $\ZZ \times \ZZ/2$-graded \psBV{} theory defined by the chiral $2$-form plus abelian holomorphic Chern--Simons theory with values in the odd symplectic vector space $\Pi R_1'$: 
\deq{
\mplet{2}^Q \; \simeq \; \chi(2) \oplus \hCS (\Pi R_1') .
}
Moreover, this equivalence is $\Sp(1)'$-equivariant. 
\end{itemize}
\end{thm}

The remainder of the section is devoted to a detailed proof. We start off with some reminders about the general yoga of twisting.

\subsection{Supersymmetric twisting}
\label{sec:twisting}

In this section we briefly recall the procedure of twisting a supersymmetric field theory. 
For a more complete formulation see \cite{CostelloHolomorphic, ESW}, though we modify the construction very slightly (see \cite[Remark 2.19]{ESW}).  
As we've already mentioned, the key piece of data is that of a square-zero supercharge $Q$.
Roughly, the twisted theory is given by deforming the classical BV operator $Q_{\rm BV}$ by $Q$. 

In the cited references, the twisting procedure is performed starting with the data of a supersymmetric theory in the BV formalism. 
This means that one starts with the data of a classical theory in the BV formalism together with an ($L_\infty$) action by the super Lie algebra of supertranslations. 
In our context, we have exhibited an $L_\infty$ action of the supersymmetry algebra on a \psBV{} theory, which acts compatibly with the $(-1)$-shifted presymplectic structure. The twisted theory will therefore also have the structure of a \psBV{} theory, which descends to smaller quasi-isomorphic descriptions after attending to the subtlety alluded to in the introduction. 

%
%
%

Classical supersymmetric theories are (at least) $\ZZ \times \ZZ/2$-graded, where the first grading is the cohomological degree and the second grading is the parity. 
By definition, a square-zero supercharge $Q$ is of bidegree $(0,1)$, whereas the classical BV differential is of bidegree $(1,0)$. 
After making the deformation $Q_{\rm BV} \rightsquigarrow Q_{\rm BV} + Q$, one could choose to remember just the totalized $\ZZ/2$ grading, with respect to which both operators are odd. 

Instead, one typically uses additional data to {\em regrade} the theory so that $Q, Q_{\rm BV}$ have the same homogenous degree.
In addition to the action by supertranslations, a classical supersymmetric theory on~$\R^d$ carries an action by the Lorentz group $\Spin(d)$. It also often carries an action by the $R$-symmetry group $G_R$, which is the set of automorphisms of $\Pi \Sigma_n$ preserving the pairing. 
For us, $d=6$ and $G_R = \Sp(n)$ for $\cN=(n,0)$ supersymmetry. 

In order to perform the twist, we use the $R$-symmetry to define a consistent graded structure, as well as to ensure that the twisted theory is well-defined not just on affine space, but on all manifolds with appropriate holonomy group. To do this, we use two additional pieces of data, which we now describe in turn.

\begin{dfn}
  Given a square-zero supercharge $Q$, a \defterm{regrading homomorphism} is a homomorphism~$\alpha : \U(1) \to G_R$ such that the weight of $Q$ under $\alpha$ is $+1$. 
\end{dfn}

Suppose $\cE = (\cE, Q_{\rm BV})$ is the cochain complex of fields of the classical theory, and for $\varphi \in \cE$, denote by $|\varphi| = (p, q \bmod 2) \in \ZZ \times \ZZ/2$ the bigrading.
Given a regrading homomorphism $\alpha$, we define a new $\ZZ \times \ZZ/2$-graded cochain complex of fields $\cE^\alpha = (\cE^{\alpha}, Q_{\rm BV})$ which agrees with $(\cE, Q_{\rm BV})$ as a totalized $\ZZ/2$-graded cochain complex with new bigrading
\[
|\varphi|_{\alpha} = |\varphi| + (\alpha(\varphi), \alpha(\varphi) \bmod 2) \in \ZZ \times \ZZ/2 
\]
where $\alpha(\varphi)$ denotes the weight of the field $\varphi \in \cE$ under $\alpha$. 
Note that $Q_{\rm BV}$ and $Q$ are both of bidegree $(1,0)$ as operators acting on the regraded fields $\cE^\alpha$. 
Our convention is that $\cE^{\alpha}$ denotes the cochain complex of fields that are regraded, but equipped with the original BV differential $Q_{\rm BV}$. 
The shifted (pre) symplectic structure remains unchanged. 

There is one last step before performing the deformation of the classical differential by the supercharge $Q$ in the regraded theory.
In general, the symmetry group $\Spin(n) \times G_R$ will no longer act on the deformed theory since $Q$ is generally not invariant under this group action. 
\begin{dfn}
Let $Q$ be a square-zero supercharge, and suppose $\iota : G \to \Spin(n)$ is a group homomorphism.
A \defterm{twisting homomorphism} (relative to $\iota$) is a homomorphism $\phi : G \to G_R$ such that $Q$ is preserved under the product $\iota \times \phi : G \to \Spin(n) \times G_R$. 
\end{dfn}

Given such a $\phi$, we can restrict the regraded theory to a representation for the group $G$, which we will denote by $\phi^* \Tilde{\cE}^\alpha$.
We will refer to as the $G$-regraded theory.  

Given a square-zero supercharge $Q$, a regrading homomorphism $\alpha$, and twisting homomorphism $\phi$ we can finally define a twist of a supersymmetric theory $\cE$. 
It is the $\ZZ \times \ZZ/2$-graded theory whose underlying cochain complex of fields is 
\[
\cE^Q = \left(\phi^*\cE^\alpha , Q_{\rm BV} + Q\right)  .
\]

\subsection{Holomorphic decomposition}

Throughout the rest of this section we fix the data of a rank-one supercharge $Q \in \Sigma_1$ (which is automatically square-zero in $\sp{1}$), and characterize the resulting twist of the $(1,0)$ tensor multiplet $\cT_{(1,0)}$.
As discussed in~\S\ref{ssec:nilps}, such a $Q$ defines a theory with three invariant directions, so we will refer to the twist as holomorphic.  
In addition to $Q$, to perform the twist we must prescribe a compatible pair of a twisting homomorphism $\phi$ and regrading homomorphism $\alpha$.

Geometrically, the supercharge $Q$ defines a complex structure $L = \CC^3 \subset V = \CC^6$ equipped with the choice of a holomorphic half-density on $L$. 

Under the subgroup $\MU(3) \subset \Spin(6)$, the spin representations decompose as
\deq{
\label{eqn:decompose1}
S_+ = \det(L)^{\frac12} \oplus L \otimes \det(L)^{-\frac12} \;\; , \;\; S_- = \det(L)^{-\frac12} \oplus L^* \otimes \det(L)^{\frac12} .
}
`
In particular, the odd part $\Sigma_1 = S_+ \otimes R_1$ of the super Lie algebra $\sp{1}$ decomposes under $\MU(3)$ as
\[
\det(L)^{\frac12} \otimes R_1 \oplus L \otimes \det(L)^{-\frac12} \otimes R_1 .
\]
The holomorphic supercharge $Q$ lies in the first factor. 

There exists a unique embedding $\U(1) \subset G_R = \Sp(1)$ under which $Q$ has weight $+1$. 
The twisting homomorphism is defined by the composition 
\[
\phi : \MU(3) \xto{{\rm det}^{\frac12}} \U(1) \hookrightarrow \Sp(1) .
\]
Under this twisting homomorphism,
the defining representation $R_1$ of $\Sp(1)$ splits as
\deq{
\label{eqn:decompose2}
  R_1 = \det(L)^{-\frac12} \oplus \det(L)^{\frac12} .
}
Additionally, we fix the regrading homomorphism to agree with the natural inclusion above: 
\[
\alpha : \U(1) \hookrightarrow \Sp(1)  .
\]
As outlined in \S \ref{sec:twisting}, the data of $\phi$ and $\alpha$ allow us to consider the $G = \MU(3)$-regraded theory.

We observe that the odd part $\Sigma_1$ of $\sp{1}$ decomposes under these twisting data as
\begin{equation}
\begin{tikzcd} [row sep = -2pt]
& -1 & 0 & 1 \\ [10pt] \hline \\ [10pt]
-2 & & & L \otimes \det(L)^{-1} \\
0 & & & \CC \cdot Q \\
1 & L & & \\
3 & \det(L) & & 
\end{tikzcd} 
\label{eqn:10decomp}
\end{equation}
Here, the horizontal grading is by the ghost $\ZZ$-degree determined by $\alpha$ and the vertical grading is by spin $\U(1) \subset \MU(3)$. 
Note that $Q$ lives in a scalar summand of ghost degree $+1$.

The decomposition of the $(1,0)$ tensor multiplet with respect to the twisting data is described in the following proposition.

\begin{prop}
\label{prop:regraded}
The $\MU(3)$-regraded $(1,0)$ tensor multiplet $\phi^* \cT_{(1,0)}^\alpha$ decomposes as
\[
  \phi^* \cT_{(1,0)}^\alpha= \thy_+(2) \oplus \Fermi^\alpha_- (R_1) \oplus \Scalar(0,\CC) .
\]
The result is depicted in Figure~\ref{fig:(1,0)regraded}.
\end{prop}

Notice that the $\MU(3)$-action descends to a $\U(3)$-action, so without confusion we will refer $\phi^* \cT^\alpha_{(1,0)}$ as the $\U(3)$-regraded theory. 

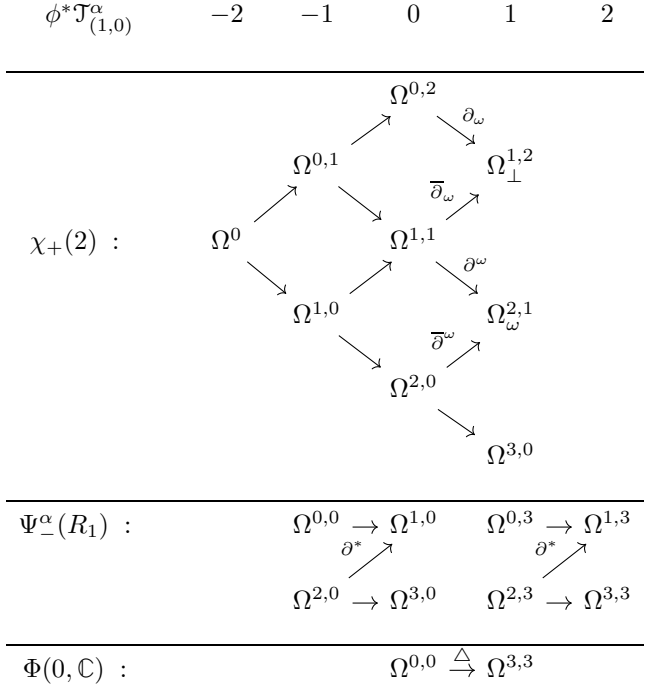
\begin{figure}
  \begin{center}
    \[
      \begin{tikzcd}[row sep = 1 em, column sep = 1 em]
       \phi^* \cT_{(1,0)}^\alpha & -2 & -1 & 0 & 1 & 2 \\ \hline
        & & & \Omega^{0,2} \ar[rd, "\partial_\omega"]  & \\
        & & \Omega^{0,1} \ar[ru] \ar[rd]  & & \Omega^{1,2}_\perp \\
     \thy_+(2) \; : \;\;\;\; &  \Omega^{0} \ar[ru] \ar[rd]  & & \Omega^{1,1} \ar[ru,"\dbar_\omega"] \ar[rd, "\partial^\omega"]  & &  \\
     & & \Omega^{1,0} \ar[ru] \ar[rd]  & & \Omega^{2,1}_\omega  \\ 
        & & & \Omega^{2,0} \ar[ru, "\dbar^\omega"] \ar[rd]  & \\
        & & & & \Omega^{3,0}  \\  \hline
    \Fermi^\alpha_-(R_1) \; : \;\;\;\; &  & \Omega^{0,0} \ar[r] & \Omega^{1,0} & \Omega^{0,3} \ar[r] & \Omega^{1,3} \\
       & & \Omega^{2,0} \ar[r] \ar[ru, "\partial^*"] & \Omega^{3,0} & \Omega^{2,3} \ar[r] \ar[ru, "\partial^*"] & \Omega^{3,3} \\ \hline
       \Scalar(0,\CC) \; : \;\;\;\; & & & \Omega^{0,0} \ar[r, "\triangle"] & \Omega^{3,3} &
      \end{tikzcd}
    \]
  \end{center}
  \caption{The regraded $\N=(1,0)$ tensor multiplet.
The unlabeled arrows denote the obvious $\partial$ or $\dbar$ operators.}
  \label{fig:(1,0)regraded}
\end{figure}

%
\begin{proof}[Proof of Proposition~\ref{prop:regraded}]
The components $\thy_+(2)$ and $\Scalar(0,\CC)$ of $\mplet{1}$ are acted on trivially by the $R$-symmetry group $G_R = \Sp(1)$, so we only need to focus on how $\Fermi_-(R_1)$ is regraded. 
According to Equations (\ref{eqn:decompose1}) and (\ref{eqn:decompose2}), the physical fields decompose under the twisting homomorphism $\phi$ by:
\deq{
\label{eqn:decompose3}
\Pi \left( \Omega^0 \otimes S_- \otimes R_1\right) = \Pi (\Omega^{0,0} \oplus \Omega^{2,0}) \oplus \Pi(\Omega^{0,3} \oplus \Omega^{2,3}) .
}
Similarly, the antifields decompose as
\deq{
\label{eqn:decompose4}
\Pi \left(\Omega^0 \otimes S_+ \otimes R_1\right) [-1] = \Pi (\Omega^{3,3} \oplus \Omega^{1,3}) [-1] \oplus \Pi(\Omega^{3,0} \oplus \Omega^{1,0}) [-1]
}

The next step is to regrade the fields according to the homomorphism $\alpha : \U(1) \hookrightarrow \Sp(1) = G_R$. 
At the level of the decomposed fields in Equation (\ref{eqn:decompose3}), this $\U(1)$ acts by weight $-1$ on the first summand $\Omega^{0,0} \oplus \Omega^{2,0}$, and by weight $+1$ on the second summand $\Omega^{0,3} \oplus \Omega^{2,3}$. 
Thus, we see that the regraded fields of Equation (\ref{eqn:decompose3}) become
\[
(\Omega^{0,0} \oplus \Omega^{2,0})[1] \oplus (\Omega^{0,3} \oplus \Omega^{2,3}) [-1]  .
\]
Similarly, the regraded anti-fields of Equation (\ref{eqn:decompose4}) become
\[
(\Omega^{3,3} \oplus \Omega^{1,3}) [-2] \oplus (\Omega^{3,0} \oplus \Omega^{1,0}) .
\]

It remains to identify the linear BV operator $Q_{\rm BV}$ in the regraded theory.
This follows from the well-known decomposition of the Dirac operator, on a K\"ahler manifold:
\begin{equation}
  \left(
  \begin{tikzcd}
 S_- \otimes R_1  \arrow{r}{\dslash} & S_+ \otimes R_1 
 \end{tikzcd}
 \right)
 \cong 
 \left( 
  \begin{tikzcd}  
\Omega^{0,0} \otimes K^{-\frac12} \otimes R_1 \arrow{r}{\partial} & \Omega^{1,0} \otimes K^{-\frac12} \otimes R_1 \\
 \Omega^{2,0} \otimes K^{-\frac12} \otimes R_1 \arrow{r}{\partial} \arrow{ru}{\partial^*} & \Omega^{3,0} \otimes K^{-\frac12} \otimes R_1 .
  \end{tikzcd}
\right)
\end{equation}

The components $\thy_+(2), \Scalar(0,\CC)$ remain unaffected by both the twisting homomorphism $\phi$ and regrading homomorphism $\alpha$. However, it is necessary in what follows to decompose these cochain complexes as $\U(3)$-representations, using information about the decomposition of the de Rham forms on a K\"ahler manifold. We recall that multiplication by the K\"ahler form determines a cochain map of degree $(1,1)$, defining a space of ``non-primitive'' forms. In what follows, we will fix a splitting into primitive and non-primitive forms in each degree, so that (for example) 
\deq{
  \Omega^{2,1} = \Omega^{2,1}_\perp \oplus \Omega^{2,1}_\omega,
}
with the latter summand being the image of $\Omega^{1,0}$ under the K\"ahler form. (Such a splitting is of course determined on compactly supported forms by the choice of K\"ahler metric.) We correspondingly decompose the $\del$ and $\dbar$ operators with respect to this splitting; we will sometimes use the subscript $\del_\omega$ to indicate a projection onto nonprimitive forms, and the superscript $\del^\omega$ for projection onto primitive forms. Verifying the isomorphism
\deq{
  \Omega^3_+ \cong \Omega^{3,0} \oplus \Omega^{2,1}_\omega \oplus \Omega^{1,2}_\perp
}
is then a straightforward representation-theoretic exercise.
\end{proof}


In Table \ref{table:regrade} we have summarized what happens to the physical fields (cohomological degree zero in the original theory) of the $(1,0)$ tensor multiplet in the regraded theory. 

\begin{table}
\begin{tabular}{c| c | c }
& $\mplet{1}$ & $\phi^* \cT_{(1,0)}^\alpha$ \\
\hline
& & \\ 
$\thy_+(2)$ & $\beta \in \Omega^2$ & $\beta \in \Omega^2 = \Omega^{2,0} \oplus \Omega^{1,1} \oplus \Omega^{0,2}$ \\
& & \\
$\Fermi_-(R_1)$ & $\psi_- \in \Pi (S_- \otimes R_1)$ & $\psi_- \in (\Omega^0 \oplus \Omega^{2,0})[1] \oplus (\Omega^{0,3} \oplus \Omega^{2,3})[-1]$ \\
& & \\
$\Scalar(0, \CC)$ & $ \phi \in \Omega^0 \otimes \C$ & $\phi \in \Omega^0 \otimes \C $ \\
& & \\
\end{tabular}
\smallskip
  \caption{The physical fields in the regraded $(1,0)$ theory.}
  \label{table:regrade}
\end{table}

%
%
\subsection{Proof of (1,0) part of Theorem \ref{thm:twist}}\label{sec:(1,0)twist}
The proof proceeds in two steps. In the first, we use the holomorphic decomposition discussed above, and deform the theory by the holomorphic supercharge to obtain a description of the twist $\cT^Q_{(1,0)}$. In the second, we give an explicit projection map which defines a quasi-isomorphism onto $\thy(2)$, and check that it defines an equivalence of \psBV{} theories. 

\subsubsection{Calculation of~$\cT^Q_{(1,0)}$}
Throughout this section, we refer to Figure~\ref{fig:(1,0)twist}, which uses the decomposition of the fields we found in the previous section and shows the additional differentials generated by the holomorphic supercharge. 
The black text denotes the fields in the component $\thy_+(2)$ of the tensor multiplet.
The red text denotes the fields in the $\Psi^\alpha_-(R_1)$ component, as in Proposition \ref{prop:regraded}.
Finally, the green text denotes the fields in the $\Phi(0,\CC)$ component. 
Each of the solid lines denotes the linear BV differential in the original, untwisted theory, see Figure \ref{fig:(1,0)regraded}.
We will use superscripts to label the components of each field by their form degree. 
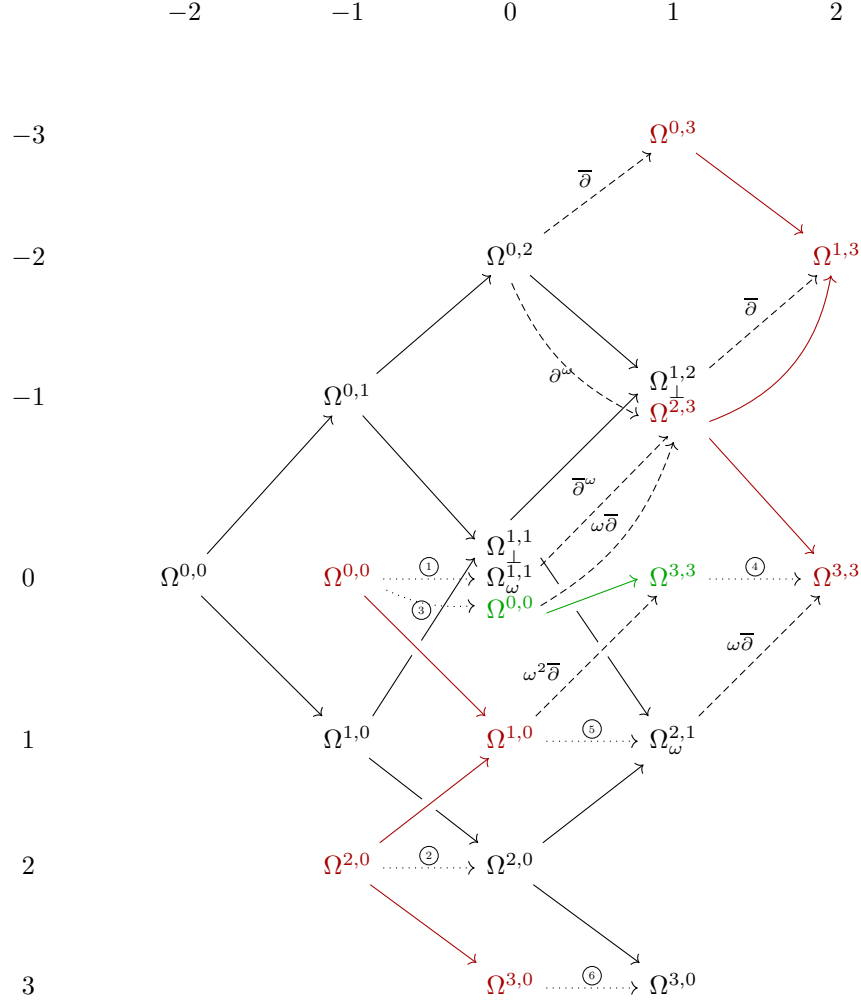
\begin{figure}
  \[
    \begin{tikzcd}[row sep = 3 em, column sep = 3.5 em] 
      & -2 & -1 & 0 & 1 & 2 \\
     -3   & & & & \fermi{\Omega^{0,3}} \arrow[red!65!black]{rd} & \\ 
     -2   & & & \Omega^{0,2} \ar[dr] \ar[dr, shift right = 1 ex, bend right = 20, dashed, "\partial^\omega" below, end anchor = {[yshift = - 2.0 ex, xshift=+0.8ex]}
     ] \ar[dashed, ur, "\dbar"] & & \fermi{\Omega^{1,3}}  \\
     -1   & & \Omega^{0,1} \arrow{ru} \arrow{rd} & &  \begin{matrix} \Omega^{1,2}_\perp \\  \fermi{\Omega^{2,3}} \arrow[red!65!black, ur, bend right = 30, start anchor = {[yshift=-2ex]east}, end anchor = {[xshift=-0.5ex]south}] \ar[dashed, ur, "\dbar"] \arrow[red!65!black]{dr} \end{matrix} 
     & \\
     0   & \Omega^{0,0} \arrow{ru} \arrow{rd}  & 
            \fermi{\Omega^{0,0}} \ar["{\circleds{1}}", dotted, r, end anchor={[xshift=0ex]} ] 
            &\begin{matrix} \Omega^{1,1}_\perp \\ \Omega^{1,1}_\omega \\ \scalar{\Omega^{0,0}} \end{matrix} 
            \arrow[ur, start anchor = {north}, end anchor = {[yshift = +1.2em]}] 
            \ar[ur, dashed, "\dbar^\omega", start anchor = {[xshift = -0.5ex, yshift=-1.5ex]}, end anchor = {[xshift = -0.5ex,yshift=+0.5ex]south}] 
            \arrow[rd, start anchor = {[yshift=+5ex,xshift=-0.5ex]}]
   \arrow[green!65!black, r, crossing over, start anchor = {[yshift=-3ex]} ]
   \arrow[dotted, "\circleds{3}", <-, crossing over, start anchor = {[xshift=0ex,yshift=-1.5ex]}, bend left = 15, l]
   \arrow[dashed,  "\omega \dbar", bend right = 20, ru, crossing over, end anchor = {south}, start anchor = {[xshift = -0.5ex, yshift = -3.5ex]}]
                           & \scalar{\Omega^{3,3}} \ar[dotted, r, "\circleds{4}"] & \fermi{\Omega^{3,3}} \\
                           1   & & \Omega^{1,0} \arrow[ru, end anchor = {[yshift=+5ex,xshift=+0ex]}] \arrow{rd} & 
                           \fermi{\Omega^{1,0}} \ar[ur, crossing over, dashed, "\omega^2 \dbar" near start] 
                                \arrow[red!65!black, crossing over, <-]{ul}
                                \ar[dotted, r, "\circleds{5}"] & 
                           \Omega^{2,1}_\omega \ar[ur, "\omega \dbar", dashed] & \\
                           2   & & \fermi{\Omega^{2,0}} \arrow[red!65!black, crossing over]{ru} \arrow[red!65!black]{rd} 
                           \ar[dotted, r, "\circleds{2}" ] & \Omega^{2,0} \arrow{ru} \arrow{rd} & & \\
      3   & & & \fermi{\Omega^{3,0}} \arrow[dotted, r, "\circleds{6}" ] & \Omega^{3,0} & \\
    \end{tikzcd}
  \]
  \caption{The holomorphic twist of the $\N=(1,0)$ tensor multiplet. 
  The horizontal grading is the cohomological grading. 
  The vertical grading is the weight with respect to $\U(1) \subset \U(3)$.} \label{fig:(1,0)twist}
\end{figure}

We have labeled the differential generated by the supercharge $Q$ by the dotted and dashed arrows, which we now proceed to justify. 
The dotted arrows $\begin{tikzcd} \; \ar[r, dotted] & \; \end{tikzcd}$ denote order zero differential operators, and the dashed arrows $\begin{tikzcd} \; \ar[r, dashed] & \; \end{tikzcd}$ are given by the labeled differential operator.
Throughout, we extensively refer to the notation in \S \ref{ssec:module} where we constructed the action of supersymmetry on the tensor multiplet. 

We begin with the component of the supersymmetry action which transforms a fermion into $\chi_+(2)$.
In the notation of \S \ref{sec:physical} this is the linear map $\rho_{\Psi,2}$. 
In the holomorphic decomposition of the fields this map is the following projection
\[
\rho_{\Psi,2} (Q \otimes -) : \Omega^{0} \oplus \Omega^{2,0} \oplus \Omega^{0,3} \oplus \Omega^{2,3} \twoheadrightarrow \omega \Omega^{0} \oplus \Omega^{2,0} \subset \Omega^{1,1} \oplus \Omega^{2,0} \subset \chi_+ (2)
\]
which reads $\rho_{\Psi, 2} (Q \otimes \psi_-) = \omega \psi_-^{0,0} + \psi_-^{2,0}$.
This term accounts for the dotted arrows in Figure \ref{fig:(1,0)twist} labeled $\circled{1}$ and $\circled{2}$. 
On the anti-fields, the map $\rho_{\Psi,2}(Q \otimes -)$ is given by the composition $\pi \circ \d$ where $\pi^+$ is the projection
\[
\pi^+ : (S_+ \otimes R_1) \otimes \Omega^4 \twoheadrightarrow S_+ \otimes R_1 .
\]
When restricted to the holomorphic supercharge $Q \in S_+ \otimes R_1$ this projection defines a linear map $\pi' (Q \otimes -)$ which reads, in holomorphic coordinates:
\[
\pi (Q \otimes -) : \Omega^{1,3} \oplus \Omega^{2,2} \oplus \Omega^{3,1} \twoheadrightarrow \Omega^{1,3} \oplus \Omega^{2,2} \to \Omega^{1,3} \oplus \Omega^{3,3} 
\]
where the last map uses the K\"{a}hler form as a linear map $\omega : \Omega^{2,2} \to \Omega^{3,3}$. 
Thus, acting on the anti-fields, $\rho_{\Psi,2}(Q \otimes -)$ reads
\[
\begin{tikzcd}
& & & \Omega^{1,2} / \omega \Omega^{0,1} \ar[r, "\dbar"] \ar[dr, "\partial"] & \Omega^{1,3} \ar[r, "="] & \Omega^{1,3} & & \\ 
\rho_{\Psi,2}(Q \otimes -) : & \Omega^{3}_+&  = & \omega \Omega^{1,0} \ar[r, "\dbar"] \ar[dr, "\partial"] & \omega \Omega^{2,2} \ar[r, "\omega"] & \Omega^{3,3} & \subset & \Psi^\alpha_-(R_1) [1]  \\ 
& & & \Omega^{3,0} \ar[r, "\dbar"] & \Omega^{3,1} \ar[r] & 0 & & 
\end{tikzcd}
\]
This accounts for the dashed arrows $\begin{tikzcd}  \Omega^{1,2} / \omega \Omega^{0,1}\ar[r, dashed, "\dbar"] & \Omega^{1,3} \end{tikzcd}$, $\begin{tikzcd}  \omega \Omega^{1,0} \ar[r, dashed, "\omega \dbar"] & \Omega^{3,3} \end{tikzcd}$, and $\begin{tikzcd}  \Omega^{1,2} / \omega \Omega^{1,3} \ar[r, dashed, "\omega \partial"] & \Omega^{3,3} \end{tikzcd}$.

We turn to the part of the supersymmetry which transforms fermions into the scalar $\Phi(0,\CC)$. 
In the notation of \S \ref{sec:physical} this is the linear map $\rho_{\Psi,0}$ which is defined using the projection of a tensor product of spin representations onto a trivial summand.
When applied to the holomorphic supercharge $Q \in S_+ \otimes R_1$, the resulting linear map is given by the projection
\[
\rho_{\Psi, 0} (Q \otimes -) : \Omega^0 \oplus \Omega^{2,0} \oplus \Omega^{0,3} \oplus \Omega^{2,3} \twoheadrightarrow \Omega^0 .
\]
In the decomposition of the fields above, this reads $\rho_{\Psi, 0} (Q \otimes \psi_-) = \psi_-^{0,0} \in \Omega^{0,0} \subset \Phi(0, \CC)$. \footnote{In standard physics notation, one would write this as $\delta_Q \phi = \psi_-^{0,0}$.}
This term accounts for the dotted arrow $\begin{tikzcd} \; \ar[r, dotted] & \; \end{tikzcd}$ in Figure \ref{fig:(1,0)twist} labeled $\circled{3}$.
Similarly, on the anti-fields we have $\rho_{\Psi, 0} (Q \otimes \phi^+) = \phi^+ \in \Omega^{3,3} \subset \Psi^\alpha_-(R_1)  [2]$. 
which one could write as $\delta_Q \psi_-^+ = \phi^+$. 
This term accounts for the dotted arrow labeled $\circled{4}$. 

Next, consider the supersymmetry which transforms $\Phi(0,\CC)$ into the fermion. 
In the notation of \S \ref{sec:physical} this is the linear map $\rho_{\Phi}$. 
Applied to the supercharge $Q$, this is the composition 
\[
\rho_{\Phi} (Q \otimes -) : \Omega^0 \xto{\d} \Omega^{1,0} \oplus \Omega^{0,1} \xto{Q} \omega \Omega^{0,1} \subset \Psi^\alpha_-(R_1) [1] 
\]
which reads $\rho_{\Phi} (Q \otimes \phi) = \omega \dbar \phi$ and accounts for the dashed arrow $\begin{tikzcd}  \Omega^0 \ar[r, dashed, "\omega \dbar"] & \omega \Omega^{0,1} \end{tikzcd}$.
On antifields this is the dashed arrow $\begin{tikzcd}  \Omega^{1,0} \ar[r, dashed, "\omega^2 \dbar"] & \Omega^{3,3} \end{tikzcd}$.

Next, consider the supersymmetry which transforms $\chi_+(2)$ into $\Psi^\alpha_-(R_1)$.
In the notation of \S \ref{sec:physical} this is the linear map $\rho_{\chi}$, which when acting on the physical fields is the composition $\pi \circ \d_-$ where $\pi$ is the projection
\[
\pi : (S_+ \otimes R_1) \otimes \Omega^3_- \to S_- \otimes R_1 .
\]
When restricted to the holomorphic supercharge $Q \in S_+ \otimes R_1$ this projection defines a linear map $\pi (Q \otimes -) : \Omega^3_- \to S_- \otimes R_1$ which reads, in holomorphic coordinates:
\[
\pi (Q \otimes -) : \Omega^{0,3} \oplus \omega\Omega^{0,1} \oplus \Omega^{2,1} / \omega \Omega^{0,1} \twoheadrightarrow \Omega^{0,3} \oplus \omega \Omega^{0,1} 
\]
and is given by the obvious projection. 
Thus, acting on the physical fields, the map $\rho_\chi(Q \otimes -)$ is the composition
\[
\begin{tikzcd}
& & & \Omega^{0,2} \ar[r, "\dbar"] \ar[dr, "\partial^\omega"] & \Omega^{0,3} \ar[r, "="] & \Omega^{0,3} & & \\ 
\rho_{\chi}(Q \otimes -) : & \Omega^{2}&  = & \Omega^{1,1} \ar[r, "\dbar^\omega"] \ar[dr, "\partial_\omega"] & \omega \Omega^{0,1} \ar[r, "="] & \Omega^{0,1} & \subset & \Psi^\alpha_-(R_1) [1]  \\ 
& & & \Omega^{2,0} \ar[r, "\dbar_\omega"] & \Omega^{2,1} / \omega \Omega^{1,0} \ar[r] & 0 & & 
\end{tikzcd}
\]
This accounts for the dashed arrows $\begin{tikzcd}  \Omega^{0,2} \ar[r, dashed, "\dbar"] & \Omega^{0,3} \end{tikzcd}$, $\begin{tikzcd}  \Omega^{0,2} \ar[r, dashed, "\partial^\omega"] & \omega \Omega^{0,1} \end{tikzcd}$, and $\begin{tikzcd}  \Omega^{1,1} \ar[r, dashed, "\dbar^\omega"] & \omega \Omega^{0,1} \end{tikzcd}$.
The anti map for this component of supersymmetry acts on $\Omega^{1,0} \oplus \Omega^{3,0} \subset \Psi^\alpha_-(R_1)$ and is defined by the projection:
\[
\rho_\chi (Q \otimes -) : \Omega^{1,0} \oplus \Omega^{3,0} \oplus \Omega^{1,3} \oplus \Omega^{3,3} \twoheadrightarrow \omega \Omega^{1,0} \oplus \Omega^{3,0} \subset \chi_+(2) [1]  .
\]
This accounts for the dotted arrows labeled $\circled{5}$ and $\circled{6}$. 
All arrows have been accounted for, and we have thus verified that the twisted theory $\cT^Q_{(1,0)}$ is described by Figure~\ref{fig:(1,0)twist}.

\subsubsection{Verification of the equivalence between $\cT^Q_{(1,0)}$ and~$\thy(2)$}

We now move to the second step of the proof. To begin, we note that there is a projection $\Phi$ from the total complex $\cT^Q_{(1,0)}$ in Figure \ref{fig:(1,0)twist} to the cochain complex $\chi(2) = \Omega^{\leq 1, \bu}[2]$:
\beqn\label{(1,0)phi}
\Phi : \cT^Q_{(1,0)} \to \chi(2) .
\eeqn
\begin{table}
  \[
    \begin{array}{ccclllll}
    \Phi(\beta^{0, j}) & =  & \beta^{0,j} & \in & \Omega^{0,j}, & \quad  j = 0,1,2;  \\
\Phi(\psi_-^{0,3}) & = &  \psi_-^{0,3} & \in & \Omega^{0,3} ; \\
\Phi(\beta^{1,0}) & = & \beta^{1,0} & \in & \Omega^{1,0} ; \\
\Phi(\beta^{1,1} + \phi^0) & = & \beta_\perp^{1,1} + (\beta^{1,1}_{\omega} - \omega \phi^0) & \in & \Omega^{1,1} ;  \\ 
\Phi([\beta^{1,2}]_\omega + \omega \psi_-^{0,1}) & = &  \beta^{1,2} + \omega \psi_-^{0,1} & \in & \Omega^{1,2} ;  \\
\Phi(\psi_-^{+1,3}) & = & \psi^{+1,3}_- & \in & \Omega^{1,3}   .
\end{array}
\]
\caption{A component description of the projection map $\Phi$}
\label{tab:Phi}
\end{table}
On components, $\Phi$ is defined by the formulas in Table~\ref{tab:Phi}; it sends all other fields to zero.
Here, $\beta_\perp^{1,1}$ and $\beta_\omega^{1,1}$ denote the components of the $(1,1)$-form under the decomposition $\Omega^{1,1} = \Omega^{1,1}_\perp \oplus \omega \Omega^0$. 
Notice that
\[
\Phi(\Tilde{Q}_{\rm BV} \psi_-^0) = \Phi (\partial \psi_-^0 + \omega \psi^{0}_- + \psi_-^0) = 0 + (\omega \psi^0_- - \omega \psi^0_-) = 0
\]
which is the only nontrivial check that $\Phi$ is a cochain map.
Notice that $\Phi$ is a map of underlying graded vector bundles, so its kernel is well-defined. 
Since all the dotted arrows are isomorphisms, the kernel of this map is acyclic, and so $\Phi$ defines a quasi-isomorphism of sheaves of cochain complexes. 

Since the supercharge $Q$ preserves the presymplectic structure $\omega_\cT$ on $\cT_{(1,0)}$,  we know that $\cT^Q_{(1,0)}$ has the induced structure of a presymplectic BV theory.
As discussed above, there is also a natural shifted presymplectic structure on $\chi(2)$, defined by the formula $\omega_\chi = \int_{\CC^3} \alpha \partial \alpha '  .$
To check that $\Phi$ defines an equivalence of \psBV{} theories, we will need to check its compatibility with these pairings. 

We note that the quasi-isomorphism $\Phi$ does \emph{not} preserve the shifted presymplectic structures in any strict sense. 
However, there does exist a two-form on the space of fields $h : \cT^Q_{(1,0), c} \times \cT^Q_{(1,0), c} \to \CC$ of degree $-2$ such that 
\beqn\label{eqn:htpy}
\omega_\cT - \Phi^* \omega_\chi = (Q_{\rm BV} +Q) h,
\eeqn
where $Q_{\rm BV}+Q$ denotes the internal differential on the cochain complex of two-forms in field space, with respect to the \emph{total} differential on $\cT^Q_{(1,0)}$. 
In writing elements of the space of two-forms, we will always suppress the integration symbol over~$\C^3$, which should be understood implicitly. We also suppress the subscripts indicating the chirality of the (untwisted) fermions.

\begin{prop}
  \label{prop:pairings}
  Consider the two-form on the space of fields
  \deq[eq:homotopy]{
    h = \beta^{3,0} \psi^{0,3} + \beta^{2,0} (\psi^+)^{1,3} + \beta^{2,1}_\omega \cdot \omega^{-1} \psi^{2,3} + \phi^{0,0} (\psi^+)^{3,3}.
}
Then $h$ defines a homotopy between $\omega_\cT$ and the pullback $\Phi^* \omega_\chi$. That is, \eqref{eqn:htpy} is satisfied.
\end{prop}
\begin{proof}
The proof is a straightforward computation.
The pairing on $\thy(2)$ is given by
\deq{
  \omega_\thy = \chi \del \chi = \chi^{1,3} \del \chi^{1,0} + (\chi^{1,2}_\perp + \chi^{1,2}_\omega ) \del ( \chi^{1,1}_\perp + \chi^{1,1}_\omega),
}
containing a total of five terms. Applying the pullback, we obtain
\deq{
\Phi^* \omega_\thy = (\psi^+)^{1,3} \del \beta^{1,0} + \left(\beta^{1,2}_\perp + \omega^{-1} \psi^{2,3}\right) \del \left(\beta^{1,1}_\perp + \beta^{1,1}_\perp + \omega \phi^{0,0} \right).
}
The pairing on $\cT^Q_{(1,0)}$ is given by
\deq{
\omega_\cT = \phi^{0,0} \phi^{3,3} + \psi^{i,0} \psi^{3-i,3} + \beta^{3,0} \dbar \beta^{0,2} + \beta^{2,1}_\omega \left( \del \beta^{0,2} + \dbar \beta^{1,1}_\perp + \dbar \beta^{1,1}_\omega \right)
  + \beta^{1,2}_\perp \left( \del \beta^{1,1}_\perp + \del \beta^{1,1}_\omega + \dbar \beta^{2,0} \right) .
}
In writing the term $\psi^{i,0}\psi^{3-i,3}$, we have suppressed the antifield symbols; of course, this means the nondegenerate pairing on the fermi fields, and would more properly be written $\psi^{\text{ev},0} (\psi^+)^{\text{odd},3} + (\psi^+)^{\text{odd},0} \psi^{\text{ev},3}$. Note also that we make no claim that all of the terms we write are nonvanishing (for example, many will identically vanish on a compact K\"ahler manifold); the point is that our claim holds formally even without using these facts.

When taking the difference of the pairings, the sixth and seventh terms of $\omega_\cT$ cancel with corresponding terms, and the result is
\begin{multline}
  \omega_\cT  - \pi^* \omega_\thy = \phi^{0,0} \phi^{3,3} + \psi^{i,0} \psi^{3-i,3} + \beta^{3,0} \dbar \beta^{0,2} + \beta^{2,1}_\omega \left( \del \beta^{0,2} + \dbar \beta^{1,1}_\perp + \dbar \beta^{1,1}_\omega \right)
  + \beta^{1,2}_\perp \dbar \beta^{2,0}
\\ - (\psi^+)^{1,3} \del \beta^{1,0} - \beta^{1,2}_\perp \omega \del \phi^{0,0} - \omega^{-1} \psi^{2,3} \del \left(\beta^{1,1}_\perp + \beta^{1,1}_\perp + \omega \phi^{0,0} \right).
 \label{eq:pairing-diff}
\end{multline}
This is obviously nonzero as a two-form on $\cT^Q_{(1,0)}$, but  we will show that it is the BV variation of the homotopy~\eqref{eq:homotopy}.
Note that the last term of the homotopy crucially contains only the scalar field, and \emph{not} $\beta^{1,1}_\omega$.

To compute the BV variation $Q_{\rm BV} h$, we will need to consider all differentials in $\cT^Q_{(1,0)}$ \emph{entering} terms that appear in the homotopy. As usual, it is helpful to refer to Figure~\ref{fig:(1,0)twist}.

As a first step, note that the homotopy $h$ pairs the fermions at the upper right of Figure~\eqref{fig:(1,0)twist} with an isomorphic subcomplex of~$\cT^Q_{(1,0)}$. All of the ``internal'' arrows in each of these Z-shaped subdiagrams can thus be ignored; the terms they generate in the variation will occur twice, once from each  side of the pairing, and will cancel after an integration by parts.

It is also clear that the arrows that do \emph{not} contain differential operators---the dotted arrows 2 through 6 in the figure---generate precisely the terms of the pairing which do not contain differential operators, on the scalar and between $\psi_=$ and~$\psi_+$. This accounts for the first two terms in~\eqref{eq:pairing-diff}.

It thus remains to consider only terms involving differential operators in both $\omega_\cT - \pi^* \omega_\thy$ and $Qh$, where we may ignore the ``internal'' differentials in computing the latter. We proceed term by term in the homotopy. The first term is 
\deq{
  Q(\beta^{3,0} \psi^{0,3}) = \beta^{3,0} \dbar \beta^{0,2},
}
which cancels with the third term of~\eqref{eq:pairing-diff}. The second term is 
\deq{
Q(\beta^{2,0} (\psi^+)^{1,3}) = \partial \beta^{1,0} (\psi^+)^{1,3} + \beta^{2,0} \dbar \beta^{1,2}_\perp.
}
These two terms cancel with the seventh and eighth terms of~\eqref{eq:pairing-diff} after an integration by parts. 

The third term in the homotopy generates the largest number of terms: we have 
\deq{
  Q( \beta^{2,1}_\omega \cdot \omega^{-1} \psi^{2,3}) = \beta^{2,1}_\omega \left( \del \beta^{0,2} + \dbar \beta^{1,1}_\perp + \dbar \beta^{1,1}_\omega + \omega \dbar \phi^{0,0} \right) + 
  \left( \del \beta^{1,1}_\perp + \del \beta^{1,1}_\omega + \omega \del \phi^{0,0} \right) \omega^{-1} \psi^{2,3}.
}
The last three terms in this variation cancel with the last three terms in~\eqref{eq:pairing-diff}, and the first three terms cancel with the fourth, fifth, and sixth terms of~\eqref{eq:pairing-diff}. The fourth term in the variation is left over.

It remains to calculate the variation of the fourth and last term of the homotopy, which is 
\deq{
Q( \phi^{0,0} (\psi^+)^{3,3} ) = \phi^{0,0} \left( \del \beta^{1,2}_\perp + \dbar \beta^{2,1}_\omega \right).
}
The first of these terms cancels the ninth and final term of~\eqref{eq:pairing-diff}, and the last term cancels the leftover piece from the variation of the third term of the homotopy (after another integration by parts). The proposition, and thus this portion of the main theorem, is proved.
\end{proof}

\subsection{Holomorphic decomposition for the (2,0) theory}

In this section we finish the second part of Theorem \ref{thm:twist} concerning the holomorphic twist of the $(2,0)$ tensor multiplet. 
Again, we fix the data of a rank one supercharge $Q$, this time viewed as an odd element of the super Lie algebra $\sp{2}$. 

Recall that the $R$-symmetry group of $(2,0)$ supersymmetry is $G_R = \Sp(2)$. 
As in the $(1,0)$ case, the supercharge $Q$ defines a complex structure $L = \CC^3 \subset V = \CC^6$ equipped with the choice of a holomorphic half-density on $L$.
The twist carries a symmetry by the subgroup group $\MU(3) \subset \Spin(6)$ whose action is defined by the twisting homomorphism
\[
\phi : \MU(3) \xto{{\rm det}^{\frac12}} \U(1) \xto{i \times 1} \Sp(1) \times \Sp(1)' \subset \Sp(2) = G_R .
\]
Here, $i : \U(1) \hookrightarrow \Sp(1)$ denotes the embedding for which $Q$ has weight $+1$. 
Also we use primes as in $\Sp(1) \times \Sp(1)' \subset \Sp(2)$ to differentiate between the two abstractly isomorphic groups. 

Under the twisting homomorphism $\phi$ the defining representation $R_2$ of $\Sp(2)$ decomposes as
\deq{
\label{eqn:decompose20}
R_2 = \det(L)^{-\frac12} \oplus \det(L)^{\frac12} \oplus R_1' 
}
where $\MU(3)$ acts trivially on $R_1'$. 
The vector representation $W$ of $\Sp(2) = \Spin(5)$ decomposes under $\phi$ as
\[
W = \CC \oplus \left(\det(L)^{-\frac12} \oplus \det(L)^{\frac12} \right) \otimes R_1' .
\]

The regrading datum is specified by the homomorphism 
\[
\alpha : \U(1) \hookrightarrow \Sp(1) \xto{i \times 1} \Sp(1) \times \Sp(1)' \subset \Sp(2) = G_R .
\]
Note that this factors through the twisting homomorphism we used in the $(1,0)$ case along the embedding $\Sp(1) \hookrightarrow \Sp(2)$.





In addition to $\MU(3)$, the twist enjoys a global symmetry by the group $\Sp(1)'$.
Moreover, these actions commute for the trivial reason that $\MU(3)$ acts trivially on $\Sp(1)'$.
Using Equation (\ref{eqn:decompose1}), we observe that, after applying the twisting homomorphism $\phi$, the odd part $\Sigma_2$ of the super Lie algebra $\sp{2}$ transforms under $\MU(3) \times \Sp(1)' \subset \Spin(6) \times \Sp(2)$ as:

\begin{equation}
  \begin{tikzcd}[row sep = -5 pt]
&     -1 & 0 & 1 \\[10pt] \hline \\[10pt]
3 &  \det(L) \phantom{L^{\frac12}_{\frac12}}\\
5/2 & \phantom{ L\otimes\det(L)^{-1} } \phantom{L^{\frac12}_{\frac12}}\\
2 & \phantom{L^{\frac12}_{\frac12}}\\
3/2 &   & \det(L)^{\frac12} \otimes \Pi R_1' \phantom{L^{\frac12}_{\frac12}}\\
1 & \color{green!45!black} L \phantom{L^{\frac12}_{\frac12}}\\
1/2 &   \phantom{L^{\frac12}_{\frac12}}\\
0 & & &    \color{red!65!black}\CC \cdot Q &  \phantom{L^{\frac12}_{\frac12}}\\
-1/2 &      & \color{blue!65!black}{L\otimes\det(L)^{-\frac12} \otimes \Pi R_1' } \phantom{L^{\frac12}_{\frac12}}\\
 -1 &     \phantom{L^{\frac12}_{\frac12}}\\
  -3/2 &    \phantom{L^{\frac12}_{\frac12}}\\
-2   & & &  L\otimes\det(L)^{-1} 
  \end{tikzcd}
  \label{eq:tgdiagram}
\end{equation}

In this table, the vertical grading organizes spin number, and the horizontal grading is by ghost $\ZZ$-degree.  
The terms involving $R_1'$ are all odd with respect to the new $\ZZ/2$-grading.

The holomorphic supercharge $Q$ lies in the red summand. 
Its only nonzero bracket occurs with the supercharges in
$L$ represented in green above, using the degree-zero pairing on the $R$-symmetry space. As remarked above, this bracket witnesses a nullhomotopy of the translations in~$L$ with respect to the holomorphic supercharge.

In Proposition~\ref{prop:decompose(2,0)}, we described the $\Sp(1) \times \Sp(1)'$ decomposition of the $(2,0)$ tensor multiplet as a sum of the $(1,0)$ tensor multiplet plus the $(1,0)$ hypermultiplet valued in the symplectic representation $R_1'$:
\[
\mplet{2} = \mplet{1} \oplus \hyper{R_1'}
\]
Analogously, accounting for the twisting data $\phi, \alpha$ just introduced we have the following description of the regraded $(2,0)$ tensor multiplet.

\begin{prop}
\label{prop:(2,0)regraded}
The $\MU(3)$-regraded $(2,0)$ tensor multiplet $\phi^* \cT_{(2,0)}^\alpha$ decomposes as 
\[
\phi^* \cT_{(2,0)}^\alpha = \phi^* \cT_{(1,0)}^\alpha \oplus \Pi \phi^*\cT^\alpha_{\rm hyp} (R_1')
\]
where $\phi^* \cT_{(1,0)}^\alpha$ is the regraded $(1,0)$ tensor multiplet as in Proposition \ref{prop:regraded} and  $\phi^*\cT^\alpha_{\rm hyp} (R_1')$ is the free BV theory of the regraded hypermultiplet whose complex of fields is displayed in Figure \ref{fig:(2,0)regraded}.
\end{prop}

In Figure \ref{fig:(2,0)regraded}, the operator $\dbar^*$ denotes the adjoint of $\dbar$ corresponding to the standard K\"{a}hler form on $\CC^3$.  
Under the regrading $\cT^{\rm hyp}(R_1') = \Phi(0,R_1') \oplus \Psi_-(R_1') \rightsquigarrow  \Pi \phi^*\cT^\alpha_{\rm hyp} (R_1')$, we will denote the decomposition of the fields as:
\beqn\label{eqn:(2,0)scalar}
\Phi(0, R_1') \ni \nu = \nu^{\frac32,0} + \nu^{\frac32,3} \in \Omega^{0} (K^{\frac12} \otimes R_1') [1] \oplus \Omega^{0,3} (K^{\frac12} \otimes R_1') [-1]
\eeqn
for the scalars and
\beqn\label{eqn:(2,0)fermion}
\Psi_-(R_1') \ni \lambda = \lambda^{\frac32,3} + \lambda^{\frac32,1} \in \Omega^{0,3}(K^{\frac32} \otimes R_1') \oplus \Omega^{0,1}(K^{\frac12} \otimes R_1') 
\eeqn
for the fermions.
A similar decomposition holds for the anti-fields which will be denoted $\nu^{+\frac32,0}$, etc..

\begin{figure}
  \begin{center}
    \[
      \begin{tikzcd}[row sep = 1 em, column sep = 1 em]
      \phi^*\cT^\alpha_{\rm hyp} (R_1') & -1 & 0 & 1 & 2 \\ \hline
	& & \Omega^{0,3}(K^{\frac12} \otimes R_1') \ar[r, "\dbar^*"] & \Omega^{0,2} (K^{\frac12} \otimes R_1') & \\
     & & \Omega^{0,1}(K^{\frac12} \otimes R_1') \ar[ur, "\dbar"] \ar[r, "\dbar^*"] & \Omega^{0}(K^{\frac12} \otimes R_1') & 
     \\ \hline
 & \Omega^{0}(K^{\frac12} \otimes R_1') \ar[r, "\triangle"] & \Omega^{0} (K^{\frac12} \otimes R_1') & \\
 & & & \Omega^{0,3}(K^{\frac12} \otimes R_1') \ar[r, "\triangle"] & \Omega^{0,3} (K^{\frac12} \otimes R_1') &
      \end{tikzcd}
    \]
  \end{center}
  \caption{The subcomplex $\phi^*\cT^\alpha_{\rm hyp} (R_1')$ of the $\MU(3)$-regraded $(2,0)$ tensor multiplet, see Proposition \ref{prop:(2,0)regraded}. 
  The top complex is the result of regrading the fermions in the $(1,0)$ hypermultiplet, and the bottom complex is the result of regrading the bosons in the $(1,0)$ hypermultiplet.}
  \label{fig:(2,0)regraded}
\end{figure}

\begin{proof}
The $\cN = (2,0)$ multiplet splits as a sum of three complexes
\[
\mplet{2} = \thy_+ (2) \oplus \Fermi_-(R_2) \oplus \Phi(0,W) .
\]
As in the case of the $\cN = (1,0)$ multiplet, the component $\chi_+(2)$ is not charged under the $R$-symmetry group $G_R = \Sp(2)$. 

The physical fields of $\Fermi_-(R_2)$ decompose under the twisting homomorphism $\phi$ as:
\beqn \label{fermi20}
\Pi \left(\Omega^0 \otimes S_- \otimes R_2\right) = \bigg( \Pi (\Omega^{0,0} \oplus \Omega^{2,0}) \oplus \Pi(\Omega^{0,3} \oplus \Omega^{2,3}) \bigg) \oplus \Pi \bigg(\Omega^0 \otimes S_- \otimes R_1' \bigg).
\eeqn
The first component in parentheses contributes to the regraded $\cN = (1,0)$ tensor as in Proposition \ref{prop:regraded}.
The second component 
\[
\Omega^0 \otimes S_- \otimes R_1' =  \Omega^0(K^{-\frac12} \otimes R_1') \oplus \Omega^{0,1}(K^{\frac12} \otimes R_1') 
\]
contributes to the regraded hypermultiplet $\Pi \Tilde{\cT}^{\rm hyp} (R_1')$. 
There is a similar decomposition for the anti-fields in $\Psi_-(R_2)$. 

Next, the physical fields of the scalar theory $\Phi(0, W)$ decompose as
\beqn\label{scalar20}
\Omega^0 \otimes W = \Omega^0 \oplus \Omega^0 \otimes \left(K^{-\frac12} \oplus K^{\frac12} \right) \otimes R_1' 
\eeqn
The first summand, the single copy of smooth functions $\Omega^0$, contributes to the regraded $(1,0)$ tensor multiplet. 
The second summand contributes to $\Pi \Tilde{\cT}^{\rm hyp} (R_1')$. 
There is a similar decomposition for the anti-fields in $\Phi(0, W)$. 

By Proposition \ref{prop:regraded}, upon regrading, we see that the components $\chi_+(2)$, the first summand of \ref{fermi20}, and the first summand of (\ref{scalar20}), combine to give the regraded $(1,0)$ tensor multiplet. 

Of the remaining terms, the only component which is acted upon nontrivially by $\Sp(2)$ is the second summand in (\ref{scalar20}) (and the corresponding antifields). 
Under $\alpha$, we see that the factor proportional to $K^{\frac12}$ has weight $-1$ and the factor $K^{\frac12}$ has weight $+1$. 
It remains to check that the BV differential decomposes as stated, but this is nearly identical to the proof of Proposition \ref{prop:regraded}. 
\end{proof}

\subsection{Proof of (2,0) part of Theorem \ref{thm:twist}}\label{sec:(2,0)twist}

\begin{figure}
  \begin{equation*}
    \begin{tikzcd}[row sep = 1 em, column sep = 2 em]
      & -2 & -1 &  0 & 1 & 2 \\ \hline
      3 & & &  & \Omega^{0,3} \ar[rdd] &  \\
      5/2 & & & & &  \\
      2 & & & \Omega^{0,2} \ar[ruu] \ar[rdd]& & \Omega^{1,3} \ar[from=ddl] \\
      3/2 & & & \fermi{\Omega^{\frac32,3}(\Pi R_1')} 
        \ar[r, dotted, crossing over, "\circleds{1}"]  & \scalar{\Omega^{\frac32, 3}(\Pi R_1')} \arrow[green!65!black, r, crossing over] & \scalar{\Omega^{\frac32,3} (\Pi R_1')} \\
      1 & & \Omega^{0,1} \ar[rdd] \ar[ruu] & & \Omega^{1,2} \ar[from=ldd] &  \\
      \frac12 & & & & \fermi{\Omega^{\frac32,2}(\Pi R_1')} \ar[from=luu, red!65!black, crossing over] \ar[ruu,dashed, "\dbar"] &  \\
      0 & \Omega^{0} \ar[ruu] \ar[rdd] & & \Omega^{1,1} & & \\
      -\frac12 & & & \fermi{\Omega^{\frac32,1}(\Pi R_1')} \ar[red!65!black, ruu] \ar[red!65!black, ddr] & &  \\
      -1 & & \Omega^{1,0} \ar[ruu] &  &  & \\
      -3/2 & & \scalar{\Omega^{\frac{3}{2},0}(\Pi R_1')} \ar[ruu, dashed, "\dbar"] \ar[r,green!65!black] & \scalar{\Omega^{\frac32,0} (\Pi R_1')} \ar[r, dotted, "\circleds{2}"] & \fermi{\Omega^{\frac32,0} (\Pi R_1')} & \\
    \end{tikzcd}
  \end{equation*}
  \caption{The holomorphically twisted $\N=(2,0)$ theory $\cT_{(2,0)}^Q$. 
The horizontal grading is the cohomological $\ZZ$-grading.
Note that the green and red text sits in {\em odd} $\ZZ/2$-degree. 
The vertical grading is the weight with respect to $\U(1) \subset \MU(3)$.
}
  \label{fig:(2,0)holtwist}
\end{figure}
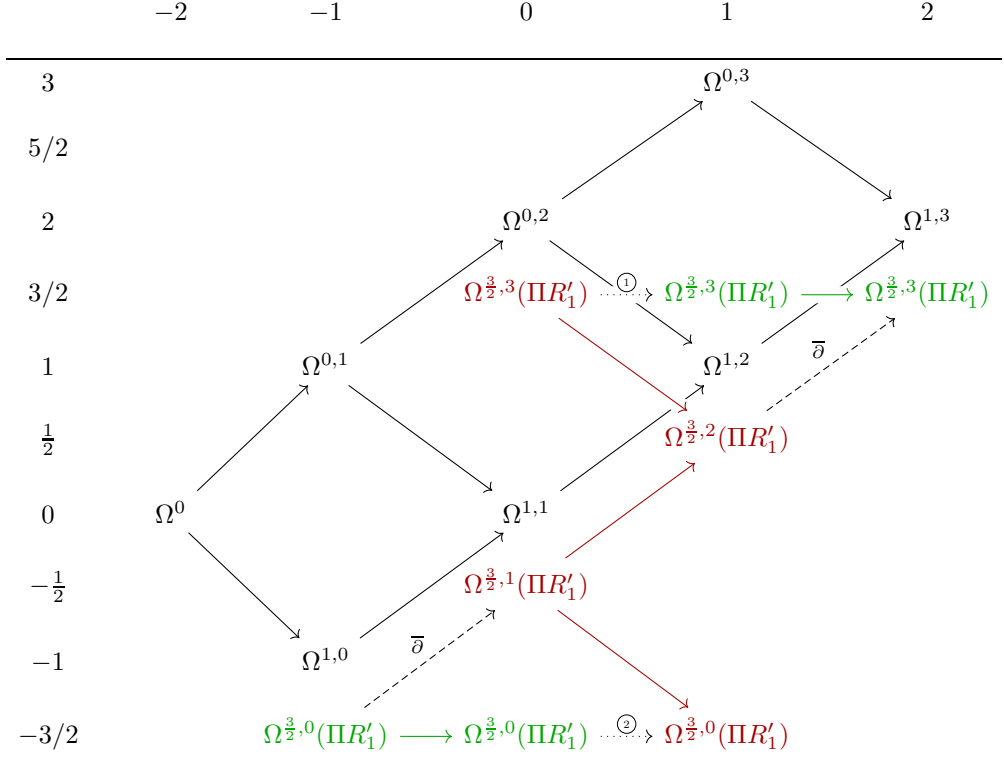

We now complete the proof of Theorem \ref{thm:twist}, which involves deforming the regraded theory described in Proposition \ref{prop:(2,0)regraded} by the holomorphic supercharge $Q$.
Throughout this section we refer to the description of the twisted theory in Figure \ref{fig:(2,0)holtwist}. 

According to Proposition \ref{prop:(2,0)regraded}, the $Q$-twisted theory splits as a sum of two complexes
\beqn\label{eqn:(2,0)holdecomp}
\mplet{1}^Q \oplus \cT_{\rm hyp} (R_1')^Q
\eeqn
where $\mplet{1}^Q$ is the $Q$-twist of the $(1,0)$ tensor multiplet and $\cT_{\rm hyp} (R_1')^Q$ is the theory obtained by deforming the $\MU(3)$-regraded hypermultiplet $\Pi \phi^*\cT^\alpha_{\rm hyp} (R_1')$ by $Q$.  

In Figure \ref{fig:(2,0)holtwist}, the black solid arrows represent the twist of the $\cN=(1,0)$ tensor multiplet, as we computed in \S \ref{sec:(1,0)twist} which corresponds to the first summand $\mplet{1}^Q$ in (\ref{eqn:(2,0)holdecomp}). 
The red text refers to the $\MU(3)$-regraded hypermultiplet $\Pi \phi^*\cT^\alpha_{\rm hyp} (R_1')$.
The red solid arrows represent the underlying classical BV differential of the regraded hypermultiplet.
Note that we use the shorthand notation $\Omega^{\pm \frac{3}{2}, \ell}(R_1')$ to mean the Dolbeault forms of type $(0,\ell)$ valued in the holomorphic vector bundle $K^{\pm \frac12} \otimes R_1'$. 
We have labeled the differentials generated by the holomorphic supercharge $Q$ acting on the hypermultiplet by the dotted and dashed arrows. 
As in the $\cN=(1,0)$ case, the dotted arrows $\begin{tikzcd} \; \ar[r, dotted] & \; \end{tikzcd}$ denote isomorphisms, and the dashed arrows $\begin{tikzcd} \; \ar[r, dashed] & \; \end{tikzcd}$ are given by the labeled differential operator, which we now proceed to characterize.
Again, we refer to the notation in \S \ref{ssec:module} where we constructed the action of supersymmetry on the tensor multiplet. 

We begin with the component of the supersymmetry action which transforms a fermion into a scalar. 
In the notation of \S \ref{sec:physical} this is the linear map $\rho_{\Psi,0}$. 
In the holomorphic decomposition, see Equation (\ref{eqn:(2,0)fermion}), of the fields we read off 
\[
\rho_{\Psi, 0} (Q \otimes \lambda) = \lambda^{\frac32,3} \in \Omega^{\frac32,3} \subset \Phi(0, R_1')[1]  ,
\]

This term accounts for the dotted arrow $\begin{tikzcd} \; \ar[r, dotted] & \; \end{tikzcd}$ in Figure \ref{fig:(2,0)holtwist} labeled $\circled{1}$.
Similarly, on the anti-fields we have
\[
\rho_{\Psi, 0} \left(Q \otimes \nu^+ \right) = \nu^{\frac32,0} \in \Omega^{\frac32,0} \subset \Psi_-(R_1'), 
\]
see the notation of Equation (\ref{eqn:(2,0)scalar}). 
This term accounts for the dotted arrow labeled $\circled{2}$. 

Next, we look at the component of supersymmetry which transforms a scalar into a fermion. 
In the notation of \S \ref{sec:physical} this is the linear map $\rho_{\Phi}$. 
In the holomorphic decomposition of fields we have
\[
\rho_{\Phi}(Q \otimes \nu) = \dbar \nu^{\frac32,0} \in \Omega^{\frac32, 1} \subset \Psi_-(R_1') .
\]
Similarly, on the anti-fields we have
\[
\rho_{\Phi}(Q \otimes \lambda^{+}) = \dbar \lambda^{\frac32,2} \in \Omega^{\frac32,3} \subset \Phi (0, R_1') .
\]
These maps account for each of the dashed arrows in Figure \ref{fig:(2,0)holtwist}.

Next, we will describe an equivalence of presymplectic BV theories
\[
\Phi : \cT_{(2,0)}^Q \to \chi(2) \oplus \hCS (\Pi R_1')
\]

On the $(1,0)$ tensor multiplet summand of the $(2,0)$ theory, the map $\Phi$ is defined to be the map (\ref{(1,0)phi}) that we used in the twist of the $(1,0)$ multiplet. 

On the $(1,0)$ hypermultiplet summand, the map is defined as follows. 
\beqn
\begin{array}{ccccccc}
\Phi(\nu^{\frac32, 0}) & = & \nu^{\frac32, 0} & \in & \Omega^{\frac32,0} \\
\Phi (\lambda^{\frac32 , 1}) & =&  \lambda^{\frac32, 1} & \in & \Omega^{\frac32, 1} \\ 
\Phi(\lambda^{\frac32, 2} + \nu^{\frac32, 3}) & = & \lambda^{\frac32, 2} - \dbar^* \nu^{+\frac32,3} & \in & \Omega^{\frac32, 2} \label{eqn:phihyper} \\
\Phi(\nu^{\frac32,3}) & = & \nu^{\frac32, 3} & \in & \Omega^{\frac32,3} .
\end{array}
\eeqn
The map $\Phi$ annihilates the remaining fields of the $(1,0)$ hypermultiplet. 

On the hypermultiplet, we note that this map is \emph{not} the obvious projection map of graded vector spaces, but is ``corrected" to account for the differentials mapping out of the acyclic subcomplex at the upper left of the hypermultiplet in Figure~\ref{fig:(2,0)holtwist}. 
The correction in this case is analogous to standard twist calculations, and follows the general rubric presented in Proposition~1.23 of~\cite{ESW}.
By this result, and the theorem for the $(1,0)$ tensor multiplet, it follows that $\Phi$ is a quasi-isomorphism. 

We have already shown how the map on the $(1,0)$ tensor multiplet is compatible with the degree $(-1)$ presymplectic structures. 
It is immediate to check that the map $\Phi$ restricted to the hypermultiplet strictly preserves the degree $(-1)$ presymplectic structures.


\subsubsection{An alternative description} \label{sec:alternate}


There is an alternative to the twisting data $(\phi, \alpha)$ in the case of the $(2,0)$ tensor multiplet. 
The key difference is that this variation admits a smaller global symmetry group. 
Note that the theory described in the previous section carries a global symmetry by the group $\MU(3) \times \Sp(1)'$, even {\em after} twisting.
This alternative twist breaks this global $\Sp(1)'$ symmetry completely, but further descends the $\MU(3)$-action to an action by $\U(3)$.

The reason this twist enjoys a smaller symmetry group is because it depends on the choice of a polarization of the $2$-dimensional symplectic vector space $R_1'$.
Such a polarization determines an embedding $i' : \U(1) \hookrightarrow \Sp(1)'$ which we now fix. 

Define the new twisting homomorphism by the composition
\[
\Tilde{\phi} : \MU(3) \xto{{\rm det}^{\frac12}} \U(1) \xto{\rm diag} \U(1) \times \U(1) \xto{i \times i'} \Sp(1) \times \Sp(1)' \subset \Sp(2).
\]

As in the previous section, $i : \U(1) \to \Sp(1)$ denotes the homomorphism for which $Q$ has weight $+1$.

Additionally, we have the regrading homomorphism 
\[
\Tilde{\alpha} : \U(1) \xto{\rm diag} \U(1) \times \U(1) \xto{i \times i'} \Sp(1) \times \Sp(1)' \subset \Sp(2) .
\]

To simplify the notation in the next section, we will denote by $\Tilde{\cT}_{(2,0)}$ the $\MU(3)$-regraded $(2,0)$ theory using this twisting data. 
The $Q$-twisted theory will be denoted by $\Tilde{\cT}^Q_{(2,0)}$. 

With this choice of a regrading homomorphism, the twisted theory $\Tilde{\cT}_{(2,0)}^Q$ descends to a $\U(3)$-equivariant theory and is concentrated in even $\ZZ/2$-degree and hence defines a $\ZZ$-graded theory. 
Aside from this, the only part of the calculation that changes is the subcomplex defined by the green and red text of Figure \ref{fig:(2,0)holtwist}, which we will henceforth denote by $\cA \subset \Tilde{\cT}_{(2,0)}^Q$.
 
For example, in the original description of the twist the scalar field lives in $\Pi \Omega^{\frac32, 0}(R_1')[1]$. 
According to this new twisting data this becomes
\[
\Pi \Omega^{\frac32, 0} (R_1') [1]  \oplus \Pi \Omega^{\frac32, 3} (R_1') [-1] \rightsquigarrow \left(\Omega^{0,0} \oplus \Omega^{3,0} [2]\right) \oplus  \left(\Omega^{3,0}[-2] \oplus \Omega^{3,3}[1] \right) .
\]
Similarly, using the original twisting data, the fermion field lives in $\Pi \Omega^{\frac32,1} \oplus \Pi \Omega^{\frac32 , 3}$. 
According to this new twisting data this becomes
\[
\Pi \Omega^{\frac32,1} \oplus \Pi \Omega^{\frac32 , 3}  \rightsquigarrow \left(\Omega^{0, 1} [-1] \oplus \Omega^{3,1} [1] \right) \oplus \left(\Omega^{0,3} [-1] \oplus \Omega^{3,3}[1] \right) .
\]

In total, using this alternative twisting data, the green and red subcomplex of the diagram in Figure \ref{fig:(2,0)holtwist}, which we denote $\cA$, is displayed in Figure \ref{fig:(2,0)holtwistalt}.
As before the solid arrows denote the differentials in the original untwisted theory. 
The dotted arrows $\begin{tikzcd} \; \ar[r, dotted] & \; \end{tikzcd}$ denote isomorphisms and the dashed arrows $\begin{tikzcd} \; \ar[r, dashed] & \; \end{tikzcd}$ arrows are given by the labeled differential operators induced by the action by $Q$. 
The green text labels the components arising from the scalar part of the untwisted theory, the red text labels the components arising from the fermion.  

\begin{figure} \label{fig:(2,0)holtwistalt}
  \begin{equation*}
    \begin{tikzcd}[row sep = 1 em, column sep = 2 em]
      & -2 & -1 &  0 & 1 & 2 & 3 \\ \hline
3 & & & & \fermi{\Omega^{0,3}} \ar[dr] \ar[dotted,r] & \scalar{\Omega^{0,3}} \ar[r] & \scalar{\Omega^{0,3}} \\
2 & & & & & \fermi{\Omega^{0,2}} \ar[ur, dashed, "\dbar"] & \\
1 & & & & \fermi{\Omega^{0,1}} \ar[ur] \ar[dr] & & \\
0 & & \fermi{\Omega^{3,3}} \ar[r, dotted, bend right = 15] & \begin{matrix} \scalar{\Omega^{0,0}} \\ \scalar{\Omega^{3,3}} \ar[ur, dashed, "\dbar"] \end{matrix} & \begin{matrix} \scalar{\Omega^{0,0}} \\ \scalar{\Omega^{3,3}} \ar[r,dotted, bend left = 15] \end{matrix}  & \fermi{\Omega^{0,0}} & \\
-1 & & & \fermi{\Omega^{3,2}} \ar[ur,dashed, "\dbar"] & & & \\
-2 & & \fermi{\Omega^{3,1}} \ar[ur] \ar[dr] & & & & \\
-3 & \scalar{\Omega^{3,0}} \ar[ur,dashed, "\dbar"] \ar[r] & \Omega^{3,0} \ar[r,dotted] & \fermi{\Omega^{3,0}} & & & 
    \end{tikzcd}
  \end{equation*}
  \caption{The description of the subcomplex $\cA \subset \Tilde{\cT}_{(2,0)}^Q$ using the alternative twisting data.}
\end{figure}
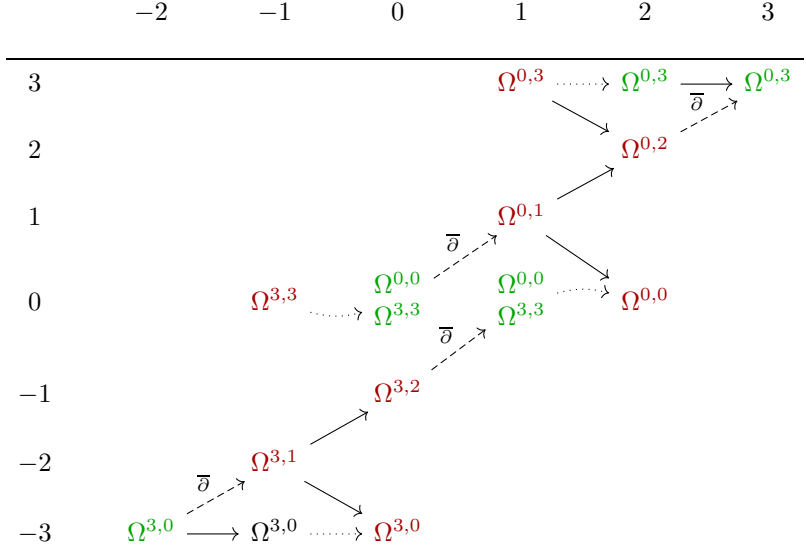

We recognize that the complex of Figure \ref{fig:(2,0)holtwistalt} admits a cochain map to the $\beta\gamma$ system on $\CC^3$.
Since the dotted arrows are isomorphisms, this cochain map is a quasi-isomorphism. 
The following proposition follows from tracing through the presymplectic BV structures, which is completely similar to the previous calculations. 

\begin{prop} \label{prop:alternate}
There is an equivalence of presymplectic BV theories
\[
\Phi : \Tilde{\cT}_{(2,0)}^Q \xto{\simeq} \chi(2) \oplus \beta\gamma(\CC) .
\]
Moreover, this equivalence is $\U(3)$-equivariant. 
\end{prop}

The map $\Phi$ is defined nearly identically to the quasi-isomorphism defined in the previous section for the twisting data $(\phi, \alpha)$. 
The only difference is that one must decompose (and twist) the formula for $\Phi$ acting on the hypermultiplet as in Equation (\ref{eqn:phihyper}). 

\subsection{The twisted factorization algebras}
\label{sec:twfact}

In \S\ref{sec:preobs} we have defined a notion of Hamiltonian observables for certain classes of presymplectic BV theories. 
For a holomorphic supercharge $Q$, each of the twisted presymplectic BV theories $\cT_{(1,0)}^Q$, $\cT^Q_{(2,0)}$ and $\Tilde{\cT}^Q_{(2,0)}$ satisfy Condition (2) in \S \ref{sec:preobs}.
So, in each of these cases we obtain a $\PP_0$-factorization algebra of Hamiltonian observables. 

The twist of the $(1,0)$ theory $\cT_{(1,0)}^Q$ is defined on any complex three-fold $X$.
We denote the corresponding factorization algebra of observables on $X$ by $\Obs_{(1,0)}$, with the supercharge $Q$ understood. 
We can describe this $\PP_0$-factorization algebra explicitly as follows.
Recall $\cT^Q_{(1,0)} \simeq \chi(2)$ which, as a cochain complex, is $\Omega^{\leq 1, \bu} [2]$ equipped with the differential $\dbar + \partial$. 
Keeping track of shifts, one has $\chi(2)^! = \Omega^{\geq 2, \bu}[2]$, again equipped with the differential $\dbar + \partial$. 
Thus, the factorization algebra is described by
\[
\Obs_{(1,0)} = \left(\cO^{sm}(\Omega^{\geq 2, \bu}[1]), \dbar + \partial \right) 
\]
where $\cO^{sm}$ denotes the ``smooth" functionals as defined in \S\ref{sec:preobs}.
Explicitly, to an open set $U \subset X$, the factorization algebra assigns the cochain complex
\[
\Obs_{(1,0)} (U) = \bigg( \Sym \left(\Omega^{\leq 1,\bu}_c (U) [3] \right), \dbar + \partial \bigg) .
\]
With this description in hand, the $\PP_0$-structure is also easy to interpret. 
Given two linear observables $\cO, \cO \in \Omega^{\leq 1,\bu}_c (U) [3]$, the $\PP_0$-bracket is
\beqn\label{p0formula}
\{\cO , \cO' \} = \int_U \cO \partial \cO' .
\eeqn
The bracket extends to non-linear observables by the graded Leibniz rule. 
In \cite{BrianOwenEugene} this $\PP_0$-factorization algebra has appeared as the factorization algebra of boundary observables of abelian $7$-dimensional Chern--Simons theory. 
For more discussion on the relationship to $7$-dimensional Chern--Simons theory and topological M-theory we refer to \S \ref{sec:bcov}. 

We will not explicitly need to mention the factorization algebra associated to the twist of the $(2,0)$ theory $\cT^Q_{(2,0)}$.
However, we will study the factorization algebra associated to its alternative twist $\Tilde{\cT}^Q_{(2,0)}$, which we will denote by $\Obs_{(2,0)}$. 
Again, this theory exists on any complex three-fold $X$.
Similarly to the $(1,0)$ case, we obtain the following explicit description of this factorization algebra.
To an open set $U \subset X$, it assigns the cochain complex
\[
\Obs_{(2,0)} (U) = \bigg( \Sym \left(\Omega^{\leq 1,\bu}_c (U) [3] \oplus \Omega^{3,\bu}_c(U)[3] \oplus \Omega^{0,\bu}_c(U) [1] \right)  , \dbar + \partial \bigg) .
\]
The $\PP_0$-bracket on linear observables is again straightforward. 
The first linear factor is the same as in the $(1,0)$ case.
The second two linear factors are the linear observables of the $\beta\gamma$ system on $\CC^3$. 
For linear observables in $\Omega^{\leq 1, \bu}(U)[3]$ it is given by the same formula as in (\ref{p0formula}). 
The only other nonzero bracket between linear observables occurs between elements $\cO \in \Omega^{3,\bu}_c(U) [3]$ and $\cO' \in \Omega^{0,\bu}_c(U)[1]$ where it is given by
\[
\{\cO, \cO'\} = \int_U \cO \cO'  .
\]

%

\section{The non-minimal twist}
\label{sec:nonmin}

We have classified in \S \ref{sec:susy6} the possible twisting supercharges of the $(2,0)$ supersymmetry algebra.
We found that they were characterized by the rank of the supercharge, which for a non-trivial square-zero element could be either one or two. 
The minimal, rank one, case was studied in the last section.
We now turn to the further, non-minimal, twist of the $(2,0)$ theory. 

Upon applying a twisting homomorphism more natural to the non-minimal twisting supercharge, the non-minimal twist exists on manifolds of the form $M^4 \times \Sigma$ where $M^4$ is a smooth four-manifold and $\Sigma$ is a Riemann surface. 
Since the non-minimal supercharge leaves five directions invariant, this theory depends topologically on $M^4$ and holomorphically on $\Sigma$. 

Our main result is the following; see Theorems~\ref{thm:nm} and~\ref{thm:nmSO4} for more careful statements.
\begin{thm}
The non-minimal twist of the abelian $(2,0)$ tensor multiplet on $\RR^4 \times \CC$ is equivalent, as a presymplectic BV theory, to the theory whose complex of fields is
\[
  \Omega^{\bu}(\RR^4) \mathbin{\Hat{\otimes}} \Omega^{0,\bu}(\CC) [2]  .
\]
The $(-1)$-shifted presymplectic structure is
\beqn\label{presymplecticnm}
(\alpha, \alpha ') \mapsto \int \alpha \partial_{\CC} \alpha' .
\eeqn
Here, $\partial_{\CC}$ denotes the holomorphic de Rham operator on $\CC$. 
\end{thm}

\subsection{The non-minimal deformation}

Before computing the twist, it is instructive to get a handle on the explicit data involved in choosing a non-minimal twisting supercharge. 
As a $\Spin(6)\times\Sp(2)$-module, the odd part of the supertranslation algebra $\fp_{(2,0)}$ is $\Sigma_2 \cong \Pi S_+ \otimes R_2$. 
It is thus easy to compute the stabilizer of a chosen rank-one supercharge, which is the product of the respective stabilizers of fixed vectors in~$S_+$ and~$R_2$ separately. 
This is the subgroup $\MU(3) \times \Sp(1)' \times \U(1) \subset \Spin(6) \times \Sp(2)$.
As representations of the stabilizer, $S_+$  and~$R_2$ decompose as
\deq{
S_+ = \det(L)^{\frac12} \oplus L\otimes \det(L)^{-\frac12}, \qquad
  R_2 = \C^{-1} \oplus (R_1')^{0} \oplus \C^{+1}.
}
Here, the superscripts $\CC^{\pm 1}$ denote the charges under $\U(1)$. 

We can thus consider the following diagram representing the  decomposition of $\Sigma_2$ as a $\MU(3) \times \U(1) \subset \Spin(6) \times \Sp(2)$ representation:
\begin{equation} \label{eqn:s6s2}
  \begin{tikzcd}
    \color{red!65!black}{\det(L)^{\frac12} \otimes \C^{-1}} & \det(L)^{\frac12} \otimes (R_1')^0 & \det(L)^{\frac12} \otimes \C^{+1} \\
    L\otimes\det(L)^{-\frac12}\otimes \C^{-1} & \color{blue!65!black}{L\otimes\det(L)^{-\frac12} \otimes (R_1')^0 } & \color{green!45!black}{L\otimes\det(L)^{-\frac12} \otimes \C^{+1}}
  \end{tikzcd}
\end{equation}
The holomorphic supercharge is indicated in red (note that we have not yet applied any twisting homomorphism). 
Its only nonzero bracket occurs with the supercharges in~$L\otimes\det(L)^{-\frac12}\otimes\C^{+1}$, represented in green above, using the degree-zero pairing on the $R$-symmetry space. As remarked above, this bracket witnesses a nullhomotopy of the translations in~$L$ with respect to the holomorphic supercharge.
The other bracket map of interest to us pairs the supercharges represented in blue with themselves, via the map 
\deq{
  (L\otimes\det(L)^{-\frac12} \otimes (R_1')^0)^{\otimes 2} \to \wedge^2 L \otimes \det(L)^{-1} \otimes \wedge^2 R_1' \cong L^\vee.
}

\begin{rmk}
This equivariant decomposition makes clear the structure of the tangent space to the nilpotence variety at a holomorphic supercharge. The dimension of the normal bundle is $3$, represented by the component colored green above; all other supercharges anticommute with $Q$, and therefore define first-order deformations, which are tangent vectors to the nilpotence variety. The dimension of the tangent space at a holomorphic supercharge is thus $12$, although the projective variety is in fact only 10-dimensional. The fibers of the tangent bundle are ``too large'' because the holomorphic locus is in fact the singular locus of the variety. In fact, as remarked above, 
the singular locus (or space of holomorphic supercharges) is a copy of $\PP^3 \times \PP^3$, consisting of four-by-four matrices of rank one; its tangent space is spanned by the black entries in the diagram~\eqref{eq:tgdiagram}. The red entry is $Q$ itself, representing the tangent direction along the affine cone of the projective variety. 
\end{rmk}

The deformations represented by the blue elements are of interest here; they generate the non-minimal twist (and therefore represent deforming away from the holomorphic locus of the nilpotence variety, into the locus of nonminimal supercharges). However, not all such infinitesimal deformations give rise to finite deformations of~$Q$; geometrically, this corresponds to the fact that the nilpotence variety is singular, and not all vectors in the algebraic tangent space correspond to paths in the variety. Since the nilpotence conditions are quadratic, though, this can be checked at order two: for a deforming supercharge $Q' \in L\otimes\det(L)^{-\frac12} \otimes (R_1')^0$, we just need the condition that 
\[
[Q',Q'] = 0 
\] 
inside $\fp_{(2,0)}$. 
Examining the bracket map discussed above shows immediately that the deforming supercharges with zero self-bracket are precisely the rank-one elements:
\deq{
  Q' = \alpha \otimes w : \quad \alpha \in L\otimes\det(L)^{-\frac12} , ~ w \in (R_1')^0.
}

The data of $Q'$ has an especially nice interpretation through the lens of the holomorphic twist.
We recall the alternative twisting homomorphism $\Tilde{\phi}$ of the $(2,0)$ theory from \S \ref{sec:alternate}. 
Notice that this twisting homomorphism breaks the $\Sp(1)'$ symmetry by fixing a polarization of $R_1'$.
Further, upon twisting by $\Tilde{\phi}$ the relevant component of the spinor representation decomposes under $\MU(3)$ as
\beqn\label{Qbivector}
L\otimes\det(L)^{-\frac12} \otimes (R_1')^0 \rightsquigarrow L \otimes \det (L)^{-1} \oplus L .
\eeqn
Without loss of generality we can assume $Q'$ lies in the first factor. 
Thus, from the perspective of the holomorphic twist, the datum of a further nonminimal twist therefore consists precisely of a polarization of the symplectic vector space $R_1'$, together with a nonzero translation invariant section of $\wedge^2 T^{1,0} \CC^3$, where $T^{1,0} \CC^3$ is the holomorphic tangent bundle. 
We choose holomorphic coordinates $(w_1,w_2, z)$ on $\CC^3$ and identify, without loss of generality, the supercharge $Q'$ with the translation invariant bivector 
\[
Q' = \partial_{w_1} \wedge \partial_{w_2} .
\]

The non-minimal twisting supercharge is of the form
\[
\mathbf{Q} := Q + Q'
\]
where $Q$ is the minimal supercharge lying in the red component of (\ref{eqn:s6s2}) and $Q'$ is a rank one supercharge lying in the blue component of (\ref{eqn:s6s2}). 

Most of the remainder of this section is devoted to the proof of the following description of the non-minimal twist.

\begin{thm} \label{thm:nm}
Using the twisting data $(\Tilde{\phi}, \Tilde{\alpha})$, the $\bQ$-twist of the $(2,0)$ tensor multiplet $\cT^{\bf Q}_{(2,0)}$ is equivalent to the free presymplectic BV theory whose complex of fields is
\beqn\label{eqn:nm}
\begin{tikzcd}
\mathbf{T} \;\; = \;\; \bigg(\Omega^{\leq 1, \bu}(\CC^3) / (\d z) [2] \ar[r, "\Pi \circ \partial"] &  \Omega^{0,\bu}(\CC^3) \bigg) .
\end{tikzcd}
\eeqn
where $\Pi$ is the translation invariant bivector $\partial_{w_1} \wedge \partial_{w_2}$. 
\end{thm}

In the statement of the theorem, the $(-1)$-shifted presymplectic form on $\bT$ is 
\[
(\alpha, \alpha ') \mapsto \int \alpha \partial_{\CC_z} \alpha' .
\]
Here, $\partial_{\CC}$ denote the holomorphic de Rham operator on $\CC$. 

\begin{figure}
  \begin{equation*}
    \begin{tikzcd}[row sep = 3 em, column sep = 3 em]
      -2 & -1 &  0 & 1 & 2 & 3 \\ \hline
       & &  & \Omega^{0,3} \ar[rd] &  \\
       & & \Omega^{0,2} \ar[ru] \ar[rd] & \Omega^{3,3} \ar[r, dotted, "{\<\Pi,\cdot\>}"] & \Omega^{1,3}  \ar[dr, dashed, "\partial_{\CC^2}"]  \\
       & \Omega^{0,1} \ar[rd] \ar[ru] & \Omega^{3,2} \ar[r, dotted, "{\<\Pi,\cdot\>}"] \ar[ru] & \Omega^{1,2}  \ar[dr, dashed, "\partial_{\CC^2}"]  \ar[ru] & & \Omega^{0,3} \\
       \Omega^{0,0} \ar[ru] \ar[rd] & \Omega^{3,1} \ar[r, dotted, "{\<\Pi,\cdot\>}"] \ar[ru] & \Omega^{1,1} \ar[dr, dashed, "\partial_{\CC^2}"] \ar[ru] & & \Omega^{0,2} \ar[ru] \\
       \Omega^{3,0} \ar[ru] \ar[r,dotted, "{\<\Pi,\cdot\>}"]  & \Omega^{1,0} \ar[dr, dashed, "\partial_{\CC^2}"]\ar[ru] & & \Omega^{0,1} \ar[ru] &  & \\
       & & \Omega^0 \ar[ru] & & & \\
    \end{tikzcd}
  \end{equation*}
  \caption{The solid arrows represent the holomorphic twist of the $\cN=(2,0)$ theory.
The dashed and dotted arrows represent the action by the supercharge $Q'$ which deforms the minimal twist to the non-minimal twist.}
  \label{fig:nonmin}
\end{figure}
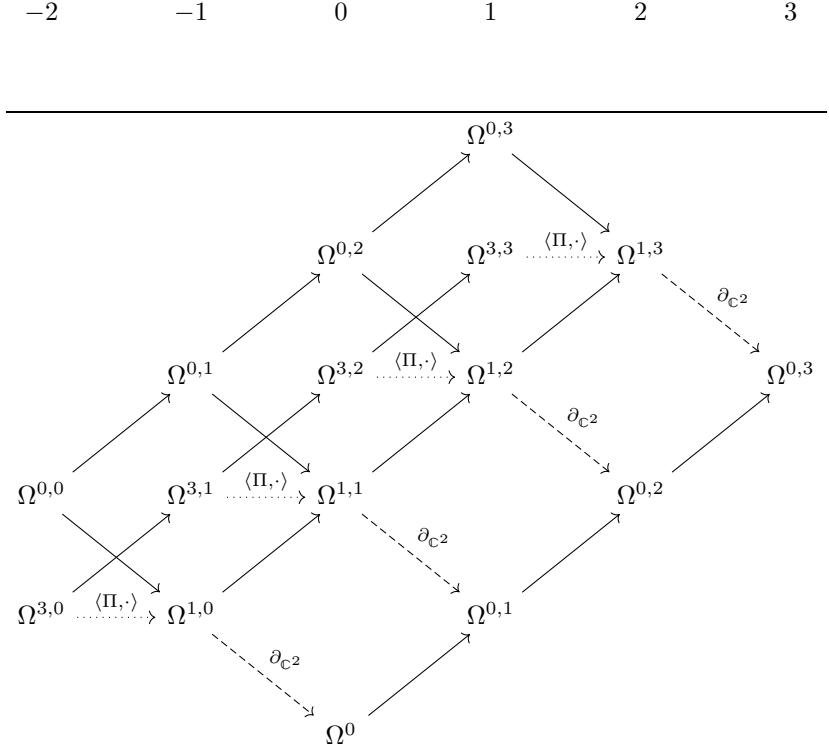

\begin{rmk}
In the description of the non-minimal twist in (\ref{eqn:nm}), we have used the Calabi--Yau form $\d w_1 \d w_2 \d z$ on $\CC^3$ and the fact that $\<\Pi, \d w_1 \d w_2 \d z\> = \d z$. 
On a general, not necessarily Calabi--Yau, three-fold $X$ equipped with a bivector $\Pi \in {\rm PV}^{2,hol}(X)$ we can write this description more invariantly as
\[
\begin{tikzcd}
\Omega^{\leq 1, \bu}(X) / ({\rm Im} (\Pi)) \; [2] \ar[r, "\Pi \circ \partial"] &  \Omega^{0,\bu}(X) 
\end{tikzcd}
\]
Here, ${\rm Im}(\Pi) \subset T^*_X$ is the image of the bundle map $\Pi : \wedge^3 T_X^* \to T^*_X$. 

In~\eqref{Qbivector}, we could have alternatively chosen the deformation to a non-minimal supercharge $Q'$ to be the data of a holomorphic one-form.
On $\CC^3$, these lead to equivalent non-minimal twists.
Globally, however, the data of a non-vanishing holomorphic one-form $\eta \in \Omega^{1,hol}(X)$ leads to the following description of the twist:
\[
\begin{tikzcd}
\Omega^{\leq 1, \bu}(X) / (\eta) \; [2] \ar[r, "\eta \wedge \partial"] &  \Omega^{3,\bu}(X) .
\end{tikzcd}
\]
Here, $\Omega^{\leq 1, \bu}(X) / (\eta)$ denotes the quotient of $\Omega^{\leq 1, \bu}(X)$ by the subspace $\eta \Omega^{0,\bu}(X)$.

We pose the question as to what extent these global descriptions of the non-minimal twist of the $(2,0)$ theory depend on the holomorphic data of a bivector or one-form respectively.  
\end{rmk}

\subsection{Symmetries of the holomorphic twist}

The first step in the proof of Theorem \ref{thm:nm} is to exhibit the action of the non-minimal deformation $Q'$ on the description of the minimal, holomorphic, twist of the $(2,0)$ theory given in the previous section.
We will use the description of Proposition \ref{prop:alternate} of the minimal twist in terms of the theory of the chiral two-form and the $\beta\gamma$ system.
 
At the level of the twisted theory we break the symmetry by the $(2,0)$ super Poincar\'{e} algebra $\sp{2}$ to its $Q$-cohomology.
Upon regrading and applying the twisting data of \S \ref{sec:alternate} to the $Q$-cohomology of $\sp{2}$, this gives us an action of a $\ZZ$-graded Lie algebra $\sp{2}^Q$ on the holomorphic twist of the $(2,0)$ theory. 

Let $\fg \subset \sp{2}$ be the $\ZZ/2$-graded sub Lie algebra consisting of the holomorphic translations plus the portion in blue in Equation (\ref{eqn:s6s2}).
Upon regrading and deforming by $Q$, we obtain a subalgebra $\fg^Q \subset \sp{2}^Q$. 

If $L$ is a complex three-dimensional vector space spanned by the symbols $\partial_z, \partial_{w_1}, \partial_{w_2}$, then as a $\ZZ$-graded vector space we have
\[
\fg^Q = L^* [1] \oplus L \oplus \wedge^2 L [-1] .
\]
whose elements we will denote by $(\eta, v, \pi)$. 
This space has a $\ZZ$-graded Lie algebra structure whose bracket is defined by the $\U(L)$-invariant pairing of $L^*$ with $\wedge^2 L$ as in $[\eta, \pi] = \<\eta, \pi\> \in L$. 

In the notation of \S \ref{sec:holtwist}, the minimal $Q$-twist of the $(2,0)$ theory, using the twisting data $(\Tilde{\phi}, \Tilde{\alpha})$, was denoted $\Tilde{\cT}_{(2,0)}^Q$. 
In Proposition \ref{prop:alternate} we have shown that $\Tilde{\cT}^Q_{(2,0)}$ is equivalent to the presymplectic BV theory $\chi(2) \oplus \beta\gamma(\CC)$ through an explicit quasi-isomorphism
\[
\Phi : \Tilde{\cT}^Q_{(2,0)} \to \chi(2) \oplus \beta\gamma(\CC) .
\]

As an element of $\sp{2}$ that commutes with $Q$, the deformation $Q'$ acts on $\Tilde{\cT}^Q_{(2,0)}$. 
Naively, one could transfer the action of $Q'$ along $\Phi$.
However, since $\Phi$ does not {\em strictly} preserve the presymplectic form, what results is an action of $Q'$ on $\chi(2) \oplus \beta\gamma(\CC)$ that {\em does not} preserve the shifted presymplectic pairing. 

By pulling back along $\Phi$, the $(-1)$-shifted presymplectic form $\omega_{\chi} + \omega_{\beta\gamma}$ on $\chi(2) \oplus \beta\gamma(\CC)$ defines a $(-1)$-shifted presymplectic form on $\Tilde{\cT}^Q_{(2,0)}$. 
Since we know the two shifted presymplectic structures are equivalent, we know abstractly that given any symmetry of the $(2,0)$ theory that is compatible with the holomorphic supercharge and the original presymplectic structure $\omega_{\cT}$, that we can find a homotopy equivalent symmetry that preserves $\Phi^* (\omega_\chi + \omega_{\beta\gamma})$. 

For symmetries arising from the sub Lie algebra $\fg \subset \sp{2}$, we have the following explicit result. 

\begin{lem}\label{lem:hx}
Let $X \in \fg \subset \sp{2}$ and denote by $\rho_Q(X)$ the associated endomorphism of $\Tilde{\cT}^Q_{(2,0)}$ of $\Z$-degree $|X|$. 
There exists a degree-$(|X|-1)$ endomorphism $H_X$ of $\Tilde{\cT}^Q_{(2,0)}$ such that  
\beqn\label{eqn:hx}
\Phi\left(\rho_Q (X) + [Q_{\rm BV} + Q, H_X]\right)
\eeqn
strictly preserves the $(-1)$-presymplectic form $\omega_\chi + \omega_{\beta\gamma}$. 
\end{lem}

We denote by $\Tilde{\rho}_Q(X)$ the endomorphism $\Phi(\rho_Q (X) + [Q_{\rm BV} + Q, H_X])$ of $\chi(2) \oplus \beta\gamma(\CC)$. 

\begin{proof}
If $X$ is a holomorphic translation we can simply take $\Tilde{\rho}_Q(X) = \rho_Q(X)$ and $H_X = 0$. 

Suppose that $X \in \sp{2}$ becomes an element $\pi \in \wedge^2 L \subset \fg^Q$ under the twisting homomorphism $\Tilde{\phi}$. 
For clarity, we will denote the holomorphic decomposition of the bosons and fermions in the $(1,0)$ hypermultiplet by $\nu^{\bu,\bu}$ and $\lambda^{\bu,\bu}$ respectively. 
We denote by $\beta^{\bu,\bu}, \phi^{\bu,\bu}, \psi^{\bu,\bu}_-$ elements of the holomorphically decomposed $(1,0)$ tensor multiplet. 
The homotopy is defined by
\[
H_\pi (\beta^{2,0}) = \<\pi, \beta^{2,0}\> \in \Omega^{0,0}_\nu .
\]

One can compute that the resulting endomorphism $\Phi\left(\rho_Q (\pi) + [Q_{\rm BV} + Q, H_\pi]\right)$ is given by the dotted and dashed arrows of Figure~\ref{fig:nonmin}. 
This is readily seen to preserve the shifted presymplectic structure.

For an element $X \in \fg$ which  becomes an element $\eta \in L^* \subset \fg^Q$ under the twisting homomorphism $\Tilde{\phi}$ the definition of the homotopy is similar. 
\end{proof}

To comment briefly on this calculation, we refer again to the dotted and dashed arrows in Figure~\ref{fig:nonmin}. While most of these originate in standard supersymmetry transformations of the untwisted theory, three are subtle: the leftmost dotted arrow, which carries a physical scalar field to a one-form ghost of the two-form field; the leftmost dashed arrow, which carries a one-form ghost to a physical scalar field via a differential operator; and the dashed arrow third from left, a component of which carries a physical fermi field to a fermi antifield, but is not part of the original BV differential.  

These three mysterious arrows have three different origins. The first is generated by the $L_\infty$ closure term $\rho^{(2)}_\Scalar$. The others, though, are of (untwisted) BV degree $+1$, and so cannot originate from any subtleties of the module structure. The third term is, in fact, generated by the projection map $\Phi$ in~\eqref{eqn:phihyper}, used in computing the twist of the $(1,0)$ hypermultiplet; the second term is generated by the homotopy $H_\pi$ discussed above, which replaces the fermi field $\psi^{1,0}$---whose antifield is \emph{not} eliminated after the holomorphic twist---by its ``new'' antifield $\del \beta^{1,0}$. 

The endomorphisms $\rho_Q(X)$ and $\rho_Q(X) + [Q_{\rm BV} + Q, H_X]$ are homotopy equivalent.
Thus, as a consequence of this proposition, we obtain an equivalent action of $\fg^Q$ on $\Tilde{\cT}^Q_{(2,0)}$ which is presymplectic upon applying the quasi-isomorphism $\Phi$. 
This action is described explicitly in Proposition \ref{prop:modifiedsusy} below, which we take a brief moment to foreground. 
 
As a graded vector space, $\chi(2) \oplus \beta\gamma(\CC)$ decomposes as
\[
(\mathsf{c} , \mathsf{A} , \beta , \gamma) \in \Omega^{0, \bu}[2] \oplus \Omega^{1,\bu} [1] \oplus \Omega^{3, \bu}[2] \oplus \Omega^{0,\bu}  .
\]
The first two components comprise the theory $\chi(2)$ and the second two comprise the $\beta\gamma$ system. 

Recall, there is the internal $\dbar$ differential and also the differential $\mathsf{c} \mapsto \partial \mathsf{c} = \mathsf{A}$ where $\partial$ is the holomorphic de Rham operator on $\CC^3$. 
As in the case of the full supersymmetry algebra, the action of $\fg^Q$ on $\Tilde{\cT}^Q_{(2,0)}$ through $\Tilde{\rho}_Q$ is an action only up to homotopy.
We decompose $\Tilde{\rho}_Q = \Tilde{\rho}^{(1)}_Q + \Tilde{\rho}^{(2)}_Q$ where $\Tilde{\rho}_Q^{(1)}$ is linear and $\Tilde{\rho}^{(2)}_Q$ is quadratic in $\fg^Q$. 
Unpacking the action of supersymmetry given in \S \ref{sec:susy} we obtain the following description of the action of $\fg^Q$ at the level of the holomorphic twist. 

\begin{prop}\label{prop:modifiedsusy}
The action $\Tilde{\rho}_Q = \Tilde{\rho}^{(1)}_Q + \Tilde{\rho}^{(2)}_Q$ of $\fg^Q$ on $\chi(2) \oplus \beta\gamma(\CC)$ is given by:
\begin{itemize}
\item[(1)] the linear term $\Tilde{\rho}^{(1)}_Q$ is defined on holomorphic translations $v \in L$ by $\Tilde{\rho}^{(1)}_Q (v) \alpha = L_v\alpha$, where $\alpha$ is any field.
On the remaining part of the algebra, $\Tilde{\rho}^{(1)}_Q$ is
\[
\begin{array}{ccccccccccccc}
\Tilde{\rho}^{(1)}_Q (\eta) \mathsf{A} & = & \eta \wedge \partial \mathsf{A} & \in & \Omega^{3,\bu}_{\beta} & , & \Tilde{\rho}^{(1)}_Q(\pi) \mathsf{A} & = & \<\pi , \partial \mathsf{A}\> & \in & \Omega^{0,\bu}_\gamma \\ 
\Tilde{\rho}^{(1)}_Q(\eta) \gamma & = & \eta \wedge \gamma & \in & \Omega^{1,\bu}_\mathsf{A} & , & \Tilde{\rho}^{(1)}_Q(\pi) \beta & = & \<\pi , \beta\> & \in & \Omega^{1,\bu}_{\mathsf{A}} .
\end{array} 
\]
Whenever it appears, the symbol $\<\cdot , \cdot\>$ refers to the obvious $\U(L)$-invariant pairing. 

\item[(2)] 
The quadratic term $\Tilde{\rho}^{(2)}$ is given by
\[
\begin{array}{ccccc}
\Tilde{\rho}^{(2)}_Q (\eta \otimes \pi) \mathsf{A} & = & \iota_{\<\eta, \pi\>} \mathsf{A} & \in & \Omega^{0,\bu}_{\mathsf{c}}.
\end{array}
\]
\end{itemize}
\end{prop}
Conceptually, as remarked above, the key step in proving this proposition is to observe that the homotopy $H_\pi$ described in Lemma~\ref{lem:hx} generates the transformation $\Tilde{\rho}^{(1)}(\pi)  \mathsf{A} = \<\pi, \partial \mathsf{A}\>$ via~\eqref{eqn:hx}.

\begin{rmk}
It is instructive to verify directly that the action described in the above proposition is an $L_\infty$-action. 
To see this, the key relation to observe is 
\[
\eta \cdot (\pi \cdot \mathsf{A}) - \pi \cdot (\eta \cdot \mathsf{A}) =  \iota_{\<\eta, \pi\>} \partial \mathsf{A} .
\]
for any $\eta, \pi, \mathsf{A}$ where the $\cdot$ denotes the linear action $\Tilde{\rho}_Q^{(1)}$. 
\end{rmk}

We can now give a proof of the main result of this section. 
\begin{proof}[Proof of Theorem \ref{thm:nm}]
The shifted presymplectic action of $X \in \fg^Q$ at the level of the holomorphic twist $\chi(2) \oplus \beta\gamma(\CC)$ is given by $\Tilde{\rho}_Q (X)$. 

The non-minimal deformation $Q'$ determines a nontrivial element in $\wedge^2 L [-1] \subset \fg^Q$ that we identify with $\partial_{w_1} \wedge \partial_{w_2}$.
Schematically, the action $\Tilde{\rho}_{Q}(Q')$ is given by the dotted and dashed arrows in Figure \ref{fig:nonmin}.
Let $\Tilde{Q}_{\rm BV}$ denote the solid arrows in this figure, which describes the linear BV differential of the holomorphic twist. 

We observe that each of the dotted arrows is of the form
\[
\<\Pi, -\> : \Omega^{3, \bu}(\CC^3) \to \Omega^{1, \bu} (\CC^3) .
\]
If we decompose $\Omega^{1,\bu}(\CC^3)$ as $\Omega^{1,\bu}(\CC^2_w) \mathbin{\Hat{\otimes}} \Omega^{0,\bu}(\CC_z) \oplus \Omega^{0,\bu}(\CC^2) \mathbin{\Hat{\otimes}} \Omega^{1,\bu}(\CC_z)$ 
we see that this map is an isomorphism onto the second component. 

So, we see that there is a projection from the total complex $\left(\Tilde{\cT}^Q_{(2,0)}, \Tilde{Q}_{\rm BV} + Q' \right)$ to $\bT$ whose kernel is acyclic. 
\end{proof}

\subsection{A refined twisting homomorphism}

We proceed to describe the twisting data that is somewhat more natural to the non-minimal twist. 

Consider the twisting homomorphism is defined by the composition
\[
\phi_{\rm top} : \U(2) \times \U(1) \to \U(3) \times \U(3) \xto{\det^{\frac12} \times \det^{\frac12}}  \U(1) \times \U(1) \xto{(i , (i')^{-1})} \Sp(1) \times \Sp(1)' \hookrightarrow \Sp(2)  .
\]
The first map is the block diagonal embedding of $(A, x) \in \U(2) \times \U(1)$ into $\U(3)$ via $\begin{pmatrix} A & 0 \\ 0 & x \end{pmatrix}$ in the first factor and via $\begin{pmatrix} A & 0 \\ 0 & x^{-1} \end{pmatrix}$ into the second factor. 
Also, $i : \U(1) \to \Sp(1)$ is the unique homomorphism for which $Q$ has weight $+1$ and $i' : \U(1) \to \Sp(1)$ is defined by the polarization determined by $Q'$. 

Additionally, we have the regrading homomorphism 
\[
\alpha_{\rm top} : \U(1) \xto{\rm diag} \U(1) \times \U(1) \xto{i \times i'} \Sp(1) \times \Sp(1)' \subset \Sp(2) .
\]
Note that this is identical to the regrading homomorphism $\Tilde{\alpha}$ of \S \ref{sec:alternate}). 

It is a direct calculation to verify that $\phi_{\rm top}$, $\alpha_{\rm top}$ constitute twisting data for the the nonminimal twisting supercharge ${\bf Q} = Q+Q'$. 

The affect of $\phi_{\rm top}$ is to ``twist" the graded vector space underlying $\bT$ of Theorem \ref{thm:nm} to
\beqn\label{nmalt}
\phi^* \bT = \left(\Omega^{0,\bu}(\CC^2) [2] \oplus \Omega^{1,\bu}(\CC^2)[1] \oplus \Omega^{3,\bu}(\CC^2)\right) \otimes \Omega^{0,\bu}(\CC_z) .
\eeqn
The differential then becomes $\dbar_{\CC^2} + \partial_{\CC^2} + \dbar_{\CC_z}$. 

Moreover, $\phi_{\rm top}$ satisfies the following property. 
First, note that there is a natural embedding $j : \U(2) \hookrightarrow \SO(4)$. 
Unwinding the definition of $\phi_{\rm top}$ above, one finds that it can be factored as
\[
\begin{tikzcd} 
\U(2) \times \U(1) \ar[dr, "j \times 1"'] \ar[rr, "\phi_{\rm top}"] & & {\rm Sp}(2) \\
& \SO(4) \times \U(1) \ar[ur, "\phi_{\rm SO}"'] &
\end{tikzcd}
\]
where $\phi_{\rm SO}$ is the following composition of maps:
\[
  \SO(4) \times \U(1) \xto{\cong} \SU(2) \times \SU(2) \times \U(1) \xto{\pr_{1,3}} \Spin(3) \times \Spin(2) \to \Spin(5) \xto{\cong} \Sp(2) .
\]

The pair $(\phi_{\rm SO}, \alpha)$ constitutes yet another set of twisting data for the non-minimal supercharge $\bQ = Q + Q'$. (To see that the embedding along a short root of $\lie{sp}(2)$ is the correct one, recall that $\bQ$ is defined by the wedge  pairing of a Lagrangian subspace in $R_2$ with a two-dimensional subspace in~$S_+$; the stabilizer of a Lagrangian subspace in $\lie{sp}(n)$ is $\lie{gl}(n)$, embedded along the short roots via the diagram inclusion $A_{n-1} \hookrightarrow C_n$.
The twisting homomorphism $\phi_{\rm SO}$ allows one to identify 
(\ref{nmalt}) with $\Omega^\bu (\RR^4) \otimes \Omega^{0,\bu}(\CC_z)[2]$. 
As a consequence, we have the following. 

\begin{prop} \label{thm:nmSO4}
Using the twisting data $(\phi_{\rm SO}, \alpha_{\rm SO})$, the $\bQ$-twist of the $(2,0)$ tensor multiplet $\cT^{\bf Q}_{(2,0)}$ is equivalent to the free presymplectic BV theory whose BV complex of fields is
\[
  \Omega^{\bu} (\RR^4) \mathbin{\Hat{\otimes}} \Omega^{0,\bu}(\CC) [2] 
\]
and whose $(-1)$-shifted presymplectic structure is given in (\ref{presymplecticnm}). 
The equivalence is $\SO(4) \times \U(1)$-equivariant. 
\end{prop}

This description of the nonminimal twist makes sense on a manifold of the form $M \times \Sigma$ where $M$ is a smooth four-manifold and $\Sigma$ is a Riemann surface. 

\section{Comparison to Kodaira--Spencer gravity}\label{sec:bcov}

In this section we document a relationship between the twist of the tensor multiplet and a holomorphic theory defined on Calabi--Yau manifolds that has roots in string theory and theories of supergravity. 
This theory, which we will refer to as Kodaira--Spencer theory, is gravitational in the sense that it describes variations of the Calabi--Yau structure; it was first introduced in~\cite{BCOV} as the closed string field theory describing the $B$-twisted topological string on three-folds. 
Work of Costello--Li \cite{CLBCOV1, CLsugra, CLtypeI} has began to systematically exhibit the relationship of Kodaira--Spencer theory on more general manifolds to twists of other classes of string theories and theories of supergravity. 

The main result of this section (Proposition~\ref{prop:bcov}) can be stated heuristically as follows: up to topological degrees of freedom, {\em the theory of the field strengths of the holomorphic twist of the abelian $(2,0)$ tensor multiplet on a Calabi--Yau three-fold is equivalent to the free limit of (minimal) Kodaira--Spencer theory}. 
There is a similar statement for the $(1,0)$ tensor multiplet and a Type I Kodaira--Spencer theory.
This resolves a simple form of a conjecture in Costello--Li in \cite{CLBCOV1}.
This can also be seen as an enhancement of a result of Mari{\~n}o, Minasian, Moore, and Strominger \cite{MMMS}, where it is shown that the equations of motion of the M5 brane theory on a Calabi--Yau three-fold include the Kodaira--Spencer equations of motion.

We consider Kodaira--Spencer theory on a Calabi--Yau three-fold $X$, and we denote by $\Omega$ the nowhere vanishing holomorphic volume form.
Denote by ${\rm PV}^{i,j}(X) = \Gamma(X, \wedge^i T_X \otimes \wedge^j \Bar{T})$ the $j$th term in the Dolbeault resolution of polyvector fields of type $i$. 
The fields of Kodaira--Spencer theory are
\[
\cT_{\rm KS} \overset{\rm def}{=} {\rm PV}^{\bu,\bu}(X)[[t]] [2] .
\]
Here $t$ denotes a formal parameter of degree $+2$. 
The gradings are such that the degree of the component $t^k {\rm PV}^{i,j}$ is $i+j+2k-2$.
The complex of fields carries the differential
\[
Q_{\rm KS} = \dbar + t \partial_{\Omega}
\]
where $\partial_\Omega$ fits into the diagram
\[
\begin{tikzcd}
{\rm PV}^{i,j}(X) \ar[r, "\partial_\Omega"] \ar[d, "\Omega"', "\cong"] & {\rm PV}^{i-1, j}(X) \ar[d, "\Omega", "\cong"'] \\
\Omega^{3-i, j}(X) \ar[r, "\partial"] & \Omega^{4-i,\bu}(X) .
\end{tikzcd} 
\]
for $i \geq 1$. 
Note that $\partial_\Omega$ is an operator of degree $-1$ on $\cT_{\rm KS}$, so that $\dbar + t \partial_\Omega$ is an operator of homogenous degree $+1$. 
The fields of Kodaira--Spencer theory are not the sections of a finite rank vector bundle, but we will pick out certain subspaces of fields which are the sections of a finite rank bundle. 

Kodaira--Spencer theory fits into the BV formalism as a (degenerate) Poisson BV theory \cite{CLBCOV1}.
For a precise definition of a Poisson BV theory see \cite{ButsonYoo}.  
The degree $+1$ Poisson bivector $\Pi_{\rm KS}$ on $\cT_{\rm KS}$ which endows $\cT_{\rm KS}$ with a Poisson BV structure is defined by
\[
  \Pi_{\rm KS}  = (\partial \otimes 1) \delta_{\rm \Delta} \in \Bar{\cT}_{\rm KS}(X) \mathbin{\Hat{\otimes}} \Bar{\cT}_{\rm KS}(X) .
\]
Here, $\delta_{\Delta}$ is the Dirac delta-function on the diagonal in $X \times X$. 

Any Poisson BV theory defines a $\PP_0$-factorization algebra of observables \cite{ButsonYoo}.
For the free limit of Kodaira--Spencer theory this $\PP_0$-factorization algebra is completely explicit. 
To an open set $U \subset X$ one assigns the cochain complex:
\begin{align*}
\Obs_{\rm KS} (U) & = \bigg( \cO^{sm} \left(\cT_{{\rm KS}} (U)\right), Q_{\rm KS} \bigg) \\ & = \bigg(\Sym\left(\cT_{{\rm KS},c}^! (U)\right) \; ,  \; Q_{\rm KS} \bigg)  .
\end{align*}
The BV bracket is defined via contraction with $\Pi_{\rm KS}$. 
We denote the resulting $\PP_0$-factorization algebra for Kodaira--Spencer theory by ${\rm Obs}_{\rm KS}$.

There are variations of the theory obtained by looking at certain subcomplexes of $\cT_{\rm KS}$ and by restricting the $\PP_0$-bivector.
They are called: {\em minimal} Kodaira--Spencer theory, denoted by $\Tilde{\cT}_{\rm KS}$; {\em Type I} Kodaira--Spencer theory, denoted $\cT_{\rm Type \; I}$; and {\em minimal Type I} Kodaira--Spencer theory, denoted $\Tilde{\cT}_{\rm Type I}$. 
They fit into the following diagram of embeddings of complexes of fields:
\[
\begin{tikzcd}
& \Tilde{\cT}_{\rm KS} \ar[rd] & \\
\Tilde{\cT}_{\rm Type \; I} \ar[ur] \ar[dr] & & \cT_{\rm KS} \\
& \cT_{\rm Type \; I}\ar[ur] & 
\end{tikzcd}
\]
The corresponding $\PP_0$ factorization algebras of classical observables will be denoted $\Tilde{\Obs}_{\rm KS}$, $\Obs_{\rm Type\; I}$, and $\Tilde{\Obs}_{\rm Type \; I}$ (whose definitions we recall below). 

The goal of this section is relate Kodaira--Spencer theory to the twists of the $(1,0)$ and $(2,0)$ superconformal theories using factorization algebras.
Recall that in \S \ref{sec:obs} we showed that {\em \psBV{}} theories, such as the chiral $2k$-form $\chi(2k)$, admit a $\PP_0$-factorization algebra consisting of the ``Hamitlonian" observables. 
We have provided a detailed description of the factorization algebras associated to the holomorphic twists of the $(1,0)$ and $(2,0)$ theories in \S \ref{sec:twfact}.
The main result is the following.

\begin{prop}\label{prop:bcov}
Let $X$ be a Calabi--Yau three-fold and $Q$ be a holomorphic supercharge. 
The following statements are true regarding the holomorphic twists $\cT^Q_{(2,0)}$ and $\cT^Q_{(1,0)}$ of the $\cN = (2,0)$ and $\cN = (1,0)$ tensor multiplets, respectively:
\begin{itemize}
\item[(1)] There is a sequence of morphisms of complexes of fields:
\beqn\label{bcovdiagram1}
\begin{tikzcd}
& \beta\gamma (\CC) \oplus \Omega^{\geq 2, \bu}[1] \ar["f", dr] & \\
\cT^Q_{(2,0)} \ar[ur, "F"] & & \Tilde{\cT}_{\rm KS} .
\end{tikzcd}  
\eeqn
which induces a morphism of $\PP_0$-factorization algebras on $X$:
\[
\Tilde{\Obs}_{\rm KS} \to \Obs_{(2,0)}
\]
whose fiber is a locally constant factorization algebra.
\item[(2)] There is a sequence of morphisms of complexes of fields:
\beqn\label{bcovdiagram2}
\begin{tikzcd}
& \Omega^{\geq 2, \bu}[1] \ar[dr, "f", "\cong"'] & \\
\cT_{(1,0)}^Q \ar[ur, "F"] & & \Tilde{\cT}_{{\rm Type \; I}} .
\end{tikzcd}  
\eeqn
which induces a quasi-isomorphism of $\PP_0$-factorization algebras on $X$:
\[
\Tilde{\Obs}_{\rm Type \; I} \xto{\simeq} \Obs_{(1,0)}
\]
\end{itemize}
\end{prop}

These result may be summarized as follows.
For the $(1,0)$ theory, one finds that the factorization algebra of Hamiltonian observables of the \psBV{} theory $\cT^Q_{(1,0)}$ is equivalent to the free limit of the observables of Type I Kodaira--Spencer theory. 
For the $(2,0)$ theory, the observables of the \psBV{} theory $\cT^Q_{(2,0)}$ differ from the free limit of the observables of minimal Kodaira--Spencer theory by a locally constant factorization algebra. 
This locally constant part has been appeared in \cite{SuryaYoo} in what they refer to as Kodaira--Spencer theory {\em with potentials}. 

The connection between Kodaira--Spencer theory and the tensor multiplet is through the field strength. 
Indeed, the map labeled $F$ in the above statement is a holomorphic version of the field strength of the chiral two-form. 
In the sections below, it is defined as the obvious extension of the following map of Dolbeault complexes
\[
\partial : \Omega^{\leq 1, \bu} [2] \to \Omega^{\geq 2, \bu} [1] 
\]
given by the holomorphic de Rham operator. 
By our results in the previous sections, given a two-form element $\chi \in \Omega^{\leq 1, \bu}(X)$, the component of the three-form field strength $\partial \chi \in \Omega^{\geq 2, \bu}$ is the only piece that survives the holomorphic twist.

Finally, we remark that the $\PP_0$-factorization algebra $\Tilde{\Obs}_{\rm Type \; I} \simeq \Obs_{(1,0)}$ has appeared as the factorization algebra of boundary observables of $7$-dimensional abelian Chern--Simons theory. 
Likewise, minimal Kodaira--Spencer theory $\Tilde{\cT}_{\rm KS}$ also appears as a boundary condition of $7$-dimensional abelian Chern--Simons theory.

\subsection{Minimal theory}
\def\PV{{\rm PV}}

Many of the fields in the complex $\cT_{\rm KS}$ are invisible to the shifted Poisson structure we have just introduced. 
There is a piece of $\cT_{\rm KS}$ that ``sees" the Poisson bracket, called the minimal theory.
The fields of the minimal theory form the subcomplex of fields of full Kodaira--Spencer theory $\Tilde{\cT}_{\rm BCOV} \subset {\rm PV}^{\bu,\bu}(X)[[t]] [2]$ defined by
\[
\Tilde{\cT}_{\rm KS} \overset{\rm def}{=} \bigoplus_{i+k \leq 2} t^k {\rm PV}^{i,\bu} [-i-2k+2] .
\]
The shifted Poisson tensor $\Pi_{\rm KS}$ restricts to one on this subcomplex, thus defining another Poisson BV theory whose fields are $\Tilde{\cT}_{\rm KS}$. 

\subsubsection{Proof of part (1) of Proposition \ref{prop:bcov}}

This is a direct calculation. 
Observe that the minimal fields decompose into six graded summands:
\[
\Tilde{\cT}_{\rm KS} = \PV^{0,\bu} [2] \oplus \PV^{1,\bu} [1] \oplus t \PV^{0,\bu} \oplus \PV^{2,\bu} \oplus t \PV^{1,\bu} [-1] \oplus t^2 \PV^{0,\bu} [-2] .
\]
and the differential takes the form:
\beqn\label{bcovdiagram3}
\begin{tikzcd}
\ul{-2} & \ul{-1} & \ul{0} & \ul{1} & \ul{2} \\
\PV^{0,\bu} & & & & \\
& \PV^{1,\bu} \ar[r, "t \partial_\Omega"] & t \PV^{0,\bu} & & \\
& & \PV^{2,\bu}  \ar[r, "t \partial_\Omega"] & t \PV^{1,\bu}  \ar[r, "t \partial_\Omega"] & t^2 \PV^{0,\bu} 
\end{tikzcd}
\eeqn
Using the Calabi--Yau form $\Omega$ we can identify each line above with some complex of differential forms. 
For the first line, we have $\PV^{0,\bu} \overset{\Omega}{\cong} \Omega^{3,\bu}$.
The second line is isomorphic to the cochain complex $\Omega^{\geq 2 , \bu}[1]$, where $\partial_\Omega$ is identified with the holomorphic de Rham operator. 
This is the standard resolution of closed two-forms up to a shift. 
Similarly, the third line is isomorphic to $\Omega^{\geq 1, \bu}$. 
This is the standard resolution for closed one-forms. 

In total, the cochain complex of minimal Kodaira--Spencer theory $\cT_{\rm KS}$ is isomorphic to 
\[
\Omega^{3, \bu}[2] \oplus \Omega^{\geq 2, \bu} [1] \oplus \Omega^{\geq 1, \bu} .
\]

We define the morphism $f$ in the first diagram (\ref{bcovdiagram1}) of Proposition \ref{prop:bcov}. 
Recall, the cochain complex of fields of the $\beta\gamma$ system with values in $\CC$ is
\[
\Omega^{0,\bu} \oplus \Omega^{3,\bu}[2] .
\]
On the components $\Omega^{3,\bu}[2]$ and $\Omega^{\geq 2, \bu}[1]$, we take $f$ to be the identity morphism. 
On the component $\Omega^{0,\bu}$ we take $f$ to be the holomorphic de Rham operator
\[
\partial : \Omega^{0,\bu} \to \Omega^{\geq 1,\bu} .
\]

Using the description of the holomorphic twist in \S \ref{sec:alternate}, we have identified the minimal twist of the $\cN=(2,0)$ theory with $\cT^Q_{(2,0)} \cong \chi(2) \oplus \beta\gamma(\CC)$. 
The morphism $F$ is defined to be the identity on the $\beta\gamma(\CC)$ component. 
On $\chi(2) = \Omega^{\leq 1, \bu}[2]$, $F$ is defined by the holomorphic de Rham operator
\[
\partial : \Omega^{\leq 1, \bu}[2] \to \Omega^{\geq 2, \bu} [1] .
\]

To finish the proof, we introduce an intermediate factorization algebra that we think of as the observables associated to the Poisson BV theory $\beta\gamma(\CC) \oplus \Omega^{\geq 2, \bu}[1]$. 
Let $\cF$ be the factorization algebra which assigns to $U \subset X$ the cochain complex
\[
\cF(U) = \bigg(\Sym\left(\beta\gamma^!_c(U) \oplus \Omega^{\leq 1, \bu}_c (U) [3] \right) , \dbar_{\beta \gamma} + \dbar + \partial \bigg) .
\]

The maps $f,g$ induce maps of factorization algebras
\[
\Obs_{\rm KS} \xto{f^*} \cF \xto{F^*} \Obs_{(2,0)}
\]
Following the description of $\Obs_{(2,0)}$ given in \S \ref{sec:twfact}, we observe that the map $F^*$ is a quasi--isomorphism. 
The result follows from the fact that the kernel of $f$ is the sheaf of constant functions $\ul{\CC}$.

\subsection{Type I theory}

Type I Kodaira--Spencer theory has underlying complex of fields
\[
\cT_{\rm Type \; I} = \bigoplus_{i + k = \; {\rm odd}} t^k \PV^{i,\bu}[-i-2k-2] .
\]
This describes the conjectural spacetime string field theory of the Type I topological string, see \cite{CLtypeI}.

The complex of fields of minimal Type I Kodaira--Spencer theory $\Tilde{\cT}_{\rm Type \; I}$ is the intersection of the fields of the minimal theory with the Type I theory.
The only polyvector fields that appear are of arity zero and one, so that:
\[
\Tilde{\cT}_{\rm Type \; I} =  \PV^{1,\bu} [1]  \oplus t \PV^{0,\bu} 
\]
As before, the differential is the internal $\dbar$ operator plus the operator $t \partial_\Omega$ which maps the first component to the second. 
Notice that $\Tilde{\cT}_{\rm Type \; I} \subset \Tilde{\cT}_{\rm KS}$ as the middle line in diagram (\ref{bcovdiagram3}). 

The proof of part (2) of Proposition \ref{prop:bcov} is more direct than the last section. 
We have already explained how to identify $\Tilde{\cT}_{\rm Type\; I}$ with the resolution of closed two-forms $\Omega^{\geq 2, \bu} [1]$. 
This is the isomorphism $f$ in diagram (\ref{bcovdiagram2}). 

The morphism $F$ in diagram (\ref{bcovdiagram2}) is the holomorphic de Rham operator $\partial$.
The same argument as in the last section shows that $F \circ f$ defines the desired quasi-isomorphism
\[
(F \circ f)^* :
\Tilde{\Obs}_{\rm Type \; I} \xto{\simeq} \Obs_{(1,0)} .
\]

\section{Dimensional reduction} \label{sec:dimred}

In this final section, we place the $\N=(2,0)$ theory on various geometries, including both naive dimensional reduction and compactification on product manifolds. We begin with the twisted theories, showing that the holomorphic twist reduces to the twist of five-dimensional supersymmetric Yang--Mills theory up to a copy of the constant sheaf. We then go on to give a description of the holomorphic twist after compactification along a complex surface, as well as the two-dimensional theory obtained by placing the nonminimal twist on a smooth four-manifold. Finally, we consider dimensional reduction to five dimensions at the level of the untwisted theory, and show that this produces untwisted five-dimensional Yang--Mills as expected, up to the same copy of the constant sheaf. The calculation leads us into some considerations related to electric--magnetic duality, through which we argue that the presence of the constant sheaf is reasonable and in fact expected from a physics perspective. In the final portion, we offer some more speculative thoughts on how the zero modes can be correctly handled by passing to a nonperturbative description of the theory. 

\subsection{Compactifications of the twisted theories}

\subsubsection{Reduction to twisted four-dimensional Yang--Mills}

In this section, let $E$ be an elliptic curve and $Y$ a complex surface.
We consider the holomorphic twists of the $(1,0)$ and $(2,0)$ theories on the complex three-fold $Y \times E$.
Recall, in \S \ref{sec:twfact} we have defined the factorization algebra of classical observables of the holomorphic $(1,0)$ theory $\Obs_{(1,0)}$ and of the holomorphic twist (using the alternative twisting homomorphism of \S \ref{sec:alternate}) of the $(2,0)$ theory $\Obs_{(2,0)}$. 
We look at the dimensional reduction of these factorization algebras along the elliptic curve $E$, meaning we consider the pushforward along the projection map $Y \times E \to Y$. 

Upon reduction along $E$, we find a relationship of the factorization algebras $\Obs_{(1,0)}$ and $\Obs_{(2,0)}$ to the factorization algebras associated to the holomorphic twists of pure 4d $\cN=2$ and $\cN = 4$ Yang--Mills theory for the abelian one-dimensional Lie algebra.

Following \cite{CostelloHolomorphic, ESW}, we recall the description of the holomorphic twist of supersymmetric Yang--Mills in four-dimensions.
Each of these holomorphic twists exists on any complex surface $Y$. 

For the case of $4d$ $\cN=2$, the holomorphic twist is described by the underlying complex of fields
\beqn\label{4dn2}
\Omega^{0,\bu}(Y) [\ep] [1] \oplus \Omega^{2,\bu}(Y) [\ep] [1]
\eeqn
where $\ep$ is a formal parameter of degree $+1$. 
This theory is free and is equipped with the linear BRST operator given by the $\dbar$-operator.
The degree $(-1)$ pairing on the space of fields is given by the integration pairing along $Y$ together with the Berezin integral in the odd $\ep$ direction.
That is, given fields $A + \ep A'$ and $B + \ep B'$ where $A, A' \in \Omega^{0,\bu}(Y)$, $B,B'\in \Omega^{2,\bu}(Y)$, the pairing is
\[
(A + \ep A' , B + \ep B') \mapsto \int_Y (A B' + A' B) .
\]
Since the pure supersymmetric Yang--Mills theory for an abelian Lie algebra is a free theory, we consider the ``smooth" version $\cO^{sm}$ of the classical observables, just as in \S\ref{sec:fact}. 
We denote the associated factorization algebra of classical observables by $\Obs_{4d \; \cN = 2}$.
Via the degree $(-1)$ pairing this factorization algebra is equipped with a $\PP_0$-structure.

The description of the holomorphic twist of $4d$ $\cN=4$ supersymmetric Yang--Mills theory for abelian Lie algebra is similar. 
The underlying complex of fields is
\beqn\label{4dn4}
\Omega^{0,\bu}(Y) [\ep, \delta] [1] \oplus \Omega^{2,\bu}(Y) [\ep,\delta] [2]
\eeqn
The degree $(-1)$ pairing is given by the integration pairing along $Y$ together with the Berezin integral in the odd $\ep,\delta$ directions. 

\begin{prop}
Let $\pi : Y \times E \to Y$ be the projection. 
\begin{itemize}
\item 
There is a morphism of $\PP_0$-factorization algebras on $Y$
\[
\pi_*\Obs_{(1,0)} \to \Obs_{4d \; \cN=2} 
\]
whose cofiber is a locally constant factorization algebra with trivial $\PP_0$-structure.
\item
There is a morphism of $\PP_0$-factorization algebras on $Y$
\[
\pi_*\Obs_{(2,0)} \to \Obs_{4d \; \cN=4}
\]
whose cofiber is a locally constant factorization algebra with trivial $\PP_0$-structure.
\end{itemize}
\end{prop}

\begin{proof}
We consider the $(1,0)$ case first.
Following the description in \S \ref{sec:twfact}, the factorization algebra $\Obs_{(1,0)}$ is given by the ``smooth" functionals on the sheaf of cochain complexes $\Omega^{\geq 2, \bu}[1]$ on $Y \times E$.
Since $E$ is formal, there is a quasi-isomorphism $\CC[\ep] \xto{\simeq} \Omega^{0,\bu}(E)$.
Here, $\ep$ is a chosen generator for the sheaf of sections of the anti-holomorphic canonical bundle on $E$. 
Thus, there is a quasi-isomorphism of sheaves on $Y$:
\[
\Omega^{2, \bu}(Y) [\ep] \oplus \d z \Omega^{\geq 1, \bu} (Y) [\ep] \xto{\simeq} \pi_* \Omega^{\geq 2 , \bu} .
\]
Here, $\d z$ denotes the holomorphic volume form on the elliptic curve.

The sheaf of cochain complexes $\Omega^{\geq 1, \bu}(Y)$
is a resolution for the sheaf of closed one-forms on the complex surface $Y$. 
The $\partial$-operator determines a map of cochain complexes $\partial : \Omega^{0,\bu}(Y) \to \Omega^{\geq 1, \bu}(Y)$ whose kernel is the sheaf of constant functions. 

Putting this together, we find that there is a map of sheaves of cochain complexes on $Y$:
\[
\Omega^{0,\bu}(Y) [\ep] [1] \oplus \Omega^{2,\bu}(Y) [\ep] [1] \xto{\partial} \pi_* \Omega^{\geq 2, \bu} [1].
\]
We recognize the left-hand side as the complex of fields underlying the holomorphic twist of $4d$ $\cN=2$. 
Applying the functor of taking the ``smooth" functionals $\cO^{sm}(-)$ we obtain the first statement of the proposition. 
It is immediate to verify that this map intertwines the $\PP_0$-structures.

The second statement is not much more difficult. 
Recall, the complex of fields of the holomorphic twist of the $(2,0)$ theory is obtained by adjoining the $\beta\gamma$ system on the three-fold $Y \times E$ to the holomorphic twist of the $(1,0)$ theory. 
As a sheaf on $Y \times E$, the complex of fields of the $\beta\gamma$ system is 
\[
\Omega^{0,\bu}(Y \times E) \oplus \Omega^{3,\bu} (Y \times E) [2] .
\]
Pushing forward along $\pi$ the complex becomes
\[
\Omega^{0,\bu} (Y) [\ep] \oplus \d z \Omega^{2,\bu} (Y)[\ep] [2]  .
\]
Notice that this is a symplectic BV theory with the wedge and integrate pairing.
The statement then follows from the observation that there is an isomorphism of symplectic BV theories
\[
\bigg(\Omega^{0,\bu}(Y) [\ep] [1] \oplus \Omega^{2,\bu}(Y) [\ep] [1] \bigg) \oplus \bigg(\Omega^{0,\bu} (Y) [\ep] \oplus \d z \Omega^{2,\bu} (Y)[\ep] [2]\bigg)
\]
with the holomorphic twist of $4d$ $\cN=4$ as in (\ref{4dn4}).
\end{proof}

\subsubsection{Reduction to 2d CFT}

Consider the higher dimensional $\beta\gamma$ system $\beta\gamma(X , \CC)$ on a three-fold $X$ with values in the trivial bundle.
The space of fields is $\Omega^{0,\bu}(X) \oplus \Omega^{3,\bu}(X)[2]$.
If $Y$ is a compact K\"{a}hler surface, $\Sigma$ a Riemann surface, and $\pi : Y \times \Sigma \to \Sigma$ the projection, then there is an equivalence of free BV theories on $\Sigma$:
\beqn\label{reducebg}
\pi_* \beta\gamma(X ; \CC) = \beta\gamma\left(\Sigma ; H^\bu(Y, \cO_Y)\right) .
\eeqn
This is the $\beta\gamma$ system on $\Sigma$ with values in the (graded) vector space $H^\bu(Y,\cO_Y)$. 

Let $\chi(2)$ be the theory of the chiral two-form on $Y \times \Sigma$ with values in $\CC$. 
The following result follows from a direct calculation of the sheaf-theoretic pushforward of the complex $\chi(2)$ 
along the map $\pi : Y \times \Sigma \to \Sigma$. 
In the result, we use the fact that on a compact surface $Y$ there is a symmetric bilinear form on the graded vector space $H^{\bu}(Y, \Omega^{1}_Y)[1] $ provided by Serre duality. 

\begin{lem}
Let $Y$ be a compact K\"{a}hler surface and $\Sigma$ a Riemann surface.
There is an equivalence of presymplectic BV theories on $\Sigma$
\beqn\label{chireduce}
\pi_* \chi(2)(Y \times \Sigma)  \; \simeq  \; \left(\Omega^{\bu}(\Sigma) 
\otimes H^\bu(Y , \cO_Y) \right) [2] \oplus \thy\left(0 , H^\bu(Y, \Omega^1_Y)[1]\right)(\Sigma).
\eeqn
On the right-hand side, the $(-1)$-shifted presymplectic form is trivial on the first summand and the second summand is the chiral boson on $\Sigma$ with values in the graded vector space $H^\bu(Y , \Omega^1_Y)[1]$. 
\end{lem}

Combining this lemma with (\ref{reducebg}) we obtain the following description of the reduction of the holomorphic twist of the $(2,0)$ theory along $\pi : Y \times \Sigma \to \Sigma$. 

\begin{prop}
Suppose $Y$ is a compact K\"{a}hler surface. 
The compactification of the holomorphic twist of the abelian $(2,0)$ theory along $Y$ is equivalent to the direct sum of the following three presymplectic BV theories on $\Sigma$:
\beqn\label{reduce(2,0)}
\beta\gamma(H^\bu(Y, \cO_Y)) \oplus \thy\left(0 , H^\bu(Y, \Omega^1_Y)[1]\right) \oplus  \left(\Omega^{\bu}(\Sigma) 
\otimes H^\bu(Y , \cO_Y) \right) [2] .
\eeqn
\end{prop}

The first summand in (\ref{reduce(2,0)}) is the $\beta\gamma$ chiral CFT with values in $H^\bu(Y, \cO_Y)$ and the second summand is the chiral boson with values in $H^\bu(Y, \Omega^1_Y)$. 
The last summand is quasi-isomorphic to the constant sheaf in degree $-2$, and so has only topological degrees of freedom.
Moreover, the last summand is equipped with the trivial shifted presymplectic structure. 

Next, consider the non-minimal twist of the $(2,0)$ theory which we will place on $M \times \Sigma$, where $M$ is a closed four-manifold. 
The presymplectic BV complex is of the form $\Omega^{\bu}(M) \otimes \Omega^{0,\bu}(\Sigma)[2]$. 
Similarly as in the holomorphic twist, the compactification along $M$ produces the theory of the chiral boson.
This time, however, it is valued in the graded vector space $H^\bu(M , \CC) [2]$, the de Rham cohomology of $M$ shifted by two. 
Note that the integration pairing endows this graded space with a graded symmetric bilinear form. 

\begin{prop}
Let $M$ be a closed four-manifold. 
The compactification of the non-minimal twist of the abelian $(2,0)$ theory along $M$ is equivalent, as a presymplectic BV theory, to the chiral boson $\chi(0 , H^\bu(M , \CC)[2])$ on $\Sigma$.
\end{prop}

\subsection{Untwisted dimensional reduction}

We close this paper by giving some results on dimensional reduction at the level of the untwisted theory. It is expected that dimensional reduction along a circle should produce five-dimensional $\N=2$ super Yang--Mills theory, but with an inverse dependence of the 5d coupling constant on the compactification radius, compared to the dependence expected from a typical Kaluza--Klein reduction. As in the results for the twisted theory above, we will be able to show that our formulation reduces correctly, up to a copy of the constant sheaf. We check that the presymplectic structure agrees with the standard BV pairing after dimensional reduction. Accounting correctly for this copy of the constant sheaf should require passing to a nonperturbative description of the theory; we offer some speculative comments on this, and plan to return to a more rigorous analysis in future work.

In five-dimensional $\N=2$ supersymmetry,  the $R$-symmetry group is $\Sp(2) \cong \SO(5)$, just as in  six dimensions. The chiral spinor repreduces to  the (unique) five-dimensional spinor representation; dimensional reduction of the fermions and the  scalars  in~$\cT_{(2,0)}$ is thus trivial. The only difficulty is thus to check that our description of~$\thy_+(2)$ reduces to the (nondegenerate) BV theory of a one-form gauge field in five dimensions. We check this in the theorem below; the full statement for supersymmetric theories follows immediately (Corollary~\ref{cor:untwisted}).

\begin{prop} \label{prop:tensorred}
  Let $\thy_+(2)^\text{red}$ denote the dimensional reduction of~$\thy_+(2)$ to five dimensions along $S^1$. There is a map of theories
  \deq{
    \Xi: \thy_+(2)^\text{red} \to \Scalar(1),
  }
  where $\Scalar(1)$ is the standard nondegenerate theory of perturbative abelian one-form fields on~$\R^5$. The kernel of this map is a copy of the constant sheaf in BV degree $-2$. Furthermore, the shifted presymplectic structure $\Xi^*\omega_\Scalar$ agrees with the shifted presymplectic structure inherited from~$\thy_+(2)$.
\end{prop}
\begin{proof}
We begin by placing the theory on~$\R^5 \times S^1$. For all of the fields other than the self-dual antifield, the decomposition is straightforward, using the fact that 
\deq{
  \Omega^i(\R^6) \cong \left( \Omega^i(\R^5) \otimes \Omega^0(S^1) \right)  \oplus \left( \Omega^{i-1}(\R^5) \otimes \Omega^1(S^1) \right) .
}
The self-dual forms $\Omega^3_+$, however, sit diagonally with respect to this decomposition. If $p: \R^5 \times S^1 \to \R^5$ denotes the obvious projection and $\alpha \in \Omega^3(\R^5)$ is a fixed two-form, then 
  \deq{ \pi_+ p^* \alpha = \frac{1}{2} \left( p^* \alpha -  \sqrt{-1} \, dt \wedge p^* (\star  \alpha) \right). 
  }
  Note that $\star$ here denotes the \emph{five}-dimensional Hodge star operator; $\star\alpha$ is thus an element of $\Omega^2(\R^5)$. $\theta$ is a coordinate on~$S^1$.

  In  light of this decomposition, we can  rewrite $\thy_+(2)$ on $\R^5 \times  S^1$ in the  following fashion:
  \begin{equation}
    \begin{tikzcd}
      \Omega^0(\R^5) \otimes \Omega^0(S^1)[2] \arrow{r}{\d_{\R^5}} \arrow{rd}{\d_{S^1}} &        \Omega^1(\R^5) \otimes \Omega^0(S^1)[1] \arrow{r}{\d_{\R^5}} \arrow{rd}{\d_{S^1}} &  \Omega^2(\R^5) \otimes \Omega^0(S^1)  \arrow{rd}{\d_{S^1} + (\star\d)_{\R^5}} \\
      &       \Omega^0(\R^5) \otimes \Omega^1(S^1)[1] \arrow{r}{\d_{\R^5}}  &        \Omega^1(\R^5) \otimes \Omega^1(S^1) \arrow{r}{\d_{\R^5}}  &  \Omega^2(\R^5) \otimes \Omega^1(S^1)[-1]  .
    \end{tikzcd}
  \end{equation}
  (We have made use of the fact that the spaces $\Omega^2(\R^5)\otimes\Omega^1(S^1)$ and~$\Omega^3(\R^5)\otimes\Omega^0(S^1)$ are isomorphic via the six-dimensional Hodge star.) 

  To dimensionally reduce along the circle, we pass to the cohomology of $\d_{S^1}$. This results in the following cochain complex $\thy_+(2)^\text{red}$ of vector bundles on~$\R^5$: 
\begin{equation}
    \begin{tikzcd}
      \Omega^0(\R^5)[2] \arrow{r}{\d}  &        \Omega^1(\R^5)[1] \arrow{r}{\d} &  \Omega^2(\R^5) \arrow{rd}{\star\d} \\
      &       \Omega^0(\R^5)[1] \arrow{r}{\d}  &        \Omega^1(\R^5) \arrow{r}{\d}  &  \Omega^2(\R^5) [-1].
    \end{tikzcd}
  \end{equation}
  Of  course, the shifted presymplectic structure on~$\thy_+(2)^\text{red}$ is inherited from the structure on~$\thy_+(2)$. It is easy to verify that, as  a  skew map from $\thy_+(2)^\text{red}$ to~$(\thy_+(2)^\text{red})^![-1]$, the pairing $\omega^\text{red}$ is just the pair of first-order differential operators represented by dashed arrows in the below diagram:
  \begin{equation}
    \label{eq:red-presymp}
  \begin{tikzcd}
    \Omega^0(\R^5)[2] \arrow{r}{}  &        \Omega^1(\R^5)[1] \arrow{r}{} &  \Omega^2(\R^5) \arrow{rd}{\star\d}  \\ 
      &       \Omega^0(\R^5)[1] \arrow{r}{}  &        \Omega^1(\R^5) \arrow{r}{}  &  \Omega^2(\R^5) [-1] \arrow[dashed]{d}{\d}\\
      & & & \Omega^3(\R^5)[-1] \arrow{r}{} & \Omega^4(\R^5)[-2] \arrow{r}{} & \Omega^5(\R^5)[-3] \\
      & & \arrow[dashed, <-, bend left = 50, crossing over]{uuu}{\d} \Omega^3(\R^5) \arrow{r}{} \arrow{ru}{\d \star} & \Omega^4(\R^5)[-1] \arrow{r}{} & \Omega^5(\R^5)[-2]. \\
    \end{tikzcd}
  \end{equation}
  There is a single check that $\omega^\text{red}$ defines a cochain map, which is trivial.

  Now, there is a strict isomorphism of cochain complexes of vector bundles, defined to be the identity map in all nonpositive degrees and the Hodge star map in degree 1. This isomorphism allows us to replace $\thy_+(2)^\text{red}$ by the cochain complex of vector bundles in the diagram below, which we will call $\cU$ for brevity:
\begin{equation}
    \begin{tikzcd}
      \Omega^0(\R^5)[2] \arrow{r}{\d}  &        \Omega^1(\R^5)[1] \arrow{r}{\d} &  \Omega^2(\R^5) \arrow{r}{\d} & \Omega^3(\R^5)[-1] \\
      &       \Omega^0(\R^5)[1] \arrow{r}{\d}  &        \Omega^1(\R^5) \arrow{ru}{\star\d}.
    \end{tikzcd}
\end{equation}
The shifted presymplectic structure on~$\cU$ is analogous to (in fact, just isomorphic to) that depicted in~\eqref{eq:red-presymp}. We are dealing with the complex 
\deq{
  \cU = (\Omega^{\leq 1}(\R^5)[1] \oplus \Omega^{\leq 3}(\R^5)[2], \star\d), 
}
and it is immediate that
\deq{
  \cU^![-1] \cong ( \Omega^{\geq 4} (\R^5) [-1] \oplus \Omega^{\geq 2}(\R^5), \d\star ).
}
The map $\omega_\cU$ consists of the two obvious de~Rham operators that define maps of degree zero from~$\cU$ to~$\cU^![-1]$. 

It remains to define the map of theories from the theorem statement. Recall that five-dimensional abelian (perturbative) Yang--Mills theory is described by the BV complex~$\Scalar(1) = \Scalar(1,\C)$ from~\S\ref{sec:pre}, which takes the form
\deq{ \Scalar(1) = (\Omega^{\leq 1}(\R^5)[1] \oplus \Omega^{\geq 4}(\R^5)[-1], \d\star\d).
}
There is then a map of cochain complexes defined by the vertical arrows in the following diagram:
\begin{equation}
  \Xi: \cU \to \Scalar(1), \quad 
  \begin{tikzcd} \Omega^{\leq 1}(\R^5)[1] \arrow{r}{\star\d } \arrow{d}{\id} & \Omega^{\leq 3}(\R^5)[2] \arrow{d}{\d} \\
    \Omega^{\leq 1}(\R^5)[1]\arrow{r}{\d\star\d} & \Omega^{\geq 4}(\R^5)[-1].
  \end{tikzcd}
\end{equation}
We now form the cone of this map. Using~\cite[Proposition~1.23]{ESW}, we can eliminate the acyclic piece, obtaining the description
\begin{equation}
  \Cone(\Xi) \cong \left( \begin{tikzcd} \Omega^{\leq 3}(\R^5)[3] \arrow{r}{\d} & \Omega^{\geq 4}(\R^5)[-1] \end{tikzcd} \right) .
\end{equation}
This complex has cohomology only in the left-hand term; after totalizing, we obtain $\Omega^\bu(\R^5)[3]$, thus finding that the kernel is a copy of the constant sheaf in degree $-2$. 
It is entirely straightforward to check that $\Xi^*\omega_\Scalar = \omega_\cU$. Since $\cU\cong\thy_+(2)^\text{red}$, we have completed the proof.
\end{proof}

  As remarked above, the full statement, pertaining to dimensional reduction of the full $\N=(2,0)$ abelian tensor multiplet, follows immediately from Proposition~\ref{prop:tensorred}:
\begin{cor} \label{cor:untwisted}
  Let $\cT^\text{red}_{(2,0)}$ denote the dimensional reduction of~$\cT_{(2,0)}$ to five dimensions along $S^1$, and let $\text{YM}_{\N=2}$ denote the $\N=2$ supersymmetric (perturbative) abelian Yang--Mills multiplet on~$\R^5$. 
  There is a map of theories 
  \deq{
    \hat{\Xi}: \cT^\text{red}_{(2,0)} \to \text{YM}_{\N=2},
  }
  defined by extending $\Xi$ by identity morphisms. The kernel of this map is a copy of the constant sheaf in BV degree $-2$. Furthermore, the shifted presymplectic structure $\Xi^*\omega_\text{YM}$ agrees with the shifted presymplectic structure inherited from~$\cT_{(2,0)}$.
\end{cor}

\subsubsection{Electric--magnetic duality and the physical interpretation of the proof of Proposition~\ref{prop:tensorred}}
  For the physicist reader, the language of the proof of Proposition~\ref{prop:tensorred} may be unfamiliar, but the manipulations should at least have a familiar flavor. We briefly recall the typical description of electric--magnetic duality that is  folk wisdom among physicists: A theory of $p$-form gauge fields in dimension $d$ has a gauge potential $A \in \Omega^p(\R^d)$, whose field strength is a gauge-invariant $(p+1)$-form given (in the abelian case) just by $F = dA$. $F$ satisfies an equation of motion $\d{\star F} = 0$, but also a ``Bianchi identity'' $\d  F = 0$, which is (at least in contractible open sets) equivalent to the existence of the potential $A$. These equations could be just as well phrased in terms of the ``dual'' field strength, the $(d-p-1)$-form $G = \star F$, with the roles of the equations of motion and the Bianchi identity reversed. In light of the Bianchi identity, $G$ can be written as the field strength of a potential $B \in \Omega^{d - p - 2}(\R^d)$. One can sum this up by saying that an equivalence is expected between the theories of $p$-forms and $(d-p-2)$-forms; the two different descriptions are sometimes referred to as the ``electric'' and ``magnetic'' sides of the duality, since the Hodge star operation in standard Maxwell theory reverses the components of $F$ that are the physical electric and magnetic fields. (In Maxwell theory in four dimensions, both the electric and magnetic gauge fields are one-forms.)

  One might therefore expect an equivalence between the theories we have called $\Scalar(p)$ and~$\Scalar(d-p-2)$. However, it is not possible to write down a quasi-isomorphism relating the two; an attempt to follow the logic of the above argument always produces results which disagree by certain shifted copies of the constant sheaf. Moreover, this is not surprising from the physical perspective: From the point of view of the electric theory, the magnetic degrees of freedom---which are the 't~Hooft operators---are nonperturbative, and are most naturally thought of as disorder operators. In the absence of any consideration of the global structure of the gauge group, or the related issue of Dirac quantization, one cannot expect to capture these degrees of freedom correctly. 

  One formulation of the physics argument above, in the BV formalism and at a perturbative level, goes as follows; see~\cite{ElliottAbelian} for a rigorous treatment of nonperturbative issues in electric--magnetic duality from a BV perspective. We begin with the BV theory 
  \begin{equation}
    \Scalar(p) = \begin{tikzcd} \Omega^{\leq p}(\R^d)[p] \arrow{r}{\d\star\d} & \Omega^{\geq (d-p)}(\R^d)[-1], \end{tikzcd}
    \end{equation}
    thinking of it as the electric description. 
    There is another cochain complex $\cF(p)$ of vector bundles on~$\R^d$, defined by
    \begin{equation}
      \cF(p) = \begin{tikzcd} \Omega^{\geq(p+1)}(\R^d) \arrow{r}{\d\star} & \Omega^{\geq(d-p)}(\R^d)[-1], \end{tikzcd}
    \end{equation}
    which can be thought of as a BV or on-shell version of the field strength, subjected to its equation of motion. ($\cF(p)$ freely resolves the sheaf of solutions to the equations $\d F = \d {\star F} = 0$.) 
    There is a curvature map $\curv: \Scalar(p) \to \cF(p)$, defined by the vertical arrows in the diagram
    \begin{equation}
      \begin{tikzcd} \Omega^{\leq p}(\R^d)[p] \arrow{d}{\d} \arrow{r}{\d\star\d} & \Omega^{\geq (d-p)}(\R^d)[-1] \arrow{d}{\id}\\
   \Omega^{\geq(p+1)}(\R^d) \arrow{r}{\d\star} & \Omega^{\geq(d-p)}(\R^d)[-1].
 \end{tikzcd}
 \end{equation}
 It extends the usual curvature map on fields by the identity on antifields.    
 The cone of $\curv$ is a shift of the de Rham complex, as in the proof above, and so there is a kernel, consisting of the constant sheaf representing zero modes of the zero-form ghost in BV degree $-p$.

 Now, applying the Hodge star in degree zero defines an isomorphism of $\cF(p)$ with $\cF(d-p-2)$. There is thus a sequence of maps of the form 
 \begin{equation}
   \begin{tikzcd}
     \Scalar(p) \arrow{r}{\curv} & \cF(p) \arrow{r}{\cong} & \cF(d-p-2) & \Scalar(d-p-2) \arrow{l}[swap]{\curv},
   \end{tikzcd}
 \end{equation}
 encapsulating a BV description of the argument above. If electric--magnetic duality were to hold at a perturbative level, all of these maps would be quasi--isomorphisms; the curvature map, however, is not, and the duality fails at the level of the constant sheaves described above. It is interesting to note that, in the description we are giving, the antifields to the electric degrees of freedom in some  sense play the role of the magnetic degrees of freedom.
 Furthermore, we remark that $\cF(p)$ does \emph{not} admit a natural shifted presymplectic structure; it does, however, admit a shifted Poisson tensor.

 In the proof of Proposition~\ref{prop:tensorred}, a very analogous set of arguments play a  role. However, the object $\cU$ that appears there is \emph{not} the curvature; in fact, it maps \emph{into} both $\Scalar(1)$ and $\Scalar(2)$ on~$\R^5$ in symmetric fashion, defining a roof of maps between them, rather than receiving maps from each. To phrase the situation in general language, we would define 
\begin{equation}
  \label{eq:defU}
  \cU(p) = \begin{tikzcd} \Omega^{\leq p}(\R^d)[p] \arrow{r}{\star \d} & \Omega^{\leq (d-p-1)}(\R^d)[d-p-2]. \end{tikzcd}
\end{equation}
It is immediate to  see that $\cU(p)$ and~$\cU(d-p-2)$ are isomorphic via the Hodge star in BV degree $+1$, and that the map $\Xi$ can be generalized to a map $\Xi(p): \cU(p) \to \Scalar(p)$. $\cU(p)$, moreover, \emph{does} admit a natural shifted presymplectic structure, as in~\eqref{eq:red-presymp}.

The proof of Proposition~\ref{prop:tensorred} relies on electric--magnetic duality, in the sense that $\thy_+(2)^\text{red} = \cU(2)(\R^5)$ is, at first glance, most naturally interpreted as a theory of a two-form. The additional copy of a constant sheaf in the theorem is also, in some sense, dual to the issue that appeared in our attempt to perturbatively formalize the standard argument. We can sum up all of these considerations, in somewhat greater generality, with the following diagram:
\begin{equation}
  \label{eq:EMpert}
  \begin{tikzcd}[column sep = 4 em]
    \cU(p) \arrow[leftrightarrow]{r}{\cong} \arrow{d}{\Xi} & \cU(d-p-2) \arrow{d}{\Xi} \\
    \Scalar(p) \arrow{d}{\curv} \arrow[dash, dashed]{r}{\text{EM dual}} & \Scalar(d-p-2) \arrow{d}{\curv} \\
    \cF(p) \arrow[leftrightarrow]{r}{\cong} & \cF(d-p-2).
  \end{tikzcd}
\end{equation}
The kernel of each vertical map in~\eqref{eq:EMpert} is an appropriately shifted copy of the constant sheaf. The failure of these vertical maps to define quasi-isomorphisms reflects the nonperturbative nature of the duality; we offer some speculation on the correct fix for this in the next section. 

\subsubsection{Speculative remarks on global structure}

There is no doubt that the reader will have been disappointed by all of the ``errors'' in the above results, having to do with zero modes (or, for mathematicians, undesirable copies of constant sheaves). Part of the reason for the discursiveness of the above remarks on electric--magnetic duality is to emphasize that we see these as representing familiar  phenomena from the physics perspective: Any on-the-nose equivalence of perturbative theories cannot possibly be a correct representation of electric--magnetic duality. The fact that electric--magnetic duality plays a role in passing from the $\N=(2,0)$ multiplet to supersymmetric Yang--Mills theory in five dimensions is also not unreasonable; in fact, this is the key reason that the dependence on the coupling constant is inverted from the standard Kaluza--Klein expectation, as remarked above. For interacting theories, electric--magnetic duality requires an inversion of the coupling constant. (The coupling constant that scales ``correctly'' with the compactification radius is not the Yang--Mills coupling constant, but the coupling constant of its magnetically dual theory of two-forms.)

In fact, it is tempting to speculate that insisting on the correct dimensional reduction at the nonperturbative level will shine a light on the nonperturbative BV formulation of electric--magnetic duality. 
Recall that the correct nonperturbative generalization of the BRST complex of an abelian gauge field---which is perturbatively just $\Omega^{\leq p}(\R^d)[p]$---is the smooth Deligne cohomology group
\begin{equation}
  \Z_\alpha(p)_D^\infty = 
  \begin{tikzcd}[column sep = 8 ex] \underline{\Z}[p+1] \arrow{r}{(2\pi i)^p \alpha} & \Omega^{\leq p}[p].
  \end{tikzcd}
\end{equation}
(Here $\alpha$ denotes a choice of real number, which plays the role of the coupling constant or radius of the gauge group; our notation here differs from the standard notation for Deligne cohomology by indicating $\alpha$ explicitly.) 
We should thus expect that it is possible to formulate a BV, or possibly presymplectic BV, description of abelian Yang--Mills theory, using Deligne cohomology groups to represent both the electric and magnetic gauge fields. In light of the considerations above, and by directly generalizing~\eqref{eq:defU}, one would attempt to write down a complex of the form
\begin{equation}
  \Z_\alpha \cU(p) = \begin{tikzcd} \Z_\alpha(p)_D^\infty \arrow{r}{\star\d} & \Z_{1/\alpha}(d-p-1)_D^\infty [-1].
  \end{tikzcd}
  \label{eq:YMDeligne}
  \end{equation}
  The inverse choices of coupling constants are necessitated by the requirement that the complex have nontrivial cohomology in degree zero. In particular, Deligne cohomology represents the curvature of a connection in a $\U(1)$ (or $\GL(1)$) bundle, and so admits a curvature map whose image is an \emph{integral} class (for $\alpha = 1$). We can rewrite the complex above in the form
  \begin{equation}
    \begin{tikzcd}
      \Z_\alpha(p)_D^\infty \arrow{dr}{\star\d} & \\
      & \Omega^{d-p-1}[-1] . \\
      \Z_{1/\alpha}(p)_D^\infty \arrow{ur}{\d}
    \end{tikzcd}
    \label{eq:EMeqs}
  \end{equation}
  By passing to the cohomology of the internal differentials of the Deligne complexes, we see that the curvatures of the electric and magnetic degrees of freedom must be related by the Hodge star. Choosing a volume form so that the Hodge star preserves integral classes makes the requirement on the coupling constants immediate, at least up to discrete choices corresponding to finite coverings of~$\U(1)$ by~$\U(1)$.  

  Describing things in this way makes our considerations seem almost trivial; of course abelian Yang--Mills theory consists of an  electric gauge field $A$ with curvature $F$, and a magnetic gauge field $B$ with curvature $G$, subject to the constraint that $F = \star G$. We emphasize that the novelty in this way of thinking consists of the fact that this pair is interpreted \emph{as a complete (presymplectic) BV theory}, where the pairing is defined by differential operators as done for $\omega_\cU$ above. In this formulation, the equations of motion (and therefore the antifields) for $F$ have been replaced by the Bianchi identities (and therefore the gauge invariances) for~$G$. This is the sense in which the magnetic gauge fields and the electric antifields are one and the same.
Electric--magnetic duality then just amounts to the trivial or manifest  statement that
\deq{
  \Z_\alpha \cU(p) \cong \Z_{1/\alpha} \cU(d-p-2).
}
  It would be interesting to make contact with other BV approaches to electric--magnetic duality, such as~\cite{ElliottAbelian}.

  Identical considerations suggest a nonperturbative definition of the theory of self-dual $(2k)$-forms;\footnote{We thank K.~Costello for suggesting this definition to us, independently of dimensional reduction.}
the reader will probably have guessed that the complex we have in mind is 
\begin{equation}
  \Z_\alpha \thy_+(2k) = \begin{tikzcd} \Z_\alpha(2k)_D^\infty \arrow{r}{\d_+} & \Omega^{2k+1}_+[-1]
  \end{tikzcd}
   = 
   \begin{tikzcd}[column sep = 8 ex] \underline{\Z}[2k+1] \arrow{r}{(2\pi i)^{2k} \alpha} & \thy_+(2k).
   \end{tikzcd}
\end{equation}
We note that, for $k=0$, this theory describes periodic (circle-valued) chiral bosons; in general, it describes a connection on an abelian gerbe, subject to the constraint that the curvature (which now defines an integral class) must be self-dual. 

Now, placing $\Z_\alpha \thy_+(2)$ on $\R^5 \times S^1$ and pushing forward along the projection map produces \emph{precisely} the complex $\Z_\alpha \cU(2)$, under the assumption that the radius of the compactification circle is one. 
To see this, note that we must use the derived pushforward, so that $\pi_* \underline{\Z} = \underline{H^*(S^1,\Z)}$. Making sense of the maps reduces to understanding the map induced on the sheaf cohomology of the circle from the map of sheaves
\deq{
  \underline{\Z} \xto{(2\pi i)^p\alpha} \underline{\C}.
}
Using Poincar\'{e} duality for $S^1$, one can argue that the induced map on $H^0$ is given by multiplication by $(2\pi i)^p \alpha$, while the induced map on $H^1$ is given by multiplication by $(2\pi i)^p \frac{{\rm vol}(S^1)}{\alpha}$. 

We expect the considerations of~\cite{WittenM5,HopkinsSinger} to play a role in our analysis at this stage; a careful formulation  of the arguments we have sketched here should make clear how the data of a quadratic refinement of the intersection form plays a role in our analysis. Such a datum  is required to make sense of the partition function of the chiral field; describing the classical theory, however, does not seem to explicitly require such a choice, at least a perturbative level. We feel that it would be interesting to further develop the formalism presented here, to make a more rigorous analysis of nonperturbative issues, and to address issues related to quantization. It will be exciting to move farther down this path in future work. 
\printbibliography

\end{document}